\documentclass[conference]{IEEEtran}
\IEEEoverridecommandlockouts

\def\BibTeX{{\rm B\kern-.05em{\sc i\kern-.025em b}\kern-.08em
    T\kern-.1667em\lower.7ex\hbox{E}\kern-.125emX}}
\usepackage[margin=1in]{geometry}
\usepackage{times}
\usepackage{graphicx} %
\usepackage{subcaption} 
\usepackage{caption}
\usepackage{tabularx}
\usepackage{lipsum}

\usepackage{hyperref}
\usepackage{url}
\usepackage{graphicx}
\usepackage{amsfonts}
\usepackage{amsmath}
\usepackage{amsthm}
\usepackage{mathtools}
\usepackage{enumerate}%

\usepackage{algorithm}
\usepackage{algpseudocode}

\newcommand{\Lap}[1]{\text{Lap}\left(#1\right)}

\usepackage[usenames, dvipsnames]{color}
\usepackage{tikz}
\usetikzlibrary{arrows}

\newcommand{\showinlater}[1]{}

\newcommand{\notshow}[1]{}
\def\calS{\mathcal{S}}

\newtheorem{theorem}{Theorem}[section]
\newtheorem{definition}[theorem]{Definition}
\newtheorem{lemma}[theorem]{Lemma}

\renewcommand{\Pr}[2]{\mathbb{P}_{#1}\left[ #2 \right]}

\usepackage{bbm}
\newcommand{\ind}[1]{\mathbbm{1}\left(#1\right)}

\newcommand{\nodeprivacy}{%
\ifnum\spacefactor=3000
\expandafter\MakeUppercase\fi%
node differential privacy}

\def\calG{\mathcal{G}}
\def\calA{\mathcal{A}}

\renewcommand{\deg}{\text{deg}}
\newcommand{\degree}[2]{\deg_{#1}\left(#2\right)}
\newcommand{\outdegree}[2]{\text{out-}\deg_{#1}\left(#2\right)}
\newcommand{\indegree}[2]{\text{in-}\deg_{#1}\left(#2\right)}

\newcommand{\edge}[2]{\text{edge}_{#1}\left(#2\right)}
\newcommand{\outedge}[2]{\text{out-edge}_{#1}\left(#2\right)}
\newcommand{\inedge}[2]{\text{in-edge}_{#1}\left(#2\right)}
\newcommand{\numedge}[2]{\text{numedge}_{#1}\left(#2\right)}

\newcommand{\Din}{D_{\text{in}}}
\newcommand{\Dout}{D_{\text{out}}}
\def\tD{\tilde{D}}
\def\tDin{\tilde{D}_{\text{in}}}
\def\tDout{\tilde{D}_{\text{out}}}

\newcommand{\highdegree}[2]{\text{highDeg}_{#1}\left(#2\right)}
\newcommand{\highoutdegree}[2]{\text{highOutDeg}_{#1}\left(#2\right)}

\newcommand{\hist}[1]{\text{hist}\left(#1\right)}
\newcommand{\histout}[1]{\text{histOut}\left(#1\right)}

\newcommand{\GS}[2]{\text{GS}_{#1}\left(#2\right)}

\renewcommand{\S}[2]{S_{#1}\left(#2\right)}
\newcommand{\f}[1]{f\left(#1\right)}

\newcommand{\rom}[1]{\uppercase\expandafter{\romannumeral #1\relax}}

\DeclarePairedDelimiter\abs{\lvert}{\rvert}%
\makeatletter
\let\oldabs\abs
\def\abs{\@ifstar{\oldabs}{\oldabs*}}

\newcommand{\seq}[3]{\left(#1\right)_{#2}^{#3}}
\newcommand{\paren}[1]{\left(#1\right)}

\def\time{\texttt{time}}

\def\argmin{\text{argmin}}
\def\argmax{\text{argmax}}
\newcommand\shapescale{7}
\newcommand\starscale{4}
\newcommand\shapenodesize{1.5pt}
\tikzset{vertex/.style = {shape=circle,fill=black,inner sep=0pt,draw,minimum size=\shapenodesize}}
\tikzset{dedge/.style = {->,> = latex'}}
\tikzset{edge/.style = {-,> = latex'}}

\newcommand{\udtriangle}{
\begin{tikzpicture}
\node[vertex] (a) at  (0,0) {};
\node[vertex] (b) at  (-1/\shapescale,-1.73/\shapescale) {};
\node[vertex] (c) at  (1/\shapescale, -1.73/\shapescale) {};
\draw[edge] (a) to (b);
\draw[edge] (b) to (c);
\draw[edge] (a) to (c);
\end{tikzpicture}
}
\newcommand{\dtriangleI}{
\begin{tikzpicture}
\node[vertex] (a) at  (0,0) {};
\node[vertex] (b) at  (-1/\shapescale,-1.73/\shapescale) {};
\node[vertex] (c) at  (1/\shapescale, -1.73/\shapescale) {};
\draw[dedge] (a) to (b);
\draw[dedge] (b) to (c);
\draw[dedge] (c) to (a);
\end{tikzpicture}
}
\newcommand{\dtriangleII}{
\begin{tikzpicture}
\node[vertex] (a) at  (0,0) {};
\node[vertex] (b) at  (-1/\shapescale,-1.73/\shapescale) {};
\node[vertex] (c) at  (1/\shapescale, -1.73/\shapescale) {};
\draw[dedge] (a) to (b);
\draw[dedge] (b) to (c);
\draw[dedge] (a) to (c);
\end{tikzpicture}
}

\newcommand{\dstarI}{
\begin{tikzpicture}
\node[vertex] (a) at  (0,0) {};
\node[vertex] (b) at  (0,  1/\starscale) {};
\node[vertex] (c) at  (0, -1/\starscale) {};
\node[vertex] (d) at  (-1/\starscale, 0) {};
\node[vertex] (e) at  (1/\starscale, 0) {};
\draw[dedge] (a) to (b);
\draw[dedge] (a) to (c);
\draw[dedge] (a) to (d);
\draw[dedge] (a) to (e);
\end{tikzpicture}
}
\newcommand{\dstarII}{
\begin{tikzpicture}
\node[vertex] (a) at  (0,0) {};
\node[vertex] (b) at  (0,  1/\starscale) {};
\node[vertex] (c) at  (0, -1/\starscale) {};
\node[vertex] (d) at  (-1/\starscale, 0) {};
\node[vertex] (e) at  (1/\starscale, 0) {};
\draw[dedge] (b) to (a);
\draw[dedge] (c) to (a);
\draw[dedge] (d) to (a);
\draw[dedge] (e) to (a);
\end{tikzpicture}
}
\newcommand{\udstar}{
\begin{tikzpicture}
\node[vertex] (a) at  (0,0) {};
\node[vertex] (b) at  (0,  1/\starscale) {};
\node[vertex] (c) at  (0, -1/\starscale) {};
\node[vertex] (d) at  (-1/\starscale, 0) {};
\node[vertex] (e) at  (1/\starscale, 0) {};
\draw[edge] (b) to (a);
\draw[edge] (c) to (a);
\draw[edge] (d) to (a);
\draw[edge] (e) to (a);
\end{tikzpicture}
}
\newcommand{\udkstar}[1]{\ensuremath{\udstar_{#1}}}
\newcommand{\dkstarI}[1]{\ensuremath{\dstarI_{#1}}}
\newcommand{\dkstarII}[1]{\ensuremath{\dstarII_{#1}}}

\newcommand{\dedge}{
\begin{tikzpicture}
\node[vertex] (a) at  (0,0) {};
\node[vertex] (b) at  (2/\shapescale, 0) {};
\draw[dedge] (a) to (b);
\end{tikzpicture}
}
\newcommand{\udedge}{
\begin{tikzpicture}
\node[vertex] (a) at  (0,0) {};
\node[vertex] (b) at  (2/\shapescale, 0) {};
\draw[edge] (a) to (b);
\end{tikzpicture}
}

\def\sensdiff{\textsc{SensDiff}}
\def\composeD{\textsc{Compose-$D$-bounded}}
\def\composeproj{\textsc{Compose-projection}}
\def\sensseqD{\textsc{SensSeq-$D$-bounded}}
\def\sensseqproj{\textsc{SensSeq-projection}}

\newcommand{\bigO}[1]{\ensuremath{O\left(#1\right)}}
\newcommand{\Var}[1]{\ensuremath{\mathrm{Var}\left(#1\right)}}

\begin{document}

\title{Differentially Private Continual Release of Graph Statistics}

\author{
\IEEEauthorblockN{
Shuang Song\IEEEauthorrefmark{1},
Susan Little\IEEEauthorrefmark{1}, 
Sanjay Mehta\IEEEauthorrefmark{1},
Staal Vinterbo\IEEEauthorrefmark{2} and
Kamalika Chaudhuri\IEEEauthorrefmark{1}
}
\IEEEauthorblockN{
\normalsize{
\textit{shs037@eng.ucsd.edu},
\textit{slittle@ucsd.edu},
\textit{srmehta@ucsd.edu},
\textit{staal.vinterbo@ntnu.no},
\textit{kamalika@eng.ucsd.edu}
}}
\IEEEauthorrefmark{1}University of California, San Diego\\
\IEEEauthorrefmark{2}Norwegian University of Science and Technology
}

\maketitle

\begin{abstract}
Motivated by understanding the dynamics of sensitive social networks over time, we consider the problem of continual release of statistics in a network that arrives online, while preserving privacy of its participants. For our privacy notion, we use differential privacy -- the gold standard in privacy for statistical data analysis.

The main challenge in this problem is maintaining a good privacy-utility tradeoff; naive solutions that compose across time, as well as solutions suited to tabular data either lead to poor utility or do not directly apply. In this work, we show that if there is a publicly known upper bound on the maximum degree of any node in the entire network sequence, then we can release many common graph statistics such as degree distributions and subgraph counts continually with a better privacy-accuracy tradeoff.
\end{abstract}


\section{Introduction}

Dynamic social networks are ubiquitous models of social and economic phenomena, and analyzing them over a period of time can allow researchers to understand various aspects of human behavior. Many social networks, however, include sensitive and personal information about the people involved. Consequently, we need to design privacy-preserving algorithms that can summarize properties of dynamic social networks over time while still preserving the privacy of the participants. 

As a concrete motivating example, consider data on HIV transmission collected from patients in a particular region over multiple years \cite{little2014using, wang2015targeting, wertheim2017social}. Advances in sequencing technology allow scientists to infer putative transmission links by measuring similarities between HIV sequences obtained from different patients. These links can then be resolved into transmission networks, reflecting the patterns of transmission  in that population. Epidemiologists would like to study properties of these networks as they grow over time to understand how HIV propagates. Since there is considerable social stigma associated with HIV, these networks are highly sensitive information, and public release of their properties needs to ensure that privacy of the included individuals is not violated. Additionally, analyses of these networks need to happen intermittently -- for example, once a year -- so that properties of the network as it evolves may be studied. %

In this paper, we consider continual privacy-preserving release of graph statistics, such as degree distributions and subgraph counts, from sensitive networks where nodes and their associated edges appear over time in an online manner. For our privacy notion, we use differential privacy~\cite{DMNS06} -- the gold standard in private data analysis. Differential privacy guarantees privacy by ensuring that the participation of a single person in the dataset does not change the probability of any outcome by much; this is enforced by adding enough noise to either the input data or to the output of a function computed on the data so as to obscure the private value of a single individual. Since in our applications, a node corresponds to a single person, we use node differential privacy~\cite{hay2009accurate}, where the goal is to hide the participation of any single node.

There are two main challenges in continually releasing graph statistics with node differential privacy. The first is that node differential privacy itself is often very difficult to attain, and can only be attained in either bounded degree graphs or graphs that can be projected to be degree-bounded. The second challenge pertains to the online nature of the problem. Prior work has looked at continual release of statistics based on streaming tabular data~\cite{bolot2013private, cao2013efficient, chan2011private, kellaris2014differentially, chen2017pegasus, dwork2010differential}; however, these works rely on the fact that in tabular data, at any time $t$, we only get information about the $t$-th individual, and not about individuals who already exist in the data. This property no longer holds in online graphs, as an incoming node may bring in new information about existing nodes in the form of connecting edges, and therefore these solutions do not directly apply.

In this work, we show that if there is a publicly known upper bound on the maximum degree of any node in the entire graph sequence, then, a {\em{difference sequence}}  -- namely, the sequence of differences in the statistics computed on subsequent graphs -- has low sensitivity. The assumption of bounded maximum degree holds for many real networks, as many real-world graphs, such as 
social interaction networks, collaboration networks, computer networks and disease transmission networks, that are scale-free with power-law degree distributions have low maximum degree. Given this assumption holds, we show in particular that the sensitivity of the entire difference sequence only depends on the publicly known upper bound, and not on the length of the sequence. This implies that we can release a private version of the difference sequence with relatively high accuracy, which can be used to continually release the target statistic with high privacy-accuracy tradeoff. 

It is commonly believed that many real-world networks, such as social interaction networks, collaboration networks, computer networks and disease transmission networks are scale-free with degree distributions following a power law; such graphs have low maximum degrees.

We derive the sensitivity of the difference sequence for a number of common graph statistics, such as degree distribution, number of high degree nodes, as well as counts of fixed subgraphs. We then implement our algorithms and evaluate them on three real and two synthetic datasets against two natural baselines. Our experimental results show that the algorithm outperforms these baselines in terms of utility for these datasets over a range of privacy parameters.

\subsection{Related Work}

Since its inception~\cite{DMNS06}, differential privacy has become the gold standard for private data analysis and has been used in a long line of work -- see~\cite{SC13, DR14} for surveys. Differential privacy guarantees privacy by ensuring that the participation of a single individual in a dataset does not significantly affect the probability of any outcome; this is enforced by adding enough noise to obscure the influence of a single person. 

To apply differential privacy to graph data, it is therefore important to determine what a single person's data contributes to the graph. Prior work has looked at two forms of differential privacy in graphs -- {\em{edge differential privacy}}, where an edge corresponds to a person's private value, and {\em{node differential privacy}}, where a single node corresponds to a person. In our motivating application, a patient corresponds to a node, and hence node differential privacy is our privacy notion of choice. 

Prior work on edge differential privacy~\cite{NRS07, hay2009accurate} has looked at how to compute a number of statistics for {\em{static graphs}} while preserving privacy. For example,~\cite{karwa2011private} computes subgraph counts, and~\cite{hay2009accurate} degree distributions with edge differential privacy. It is also known how to successfully calculate more complex graph parameters under this notion; for example,~\cite{lu2014exponential} fits exponential random graph models and~\cite{wang2013differential, ahmed2013random} computes spectral graph statistics such as pagerank. 

In contrast, achieving node differential privacy is considerably more challenging. Changing a single node and its associated edges can alter even simple statistics of a static graph significantly; this means that any differentially private solution needs to add a considerable amount of noise to hide the effect of a single node, resulting in low utility. Prior work has addressed this challenge in two separate ways. The first is to assume that there is a publicly known upper bound on the maximum degree of any node in the graph~\cite{borgs2015private, gehrke2011towards}. 

The second is to use a carefully-designed projection from the input graph to a bounded degree graph, where adding or removing a single node has less effect, and then release statistics of the projected graph with privacy. To ensure that the entire process is privacy-preserving, the projection itself is required to be {\em{smooth}} -- in the sense that changing a single node should not change the statistics of the projected graph by much. This approach has been taken by~\cite{kasiviswanathan2013analyzing}, who releases degree distributions and subgraph counts for static graphs and~\cite{blocki2013differentially}, who releases subgraph counts and local profile queries. \cite{raskhodnikova2015efficient} releases degree distributions by using a flow-based projection algorithm. Finally, ~\cite{day2016publishing} proposes an improved projection method for releasing degree distributions, and is the state-of-the-art in this area. In this paper, we show that when the graph arrives {\em{online}}, existing projection-based approaches can yield poor utility, and therefore, we consider bounded degree graphs, where domain knowledge suggests an upper bound on the maximum degree.

Finally, while we are not aware of any work on differentially private statistics on streaming graph data, prior work has looked at releasing private statistics on streaming {\em{tabular data}} in an online manner~\cite{bolot2013private, cao2013efficient, chan2011private, kellaris2014differentially, chen2017pegasus, dwork2010differential}. In these settings, however, complete information about a single person (or a group of people) arrives at each time step, which makes the problem of private release considerably easier than online graph data, where newly arriving nodes may include information in the form of edges to already existing nodes. Thus, these approaches do not directly translate to online graphs.

\section{Preliminaries}

\subsection{Graphs and Graph Sequences}\label{sec:into_graph}

Formally, we consider a graph $G = (V, E)$, where a node $v \in V$ represents a person and an edge $(u, v) \in E$ a relationship. $G$ may be directed or undirected, depending on the application. We assume that each node $v \in V$ is associated with a time stamp, denoted by $v.\time$, that records when $v$ enters the graph.

In our setting, a graph arrives online as more and more of its vertices and some of their adjacent edges become visible. More specifically, at time $t$, a set of vertices $\partial{V_t}$ arrives, along with a set of edges $\partial{E_t}$; each edge in $\partial{E_t}$ has at least one end-point in $\partial{V_t}$, and the other end-point may be a vertex that arrived earlier. These vertices and edges, along with vertices and edges that arrived earlier comprise a graph $G_t$. Given a function $f$ that operates on graphs, our goal is to output (a private approximation to) $f(G_t)$ at each time step $t$.

More formally, the arrival process comprises a {\em{graph sequence}} $\calG = (G_1, G_2, \ldots)$, which is defined as a sequence of graphs with $G_t = (V_t, E_t)$ such that $V_0 = \emptyset$, $\partial{V_t} =  \{ v: v.\time = t \}$ is the set of all nodes with time stamp $t$ and $V_t = V_{t-1} \cup \partial {V_t}$ for $t \geq 1$ is the set of all nodes with time stamps $\leq t$.
Additionally, we let $E_0 = \emptyset$, 
$ \partial{E_t} = \{ (u, v) | u \in \partial{V_t}, v \in V_t  {\text{\ or\ }} u \in V_t, v \in \partial{V_t}\}$, 
and $E_t = E_{t-1} \cup \partial{E_t}$. Given a function $f$ that operates on a graph, we define $f$ applied to the graph sequence $f(\calG)$ as the sequence $(f(G_1), f(G_2), \ldots)$. 

If the graph sequence is $\calG = (G_1, G_2, \ldots, G_T)$, then the error of $\calA(G)$ is defined as: $\sum_{t=1}^{T} | \calA(G_t) - f(G_t)|$. Our goal is to design an algorithm $\calA$ that has as low error as possible.

\subsection{Differential Privacy}\label{sec:dp}

The gold standard for privacy in data-mining applications is differential privacy~\cite{DMNS06}, which essentially ensures that the participation of a single person in a database does not change the probability of any outcome by much. The formal definition is as follows.

\begin{definition}[$\epsilon$-Differential privacy]
A randomized algorithm $\calA$ is said to guarantee $\epsilon$-differential privacy if for any two databases $S$ and $S'$ that differ in the private value of a single individual, and for any $w \in \text{Range}(\calA)$, we have
\begin{align*}
\Pr{}{\calA(S) = w} \leq e^{\epsilon} \cdot \Pr{}{\calA(S') = w},
\end{align*}
where the probability is with respect to the randomness in $\calA$. Here $\epsilon$ is a privacy parameter, often called the {\em{privacy budget}}. 

\end{definition}

\paragraph{Global Sensitivity Mechanism.} A popular differential privacy mechanism is the Global Sensitivity Mechanism, introduced by~\cite{DMNS06}. Let $f$ be a function that operates on a database $S$. The Global Sensitivity of $f$, denoted by $\GS{}{f}$, is the maximum value of the difference $\| f(S) - f(S') \|_1$ when $S$ and $S'$ are any two databases that differ by a the participation of a single person. 

Given a function $f$, a database $S$ and a privacy budget $\epsilon$, the Global Sensitivity Mechanism computes an $\epsilon$-differentially private approximation to $f(S)$ as follows: $\calA_{GS}(S) = f(S) + \Lap{\frac{\GS{}{f}}{\epsilon}}.$ It was shown in~\cite{DMNS06} that this method preserves $\epsilon$-differential privacy. 

\paragraph{Node Differential Privacy.} To apply differential privacy to graphs, we need to determine what constitutes a single person's data in a graph. For the kind of graphs that we will study, a node $v$ corresponds to a single person. This is known as {\em{node differential privacy}}~\cite{hay2009accurate}, %
which ensures that the addition or removal of a single node along with its adjacent edges does not change the probability of any outcome by much.

\subsection{Bounded Degree Graphs}

A major challenge with ensuring node differential privacy is that the global sensitivity $\GS{}{f}$ may be very large even for simple graph functions $f$, which in turn requires the addition of a large amount of noise to ensure privacy. For example, if $f$ is the number of nodes with degree $\geq 1$, and we have an empty graph $G$ on $n$ nodes, then adding a single node connected to every other node can increase $f$ by as much as $n$.

Prior work has addressed this challenge in two separate ways. The first is by considering {\em{Bounded Degree Graphs}} \cite{kasiviswanathan2013analyzing, blocki2013differentially, borgs2015private}, where an a-priori bound on the degree of any node is known to the user and the algorithm designer. This is the solution that we will consider in this paper. 

A second line of prior work ~\cite{blocki2013differentially, kasiviswanathan2013analyzing, raskhodnikova2015efficient, day2016publishing} 
presents {\em{Graph Projections}} algorithms
that may be used to project graphs into lower degree graphs such that the resulting projections have low global sensitivity for some graph functions.  In Section~\ref{sec:highdeg}, we show that natural extensions of some of these projections may be quite unstable when a graph appears online. 

We first define bounded degree graphs.
Let $\degree{G}{v}$ denote the degree of node $v$ in an undirected graph $G$, $\outdegree{G}{v}$ and $\indegree{G}{v}$ denote the out-degree and in-degree of $v$ in a directed graph $G$.
\begin{definition}\label{def:D-bdd}
An undirected graph $G = (V, E)$ is $D$-bounded if $\degree{G}{v}\leq D$ for any $v\in V$. 
A graph sequence $\calG = (G_1, G_2, \ldots)$ is $D$-bounded if for all $t$, $G_t$ is $D$-bounded. %
In other words, the degree of all nodes remain bounded by $D$ in the entire graph sequence. 

A directed graph $G = (V, E)$ is $\Dout$-out-bounded if $\outdegree{G}{v}\leq \Dout$ for any $v\in V$. 
A graph sequence $\calG = (G_1, G_2, \ldots)$ is $\Dout$-out-bounded if for all $t$, $G_t$ is $\Dout$-out-bounded. %
Similarly, $G$ is $\Din$-in-bounded if $\indegree{G}{v}\leq \Din$ for any $v\in V$. $\calG$ is $\Din$-in-bounded if for all $t$, $G_t$ is $\Din$-in-bounded.

We say a directed graph or a graph sequence is $(\Din, \Dout)$-bounded if it is both $\Din$-in-bounded and $\Dout$-out-bounded.
\end{definition}

In this work, we assume that the domain consists only of degree bounded graphs. This ensures that the global sensitivity of certain common graph functions, such as degree distribution and subgraph counts, is low, and allows us to obtain privacy with relatively low noise. Additionally, many common sensitive graphs, such as the HIV transmission graph and co-authorship networks, typically have relatively low maximum degree, thus ensuring that the assumption holds for low or moderate values of $D$.

\subsection{Graph Functions}

This work will consider two types of functions on graph sequences. The first consists of functions of the degree distribution. The specific functions we will look at for undirected graphs are $\highdegree{\tau}{G}$, which counts the number of nodes in $G$ with degree $\geq \tau$ and the degree histogram $\hist{G}$, which counts the number of nodes with degree $d$ for any $d \in \mathbb{N}_{+}$.
Similarly, for directed graph, we consider $\highoutdegree{\tau}{G}$, the number of nodes with out-degree $\geq \tau$, and the out-degree histogram $\histout{G}$, which counts the number of nodes with out-degree $d$ for any $d \in \mathbb{N}_{+}$.

The second class of functions will involve subgraph counts. Given a subgraph $S$, we will count the number of occurrences of this subgraph $\S{}{G}$ in the entire graph $G$. For example, when $S$ is a triangle, $\S{}{G}$ will count the number of triangles in the graph. When $G$ is directed, so will be the corresponding subgraphs. 

\subsection{Other Notations}\label{sec:notation}

In a directed graph, an edge is denoted by an ordered tuple, i.e., $(u, v)$ represents a directed edge pointing from node $u$ to $v$.

We use $\seq{a_t}{t = 1}{T}$ as an abbreviation for vector $(a_1, a_2, \dots, a_T)$.

For any integer $i$, we use $[i]$ to denote the set $\{1,2,\dots,i\}$.

A degree histogram $h$ is a mapping from degrees to counts, i.e., given $d \in \mathbb{N}_{+}$, $h(d)$ is the number of nodes with degree equal to $d$. 
We define the distance between two histograms $h$ and $h'$ as $\|h(d) - h'(d)\|_1 = \sum_{d\in\mathbb{N}_{+}} |h(d) - h'(d)|$.
Given two sequences of histograms 
$\seq{h_t}{t=1}{T}$ and $\seq{h'_t}{t=1}{T}$, 
we define the generalized $L_1$ distance between them as $\sum_{t = 1}^T \|h_t(d) - h'_t(d)\|_1$.

\section{Main Algorithm}\label{sec:main_alg}

Recall that we are given as input a $D$-bounded (or $(\Din, \Dout)$-bounded) graph sequence $\calG = ( G_1, G_2, \ldots, G_T )$ that arrives online, a privacy budget $\epsilon$ and a function $f$. Our goal is to publish an $\epsilon$-differentially private approximation to the sequence $f(\calG)$ in an online manner. Specifically, at time $t$, an incoming vertex set $\partial{V_t}$ and edges $\partial{E_t}$ adjacent to it and the existing vertices arrive, and our goal is to release a private approximation to $f(G_t)$ with low additive $L_1$-error. 

\paragraph{Baseline Approaches.} A naive approach is to calculate $f(G_t)$ at each $t$ and add noise proportional to its global sensitivity over $\epsilon$. Since $\partial{E_t}$ may contain information on individuals in $G_{t-1}$ in the form of adjacent edges, this procedure will not provide $\epsilon$-differential privacy. 

The correct way to do privacy accounting for this method is by sequential composition~\cite{DMNS06}. Suppose the graph sequence has total length $T$ and we allocate privacy budget $\epsilon/T$ to each time step; then at time $t$, we calculate $f(G_t)$ and add noise proportional to its global sensitivity divided by $\epsilon/T$. If the global sensitivity of $f(G_t)$ is $O(1)$, then, we add $O(T/\epsilon)$ noise to $f(G_t)$, which results in a $\Theta(T^2/\epsilon)$  expected $L_1$-error between $f(\calG)$ and the output of the algorithm. 

A second approach is to calculate $f(G_t)$ and add noise proportional to the global sensitivity of the (entire) sequence $f(\calG)$ divided by $\epsilon$. This preserves $\epsilon$-differential privacy. However, the global sensitivity of the sequence $f(\calG)$ typically grows linearly with $T$, the length of the entire graph sequence, even if the graph sequence itself is degree-bounded. For example, if $f(G)$ is the number of nodes in $G$ with degree $\geq \tau$, then, a single extra node with degree $\tau + 1$, added at time $t=1$, can increase $f(G_t)$ by $1$ for {\em{every}} $t$, resulting in a global sensitivity of $\Omega(T)$. Consequently, the expected $L_1$-error between the true value of $f(\calG)$ and the output of this approach is as again large as $\Theta(T^2/\epsilon)$.

\paragraph{Our Approach.} The main observation in this work is that for a number of popular functions, the {\em{difference sequence}} $\Delta = ( f(G_1), f(G_2) - f(G_1), f(G_3) - f(G_2), \ldots)$ has considerably better properties. Observe that unlike certain functions on tabular data \cite{bolot2013private, cao2013efficient, chan2011private, kellaris2014differentially, chen2017pegasus, dwork2010differential}, releasing $f(G_t) - f(G_{t-1})$ after adding noise proportional to its sensitivity over $\epsilon$ will still not be $\epsilon$-differentially private -- this is because $\partial{E_t}$ can still include edges adjacent to people in $G_{t-1}$. 

However, the difference sequence $\Delta$ does have  considerably less global sensitivity than $f(\calG)$. In particular, we show that if the graph sequence $\calG$ is $D$-bounded, then, the global sensitivity $\GS{}{\Delta}$ of the entire difference sequence for a number of popular functions $f$ depends only on $D$ and not on the sequence length $T$. For example, in Section~\ref{sec:highdeg}, we show that when $G$ is an undirected graph and $f$ is the number of nodes with degree $\geq \tau$, the global sensitivity of the entire difference sequence is at most $2D + 1$. 

This immediately suggests the following algorithm.  At time $t$, calculate the difference $\Delta_t = f(G_t) - f(G_{t-1})$, and add Laplace noise proportional to its global sensitivity over $\epsilon$ to get a private perturbed version $\tilde{\Delta}_t$. Release the partial sum $\sum_{s=1}^{t} \tilde{\Delta}_s$, which is an approximation to $f(G_t)$. Since the expected value of $\tilde{\Delta}_t - \Delta_t$ is independent of $T$, the maximum standard deviation of any partial sum is at most $O(\sqrt{T}/\epsilon)$, which results in an expected $L_1$-error of $O(T^{3/2}/\epsilon)$ -- better than the $O(T^2/\epsilon)$-error achieved by the two baseline approaches.

\begin{algorithm}[h]
\caption{\sensdiff(Graph sequence $\calG$, query $f$, privacy parameter $\epsilon$)}
\label{alg}
\begin{algorithmic}
\For{$t=1, \ldots, T$}
\State{Receive $\partial{V_t}$ and $\partial{E_t}$, and construct $G_t$.}
\State{Calculate $\Delta_t = f(G_t) - f(G_{t-1})$.}
\State{Let $\GS{D}{\Delta}$ be the global sensitivity of the difference sequence;}
\State{Calculate $\tilde{\Delta}_t = \Delta_t + \Lap{\frac{\GS{D}{\Delta}}{\epsilon}}$, and the partial sum $\sum_{s=1}^{t} \tilde{\Delta}_s$.}
\EndFor
\State{\textbf{return} $\seq{\sum_{s=1}^{t} \tilde{\Delta}_s}{t=1}{T}$}
\end{algorithmic}
\end{algorithm}

The full algorithm, applied to a generic function $f$, is described in Algorithm~\ref{alg}. We call it \sensdiff\ as it uses the global sensitivity of the difference sequence $\Delta$.
The rest of the paper is devoted to analyzing the global sensitivity of the difference sequence for a number of popular graph functions $f$. Our analysis exploits specific combinatorial properties of the graph functions in question, and is carried out for two popular classes of graph functions -- functions of the degree distribution and subgraph counts.

\section{Functions of Degree Distributions}\label{sec:degree}

We begin with functions of the degree distributions of the graph sequence, and consider both directed and undirected graphs. 
A summary of the results in this section is provided in Table~\ref{tab:resultsdegree}.

\begin{table*}[h]
\centering
\caption{Summary of degree distribution results.}
\label{tab:resultsdegree}
\begin{tabular}{ l || c | c}
\hline
						&	Undirected graph					&	Directed graph	\\ \hline
(out-)degree histogram	&	$4D^2 + 2D + 1$ (for $D$-bounded)	&	$4\Dout\Din + 2\Dout + 1$ (for $(\Din, \Dout)$-bounded)			\\ \hline
high-(out-)degree nodes	&	$2D+1$ (for $D$-bounded)			&	$2\Din+1$ (for $\Din$-in-bounded)				\\ \hline
\end{tabular}
\vspace{-10pt}
\end{table*}

\subsection{Undirected Graphs}

For undirected graphs, we will consider two functions applied to graph sequences -- first, the number of nodes with degree greater than or equal to a threshold $\tau$, and second, the degree histogram. 

\subsubsection{Number of High Degree Nodes}\label{sec:highdeg}

Recall that $\highdegree{\tau}{G}$ is the number of nodes in $G$ with degree $\geq \tau$. We show below, that for $D$-bounded graphs, the difference sequence corresponding to $\highdegree{\tau}{G}$ has global sensitivity at most $2D+1$.

\begin{lemma}\label{lem:high-degree-undirected}
Let $f(G) = \highdegree{\tau}{G}$. For $D$-bounded graphs, the difference sequence corresponding to $f$ has global sensitivity at most $2D+1$. In fact, the global sensitivity is $2D+1$ for any $\tau < D$.
\end{lemma}
Notice that $\tau \leq D$ is needed for the statistic to be meaningful; if $\tau > D$, there is no high-degree node. 

\paragraph{Projection Yields High Sensitivity in Graph Sequence}
A common idea in static graph analysis with \nodeprivacy\ is to project the original graph into a bounded-degree graph. The sensitivity of some common statistics on this projected graph scales with the degree bound instead of the total number of nodes.
The current state-of-the-art projection algorithm is proposed in \cite{day2016publishing}. Given a projection threshold $\tD$, a graph $G = (V, E)$ and an ordering of the nodes in $V$, the algorithm constructs a bounded-degree graph $G^{\tD}$ as follows. First, it adds all nodes in $V$ to $G^{\tD}$; then it orders all edges in $E$ according to the ordering of $V$, and for each edge $(u,v)$, adds it to $G^{\tD}$ if and only if the addition does not make the degree of either $u$ or $v$ exceed $\tD$.

This algorithm can be easily adapted to the online graph setting. \showinlater{We give a complete description of the algorithm in Appendix~\ref{sec:appendix_baseline}}.
However, it can be shown that the global sensitivity of the difference sequence is proportional to the total number of publications. 

\begin{lemma}\label{lem:high-degree projection}
Let $f(G) = \highdegree{\tau}{G^{\tD}}$. For $D$-bounded graphs, the corresponding difference sequence that ends at time $T$ has global sensitivity at least $\tD T$ for any $\tD > \tau > 0$.
\end{lemma}
Notice that $\tD > \tau$ is needed for $\highdegree{\tau}{G^{\tD}}$ to be meaningful; otherwise, we would have $\highdegree{\tau}{G^{\tD}} = |V|$ for any $G$.

\subsubsection{Degree Histogram}

Degree histogram is another informative statistic of a graph. However, we can show that even for bounded graph, the sensitivity can scale quadratically with the degree bound. We use the generalized $L_1$ distance defined in Section~\ref{sec:notation} as the distance metric for the global sensitivity.

\begin{lemma}\label{lem:degree-hist undirected}
Let $f(G) = \hist{G}$. For $D$-bounded graphs, the difference sequence corresponding to $f$ has global sensitivity $4D^2 + 2D + 1$.
\end{lemma}

\subsection{Directed Graphs}
For Directed graphs, we show similar results for the number of nodes with out-degree greater than or equal to a threshold $\tau$ and the out-degree histogram. Similar results can be obtained for in-degree.

\subsubsection{Number of High Out-Degree Nodes}
Recall that $\highoutdegree{\tau}{G}$ denotes the number of nodes in $G$ with degree $\geq \tau$. We show below that for $\Din$-in-bounded graphs, the difference sequence corresponding to $\highoutdegree{\tau}{G}$ has global sensitivity $2\Din+1$.
\begin{lemma}\label{lem:high-degree-directed}
Let $f(G) = \highoutdegree{\tau}{G}$. For $\Din$-in-bounded graphs, the difference sequence corresponding to $f$ has global sensitivity $2\Din+1$.
\end{lemma}

\subsubsection{Out-Degree Histogram}
We show that for bounded directed graphs, the sensitivity of the histogram scales quadratically with the degree bounds as well.
\begin{lemma}\label{lem:out-degree-hist directed}
Let $f(G) = \histout{G}$. For $(\Din, \Dout)$-bounded graphs, the difference sequence corresponding to $f$ has global sensitivity $4\Dout\Din + 2\Dout + 1$. 
\end{lemma}

\section{Functions of Subgraph Counts}\label{sec:subgraph}

In this section, we consider the count of some common directed and undirected subgraphs.

\paragraph{Popular subgraphs}
In undirected graphs, we consider three subgraphs. 1) an edge, including two nodes and the edge between them, 2) a triangle, including three nodes with edges between any two of them and 3) a $k$-star, including one center node $c$, $k$ boundary nodes $\{b_1,\dots,b_k\}$ and edges $\{(c,b_i):i\in[k]\}$. Table~\ref{tab:patterns_undirected} summarizes these subgraphs.

In directed graphs, we consider five subgraphs. 1) an edge, including two nodes and an directed edge between them, 2) triangle \rom{1}, including nodes $\{v_1,v_2,v_3\}$ and edges $\{(v_1,v_2),(v_2,v_3),(v_3,v_1)\}$, 3) triangle \rom{2}, including nodes $\{v_1,v_2,v_3\}$ and edges $\{(v_1,v_2),(v_1,v_3),(v_2,v_3)\}$, 4) an out-$k$-star, including one center node $c$, $k$ boundary nodes $\{b_1,\dots,b_k\}$ and edges $\{(c,b_i):i\in[k]\}$, 5) an in-$k$-star, including one center node $c$, $k$ boundary nodes $\{b_1,\dots,b_k\}$ and edges $\{(b_i,c):i\in[k]\}$.
Table~\ref{tab:patterns_directed} summarizes these subgraphs.

\begin{table}[h]
\centering
\caption{Subgraphs in undirected graphs.}
\label{tab:patterns_undirected}
\begin{tabular}{ c | c | c}
Edge 		&	Triangle	&	$k$-star		\\\hline
$S^u_E$		&	$S^u_{\triangle}$		&	$S^u_{\star}$			\\\hline
\udedge 		&	\udtriangle &	$\udkstar{}$ 
\end{tabular}
\end{table}

\begin{table}[h]
\centering
\caption{Subgraphs in directed graphs.}
\label{tab:patterns_directed}
\begin{tabular}{ c | c | c | c | c}
Edge 		&	Triangle \rom{1}	&	Triangle \rom{2}	 &	Out-$k$-star	&	In-$k$-star	\\\hline
$S^d_E$		&	$S^d_{\triangle^1}$	&	$S^d_{\triangle^2}$		&	$S^d_{\star^o}$	&	$S^d_{\star^i}$			\\\hline
\dedge 		&	\dtriangleI	&	\dtriangleII &	$\dkstarI{}$	&	$\dkstarII{}$ 
\end{tabular}
\end{table}

\subsection{Undirected Graphs}

First, we present a general result that applies to any undirected subgraph $S$. 
Recall that $\S{}{G}$ denotes the total number of copies of $S$ in graph $G$.
\begin{lemma}\label{lem:subgraph undirected}
Given any undirected subgraph $S$, if $\S{}{G}$ changes by at most $S_{+}$ with an additional node with degree $D$ (and the corresponding edges), then the difference sequence corresponding to $\S{}{\cdot}$ has global sensitivity $S_{+}$ for any $D$-bounded graph $G$.
\end{lemma}

Now we show the values of $S_{+}$ for the subgraphs listed in Table~\ref{tab:patterns_undirected}.
\begin{lemma}\label{lem:subgraph undirected patterns}
Given the degree bound $D$, the value of $S_{+}$ for some common subgraphs are:
\begin{enumerate}
\item for $S^u_E$,		$S_{+} = D$;
\item for $S^u_{\triangle}$, 	$S_{+} = {D \choose 2}$;
\item for $S^u_{\star}$, 	$S_{+} = D{D-1 \choose k-1} + {D \choose k}$.
\end{enumerate}
\end{lemma}

\subsection{Directed Graphs}

Again, we first present a lemma that applies to any subgraph $S$, and then show the values of $S_{+}$, the maximum change in the subgraph count caused by an additional node, for the subgraphs in Table~\ref{tab:patterns_directed}.
\begin{lemma}\label{lem:subgraph directed}
Given any directed subgraph $S$, if $\S{}{G}$ changes by at most $S_{+}$ with an additional node with $\Din$ in-degree and $\Dout$ out-degree (and the corresponding edges), then the difference sequence corresponding to $\S{}{\cdot}$ has global sensitivity $S_{+}$ for any $(\Din, \Dout)$-bounded graph $G$.
\end{lemma}

\begin{lemma}\label{lem:subgraph directed patterns}
Given degree bounds $\Din$ and $\Dout$, the value of $S_{+}$ for some common subgraphs are:
\begin{enumerate}
\item for $S^d_{E}$,		$S_{+} = \Din + \Dout$;
\item for $S^d_{\triangle^1}$, 	$S_{+} = \Din \Dout$;
\item for $S^d_{\triangle^2}$, $S_{+} = {{\Din + \Dout} \choose 2}$;
\item for $S^d_{\star^o}$, 	$S_{+} = \Din {\Dout - 1 \choose k - 1} + {\Dout \choose k}$;
\item for $S^d_{\star^i}$, $S_{+} = \Dout {\Din - 1 \choose k - 1} + {\Din \choose k}$.
\end{enumerate}
\end{lemma}

\section{Experiments}
We next demonstrate the practical applicability of the proposed algorithm by comparing it with some natural baselines. In particular, we investigate the following questions: \begin{enumerate}
\item What is the utility offered by \sensdiff\ as a function of the privacy parameter $\epsilon$ and the number of releases?
\item How does its utility compare with existing baselines, such as composition across time steps, and composition coupled with graph projection?
\end{enumerate}
These questions are considered in the context of five datasets -- two synthetic and three real online graphs. We consider two versions of each dataset -- directed and undirected, and two graph statistics -- the number of high-degree nodes and the number of edges.

\subsection{Methodology}

\subsubsection{Baseline Algorithms}
We consider two natural baselines based on sequential differential privacy composition -- \composeD\ and \composeproj. \composeD\ considers each graph $G_t$ in the sequence separately, and adds noise to $f(G_t)$ that is proportional to its $D$-bounded global sensitivity divided by $\epsilon/T$. Here $\epsilon$ is the privacy parameter and $T$ is the number of releases. \composeproj uses a state-of-the-art projection algorithm -- the one proposed in \cite{day2016publishing} -- to project each $G_t$ into $\tilde{G}_t$, and releases $f(\tilde{G}_t)$ after adding noise proportional to its global sensitivity divided by $\epsilon/T$.

There are two other natural approaches. The first is to compute $f(\calG)$ and add noise proportional to its $D$-bounded global sensitivity divided by $\epsilon$; the second is to use the projection algorithm to obtain a sequence of projected graphs $\tilde{\calG} = (\tilde{G}_1,\dots,\tilde{G}_T)$, compute $f(\tilde{\calG})$ and add noise proportional to the its global sensitivity divided by $\epsilon$. However, we can show that the utility of either of these approaches is guaranteed to be at most that of \composeD\ and \composeproj; details are omitted due to space constraints.

\subsubsection{Choice of Parameters}
\sensdiff\ and \composeD\ both require an a-priori bound on the graph degree. We set this bound to be the actual maximum degree rounded up to the nearest $5$-th integer.

\composeproj\ requires a projection threshold $\tD$ (or, $\tDin$ and $\tDout$ for directed graphs). This parameter is chosen by parameter tuning -- we pick the parameter value out of a predetermined list that leads to the lowest error. Note that for the sake of fairness, we do not allocate any extra privacy budget to parameter tuning, which would be the case in reality; thus our estimate of the performance of \composeproj\ is {\em{optimistic}}. 

The threshold $\tau$ in the high-degree or high-out-degree nodes experiments is set to be the $90$-th percentile of the (non-private) degree distribution.

\subsection{Datasets}
We use five datasets -- three real and two synthetic. In addition to the standard directed version of each dataset, we also consider an undirected version that is obtained by ignoring the edge directions. A brief summary of these datasets is presented in Table~\ref{tab:data}.

\begin{table*}[h]
\centering
\caption{Summary of common properties of the graph datasets. Max degree refers to the undirected version of the graph while max in-degree and max out-degree refer to the directed version.}
\label{tab:data}
\begin{tabular}{ l || c | c | c | c | c | c}
\hline
					&	\# nodes	& \# edges	&	timespan (year)	&	max degree	&	max in-degree	&	max out-degree	\\ \hline
HIV transmission	&	$1660$	&$456$	&$21$	&$13$	&$12$	&$8$				\\ \hline
Patent citation		&$91614$	&$475427$	&$15$	&$247$	&$211$	&$246$	\\ \hline
Paper citation 		&$9038$	&$5249$	&$24$	&$36$	&$19$	&$36$	\\ \hline
Synthetic \rom{1}	&$1990$	&$665$	&$20$	&$5$		&$1$		&$4$	\\ \hline
Synthetic \rom{2}	&$1088$	&$588$	&$20$	&$6$		&$1$		&$5$	\\ \hline
\end{tabular}
\vspace{-0pt}
\end{table*}

\subsubsection{Real Data Sets}

\paragraph{HIV transmission graph} This is a graph of potential HIV transmissions where a node represents a patient and an edge connects two patients whose viral sequences have high similarity. An edge thus represents a plausible transmission; the graph also has spurious edges that may correspond to a patient transmitting the disease to multiple others within a short period of time. About two-thirds of the patients have an estimated date of infection (EDI) ranging from $1996$ to $2016$, which is taken as the time stamp of the corresponding node. For the remaining nodes, EDI could not be estimated as the patient was admitted long after infection; we set the corresponding time stamps as $1995$.

\paragraph{Patents citation graph} The patents citation graph \cite{data_patents} contains all US patents granted between $1963$ and $1999$, and all citations made by patents from $1975$ to $1999$. A patent $u$'s citing patent $v$ naturally yields a directed edge from $v$ to $u$. We pick all patents under subcategory {\em{Computer Hardware \& Software}} (indexed $22$) to form both a directed and an undirected graph, and publish statistics from years $1985$ to $1999$.

\paragraph{Paper citation graph} This is a graph of articles and their citations from ACL derived by~\cite{data_citation}; each article has a recorded publication date from $1975$ to $2013$. We select the positive citations, where an edge from $v$ to $u$ implies that $u$ endorsed the article $v$, and publish statistics from $1990$ to $2013$. 

\subsubsection{Synthetic Data Sets}

In addition to the real data sets, we consider two synthetic graphs that are generated from two separate disease transmission models.

\def\BA{Barab\'asi--Albert}

\paragraph{Synthetic disease transmission graph \rom{1}}
This is a synthetic graph of disease transmissions based on the \BA\ preferential attachment model. In the \BA\ model, there are $m_0$ initial nodes, followed by a number of nodes that arrive sequentially. A node on arrival connects to $k$ existing nodes, with a higher chance of connecting to nodes with higher degree. 

We make three modifications to this model so that the generated graph is a more realistic disease transmission network. First, we
assume there is a total of $Y$ years with $n$ nodes added per year; each node has a time stamp -- its year of arrival -- and the initial nodes have time stamp $0$. All edges are directed from nodes with lower time stamps to those with higher time stamps.  Second, to model the large number of isolated nodes that exist in real disease transmission networks, we ensure that each new node is isolated with probability proportional to a parameter $P_{\text{isolated}}$. Third, in practice, an infected individual is usually less likely to spread infection as time passes due to treatment or death. We build this property into the model by adding an extra decaying factor to the connection probability, i.e., the probability of a new node's connecting to an existing node $v$ is proportional to $\degree{}{v} \times (\texttt{current\_time} - v.\time + 1)^{-c}$ (or $\outdegree{}{v} \times (\texttt{current\_time} - v.\time + 1)^{-c}$ in directed graph) where $c$ is the decay parameter. For our experiments, we generate a graph with parameters $P_{\text{isolated}} = 0.5$, $k = 1$, $m_0 = 500$, $n=70$, $Y=20$ and $c = 1$.

\paragraph{Synthetic disease transmission graph \rom{2}} This is a synthetic disease transmission graph drawn from the popular SIR model~\cite{SIRmodel} of infection overlaid on an underlying \BA\ social network. In the SIR model, a node or individual as three statuses -- susceptible (S), infectious (I) and recovered (R). The infectious individuals transmit the disease to susceptible individuals through social links with a transmission probability $P_t$. With probability $P_r$, an infectious individual can recover; once recovered, an individual will not get infected again.

We generate an undirected social network $G_{\text{interact}} = (V, E)$ from the \BA\ model, where a node represents a person and an edge a social interaction. We then simulate the transmission process as follows. Initially, every node is susceptible (S) except for $n_0$ randomly picked infectious (I) nodes. At time step $t$, an infectious node changes status to recovered (R) with probability $P_r$; then, each node $u$ that is still infectious infects each one of its neighbors in $G_{\text{interact}}$ with probability $P_i / \degree{G_{\text{interact}}}{u}$. This gives us a transmission graph where a directed edge is a disease transmission. The time stamp of each node is the time when it is infected. We note that any node in the social interaction graph that has never been infected will not appear in the transmission graph. For our experiments, we use the parameters $P_r = 0.1$, $P_i = 0.18$ and a underlying interaction graph $G_{\text{interact}}$ of size $10000$ with attachment parameter $k = 2$.

\subsection{Results}

\begin{table*}[!t]
\setlength\tabcolsep{0pt} %
\resizebox{\textwidth}{!}{
\begin{tabularx}{\textwidth}{c c c c c}
& Directed, high-degree & Undirected, high-degree & Directed, edge & Undirected, edge \\ 
\rotatebox[origin=l]{90}{\quad HIV transmission}
& \includegraphics[width=0.245\textwidth]{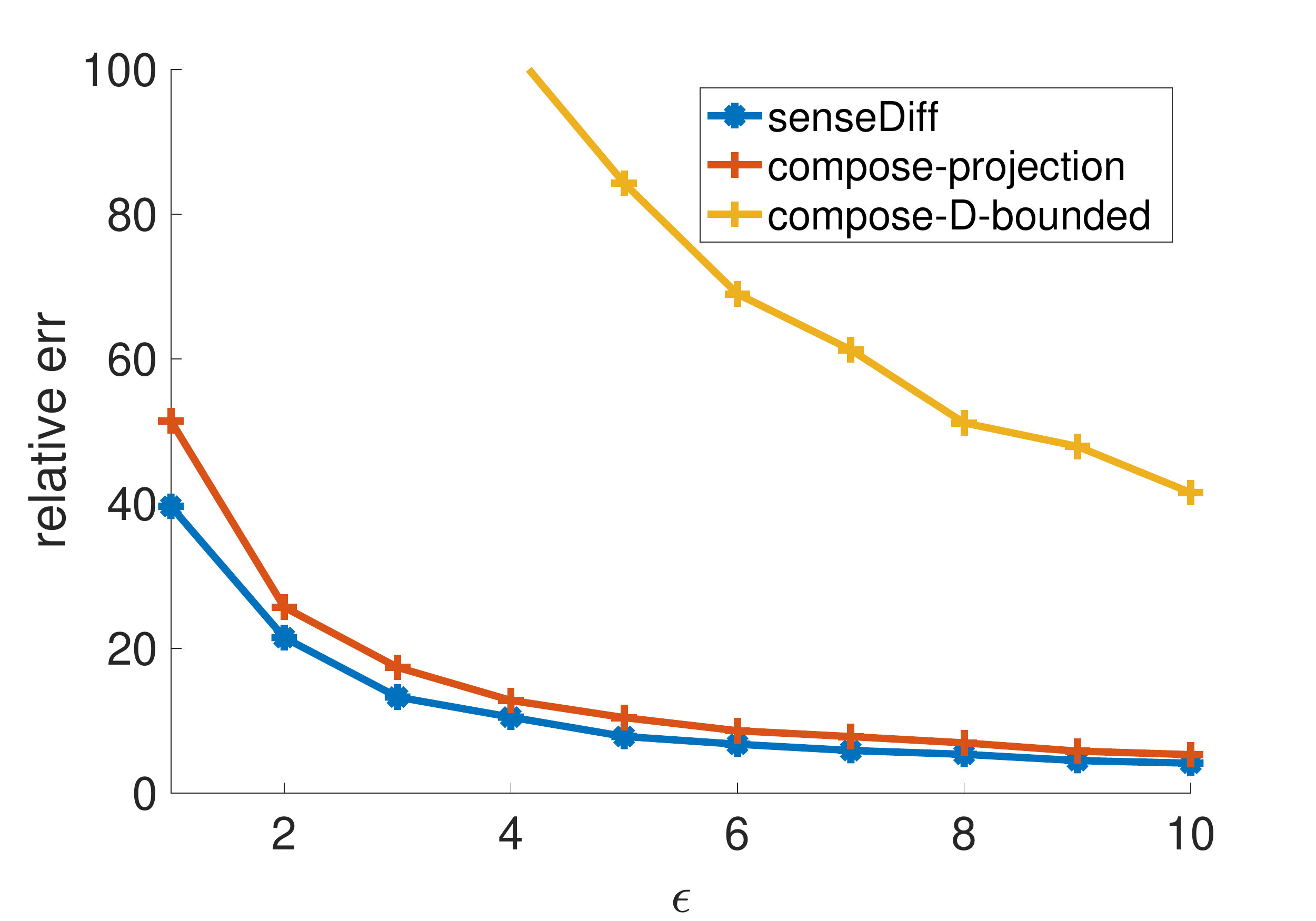}
& \includegraphics[width=0.245\textwidth]{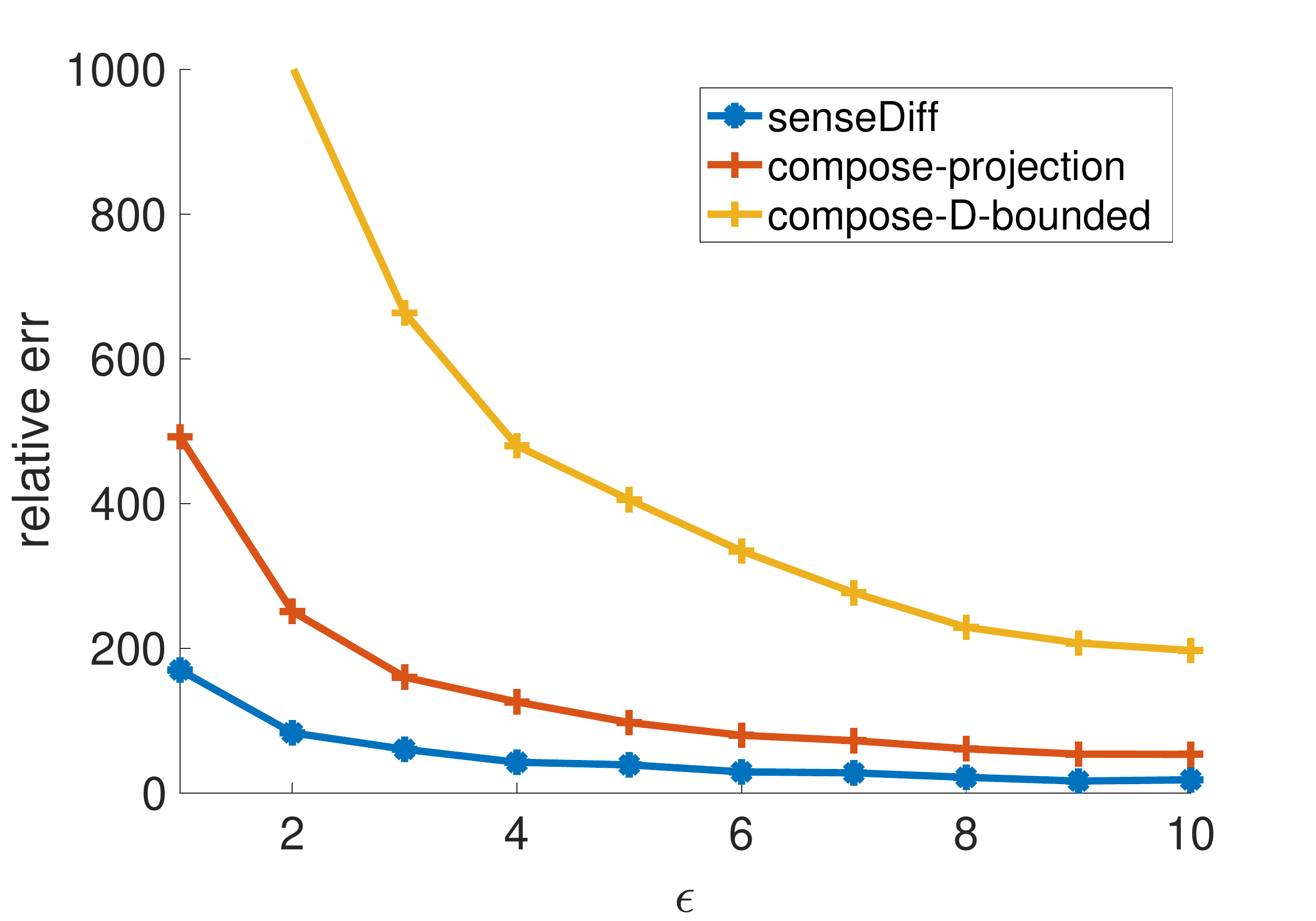}
& \includegraphics[width=0.245\textwidth]{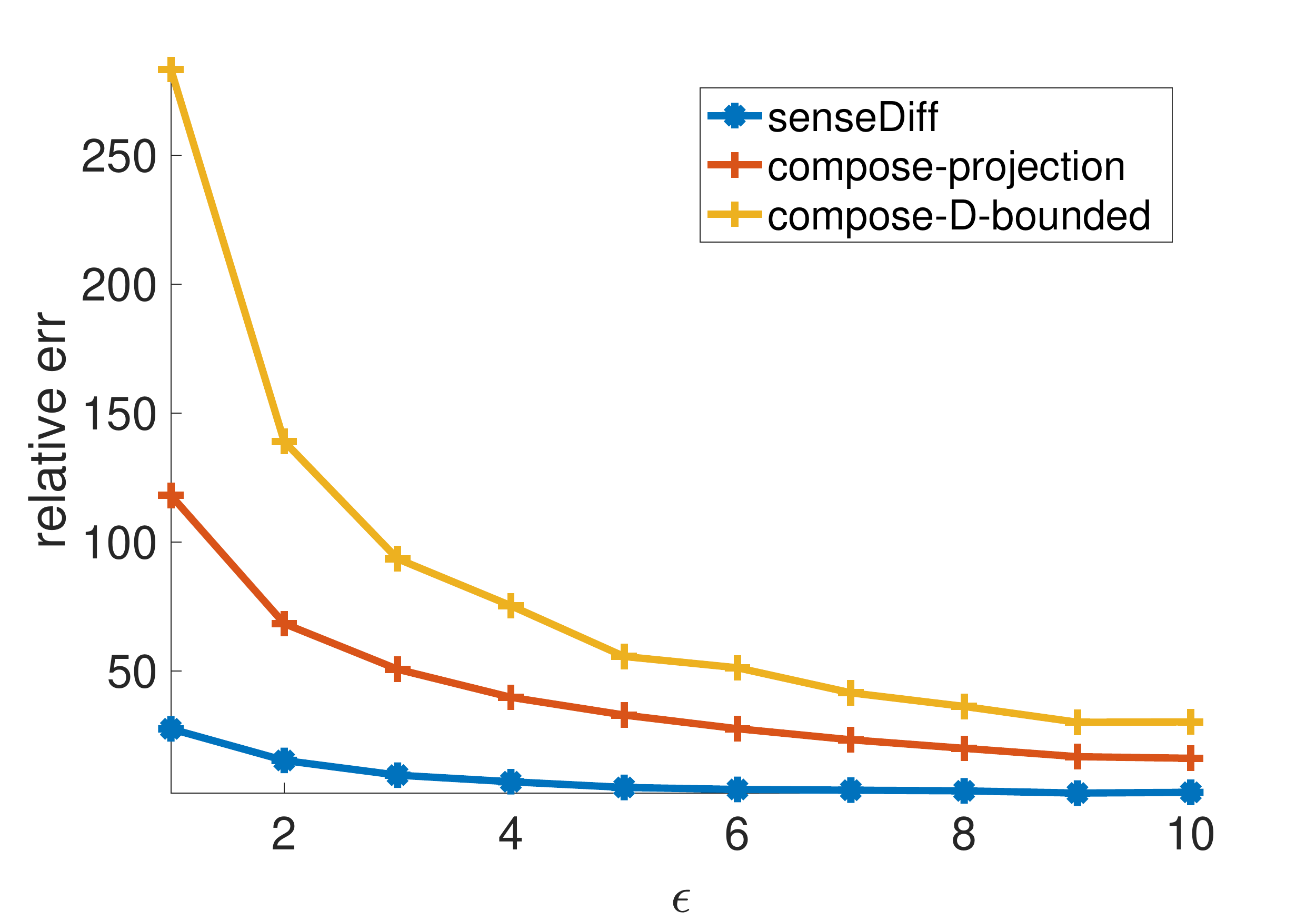}
& \includegraphics[width=0.245\textwidth]{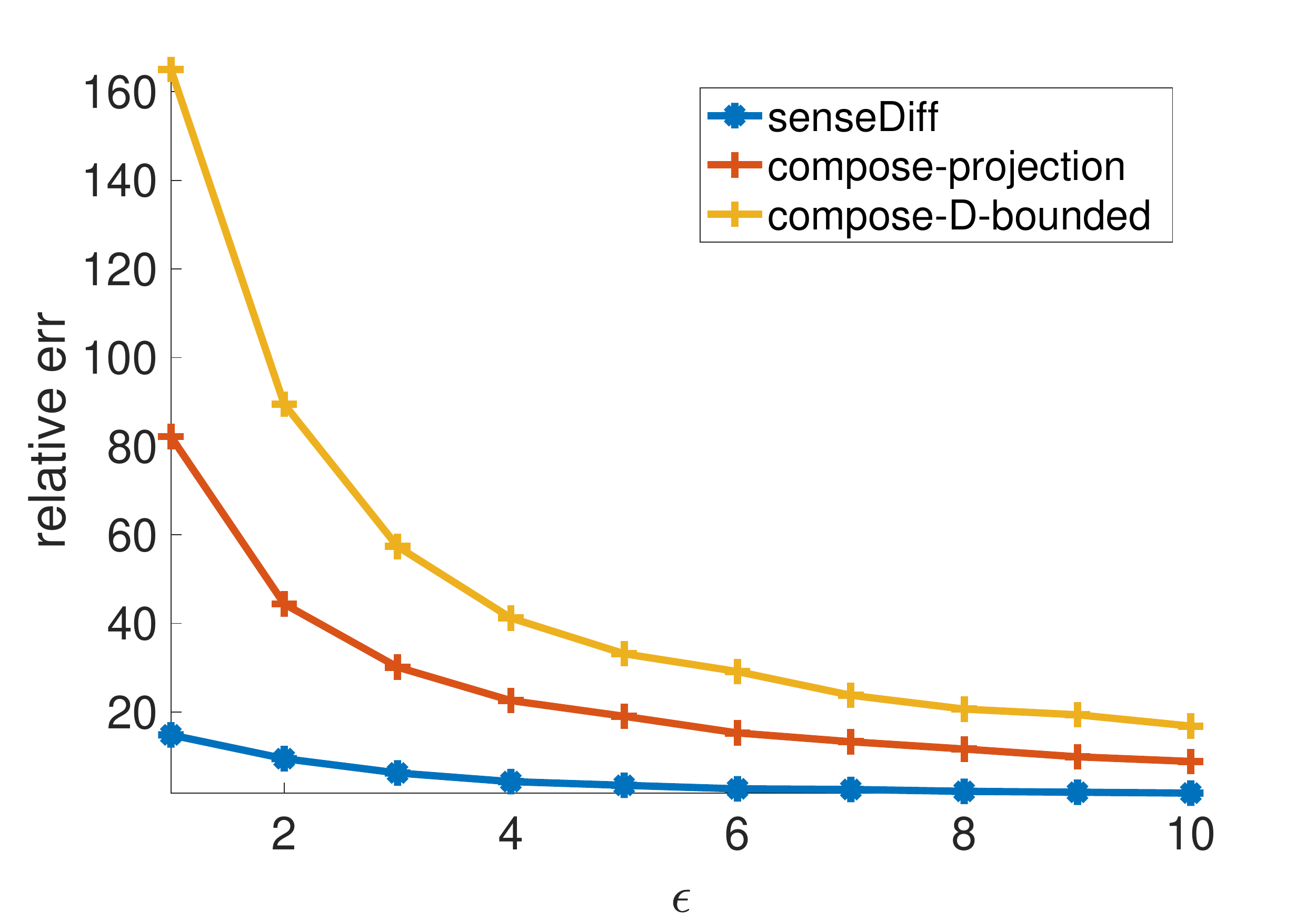}
\\
\rotatebox[origin=l]{90}{\qquad Patent citation}
&\includegraphics[width=0.245\textwidth]{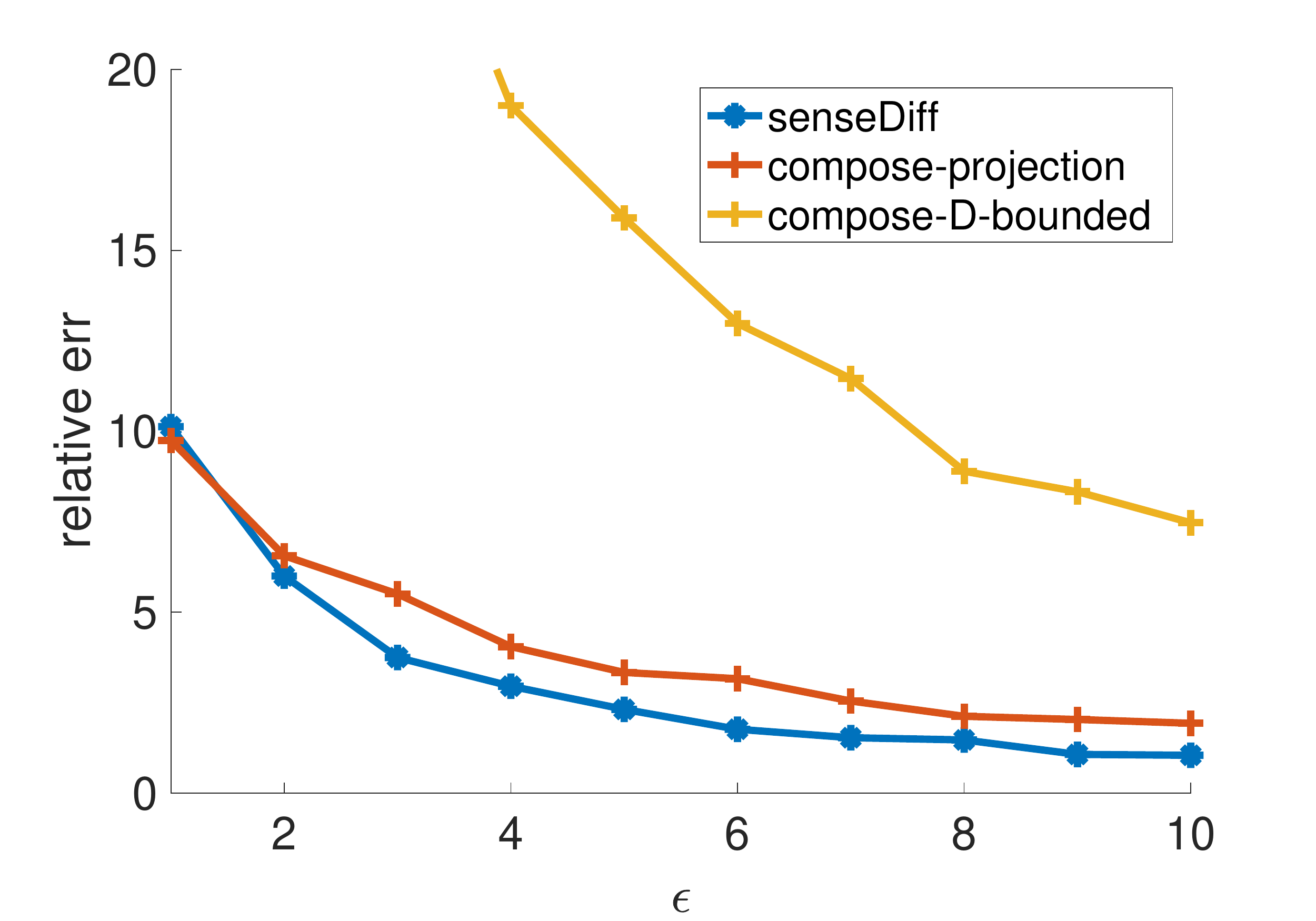}
& \includegraphics[width=0.245\textwidth]{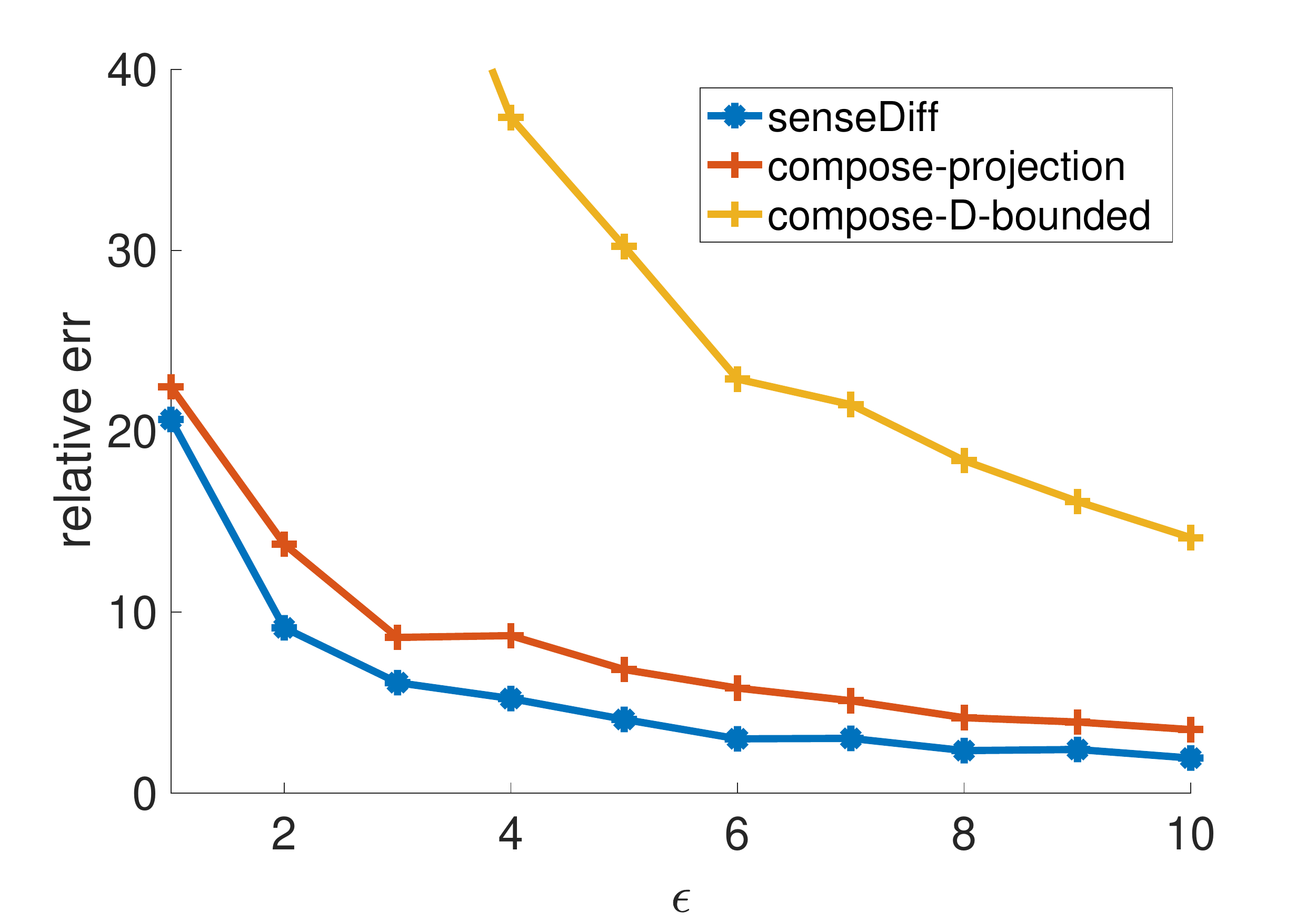}
& \includegraphics[width=0.245\textwidth]{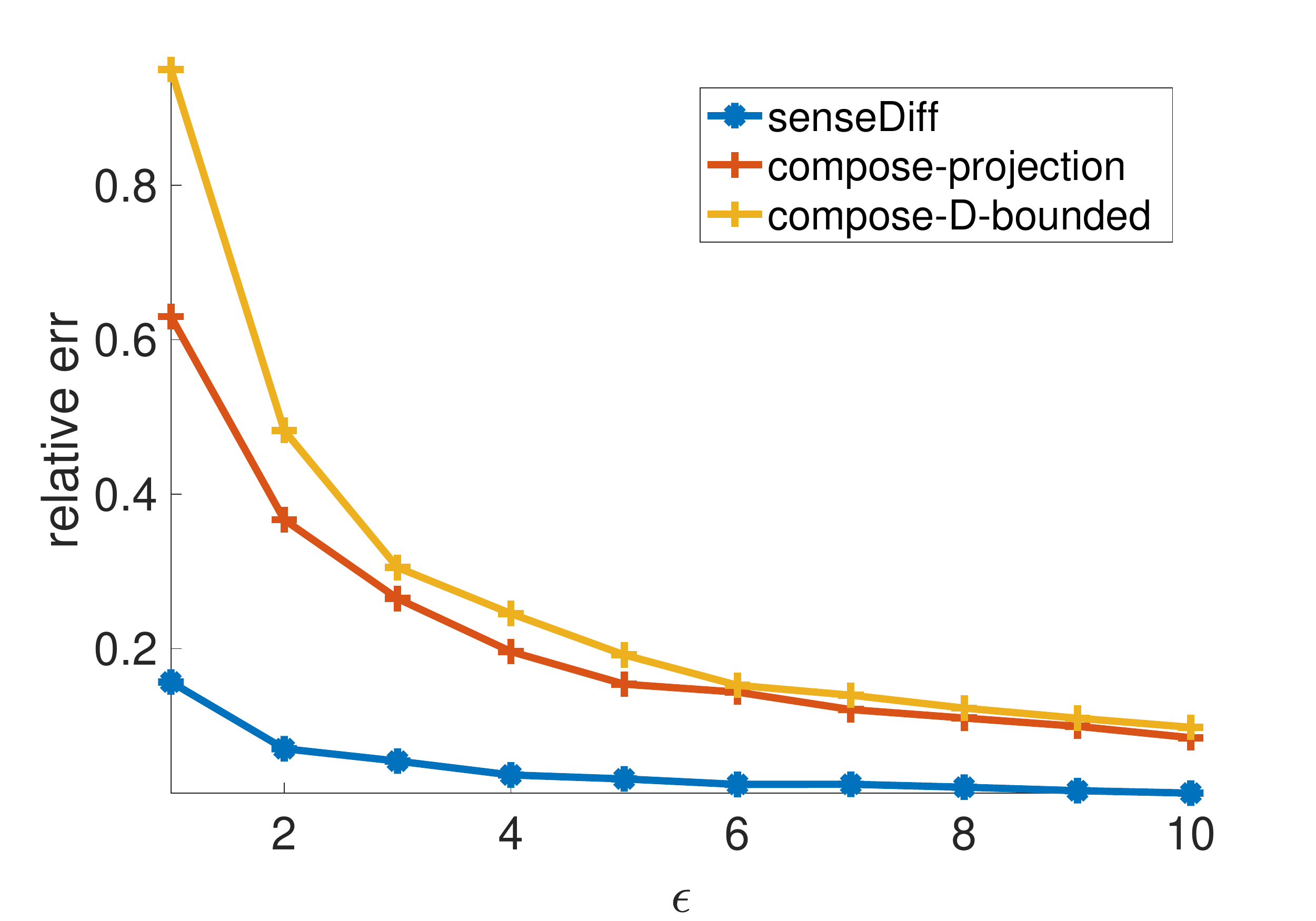}
& \includegraphics[width=0.245\textwidth]{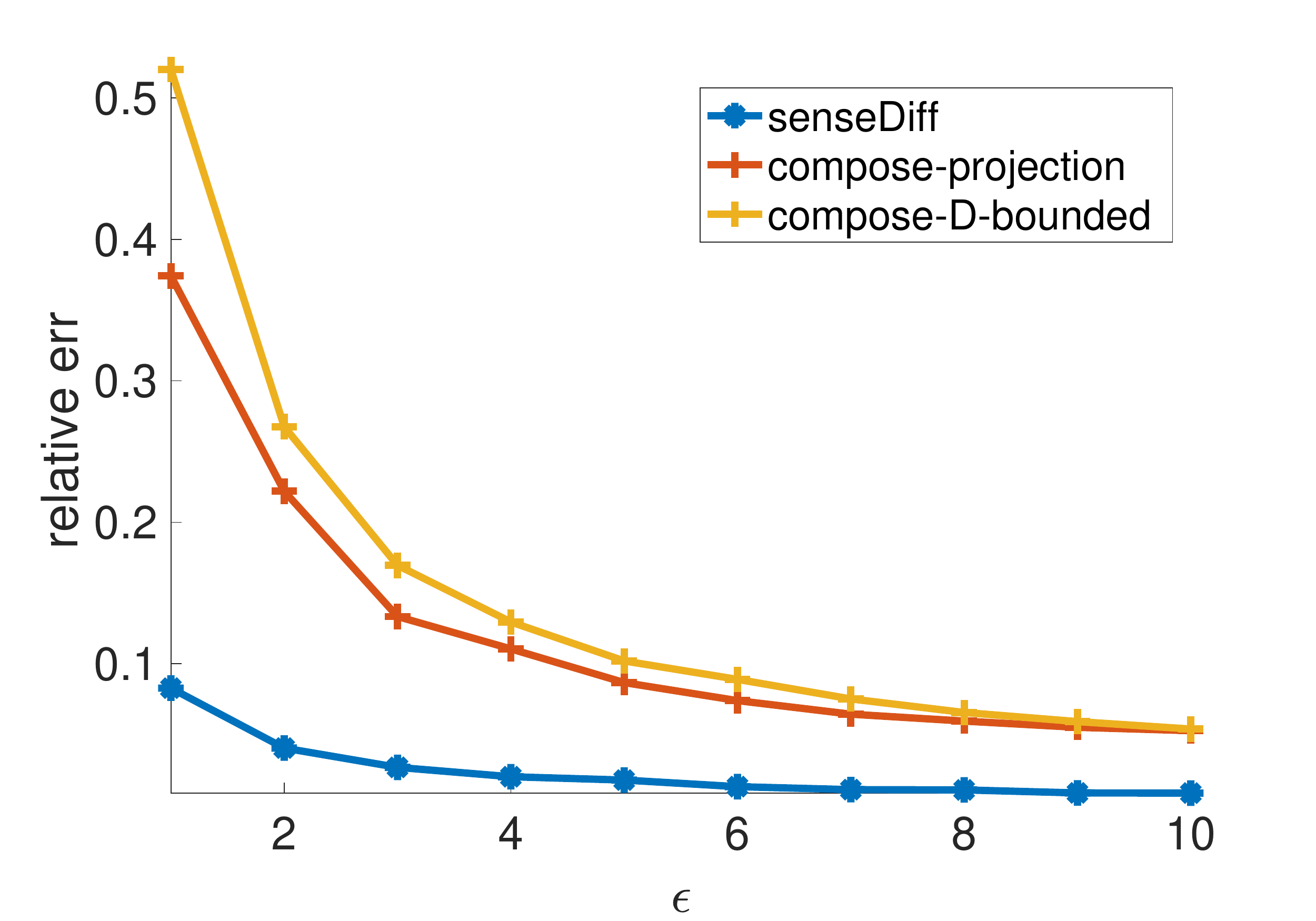}
\\ 
\rotatebox[origin=l]{90}{\qquad Paper citation}
&\includegraphics[width=0.245\textwidth]{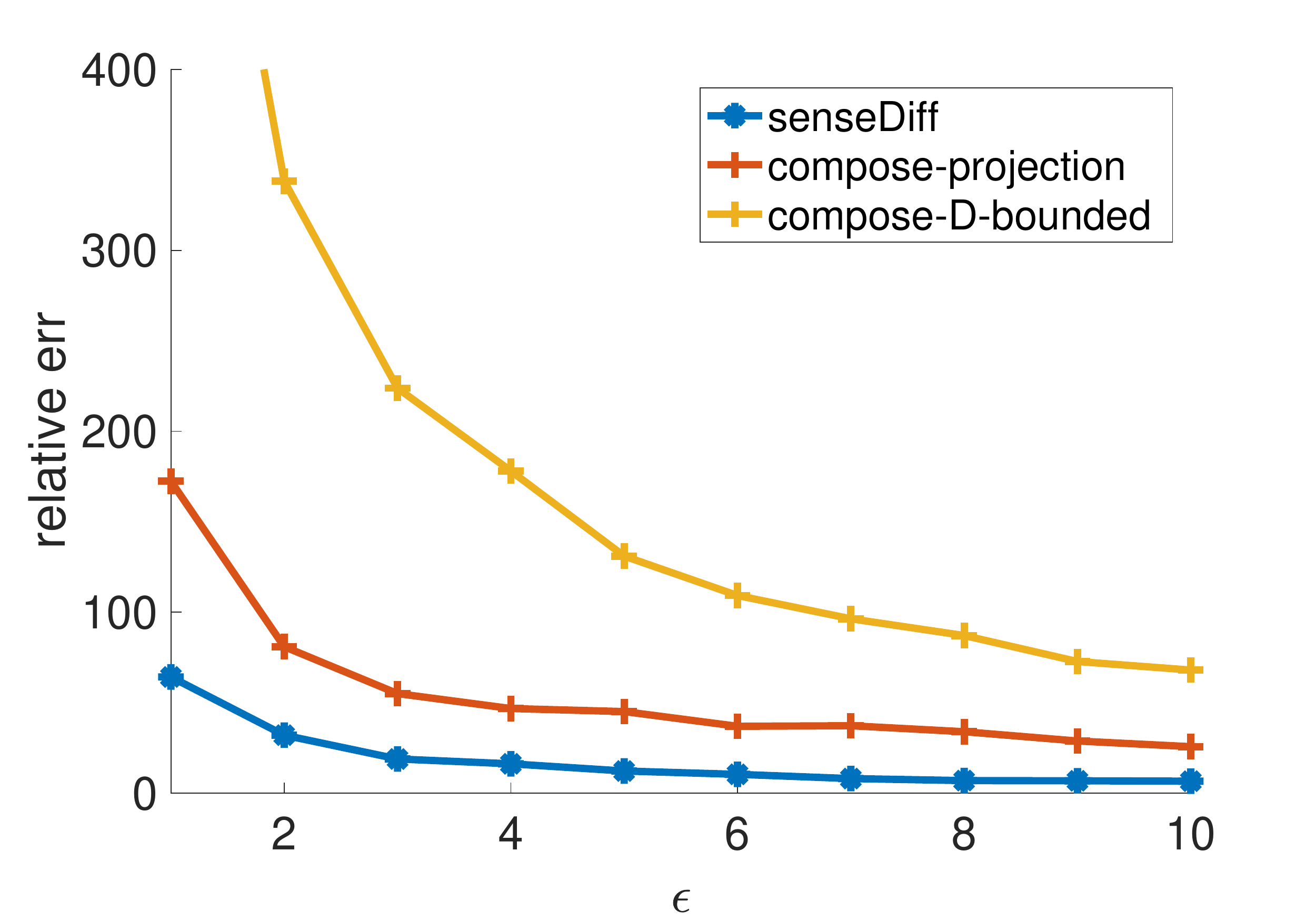}
& \includegraphics[width=0.245\textwidth]{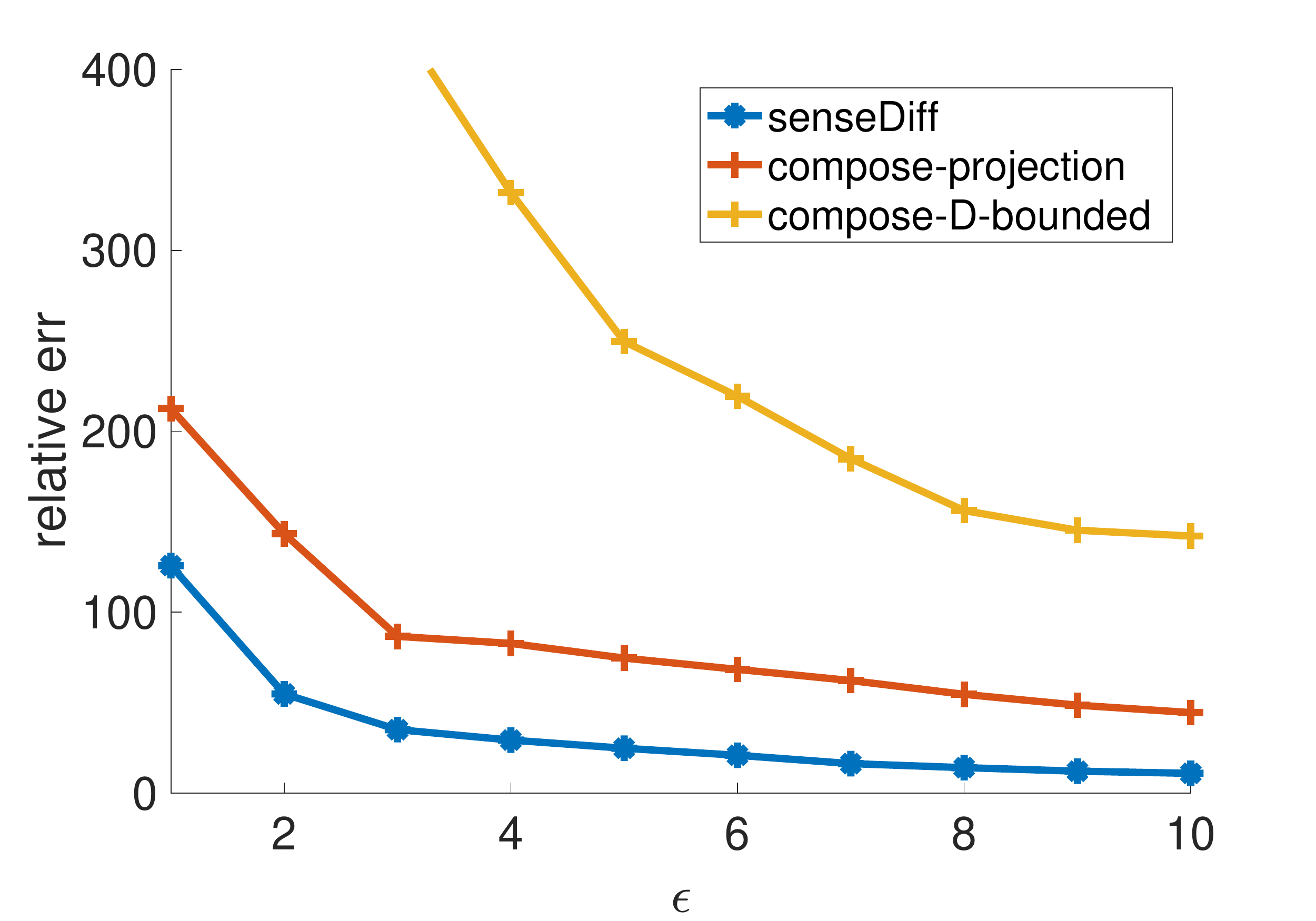}
& \includegraphics[width=0.245\textwidth]{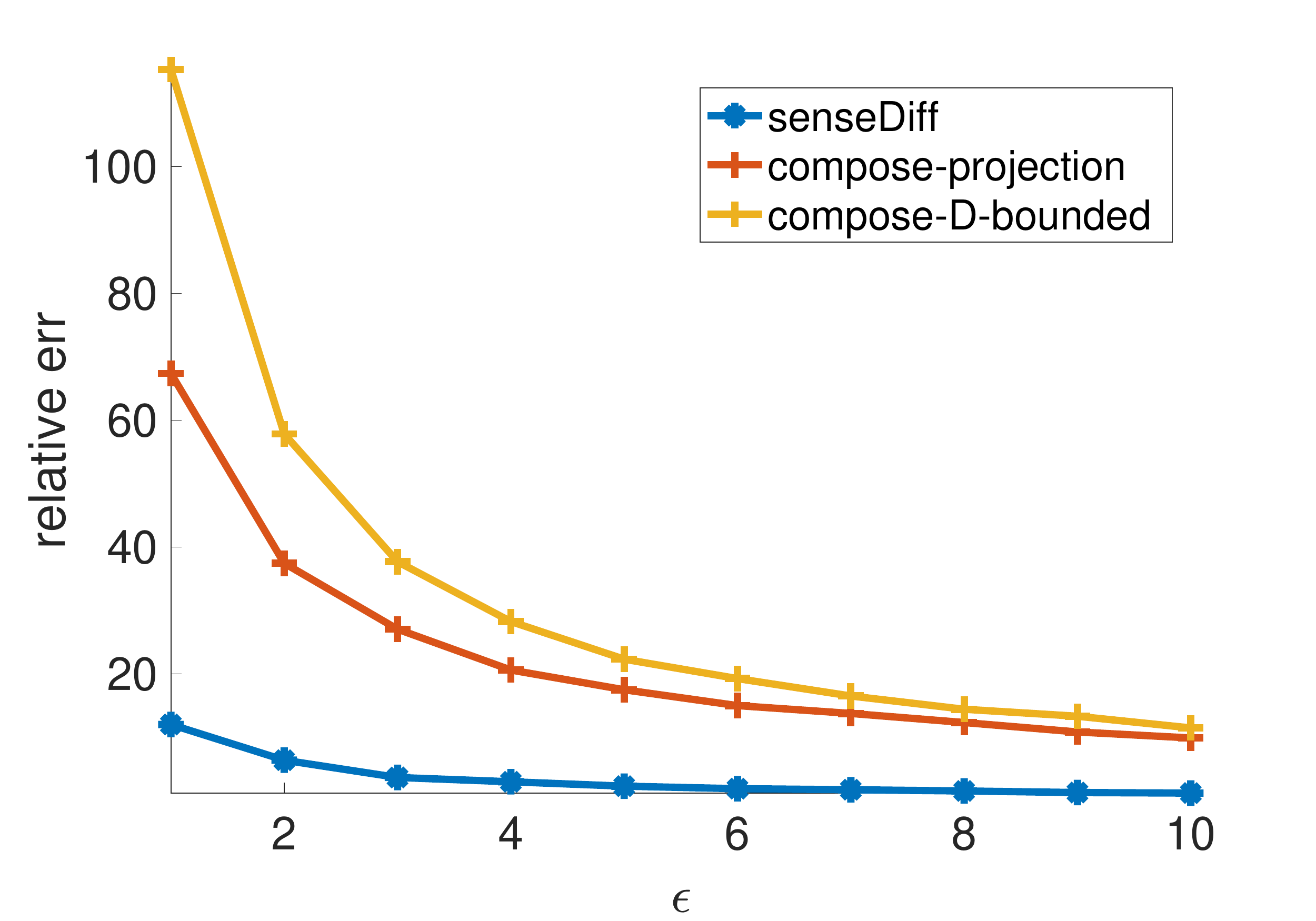}
& \includegraphics[width=0.245\textwidth]{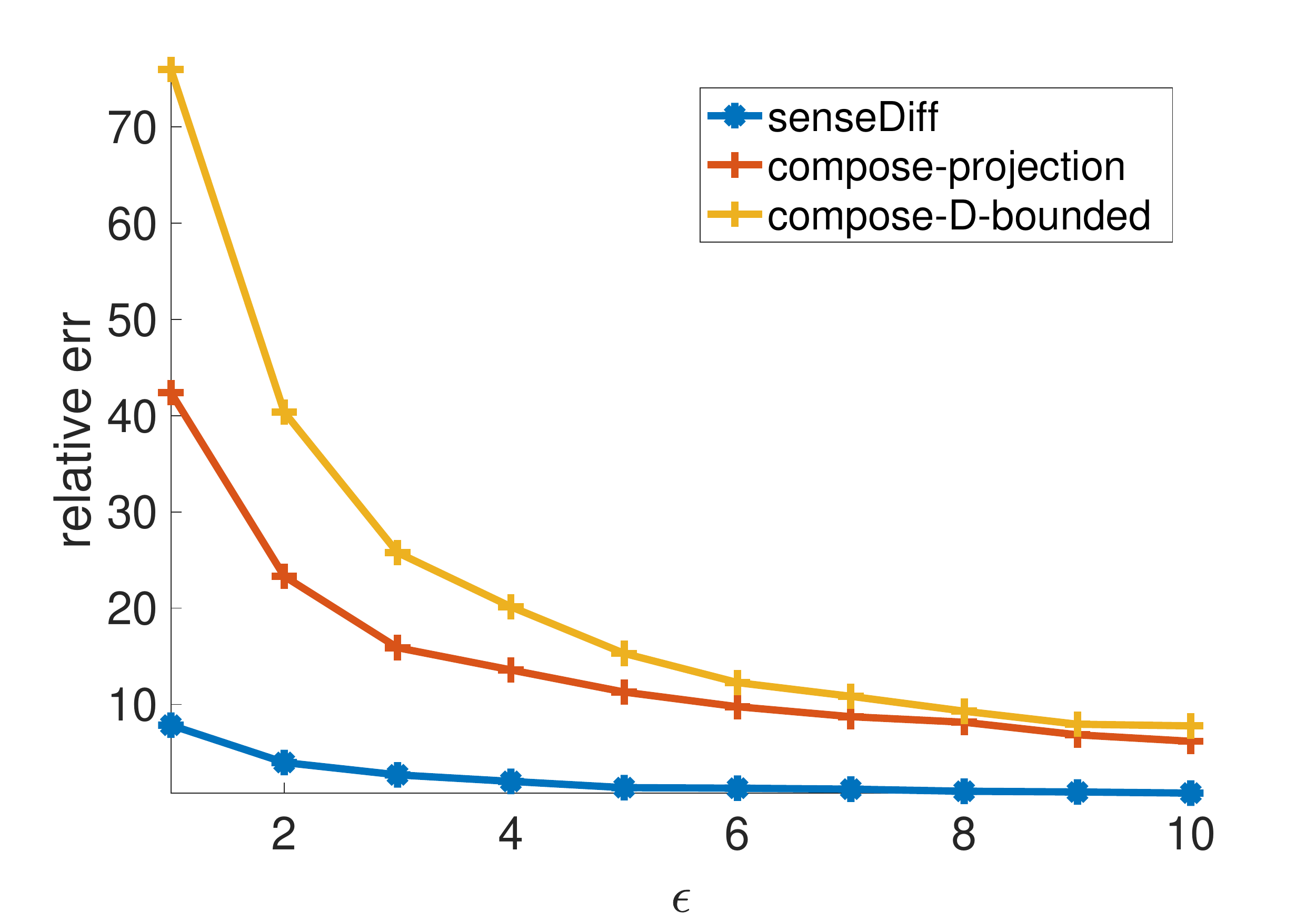}
\\ 
\rotatebox[origin=l]{90}{\qquad \quad Synthetic \rom{1}}
&\includegraphics[width=0.245\textwidth]{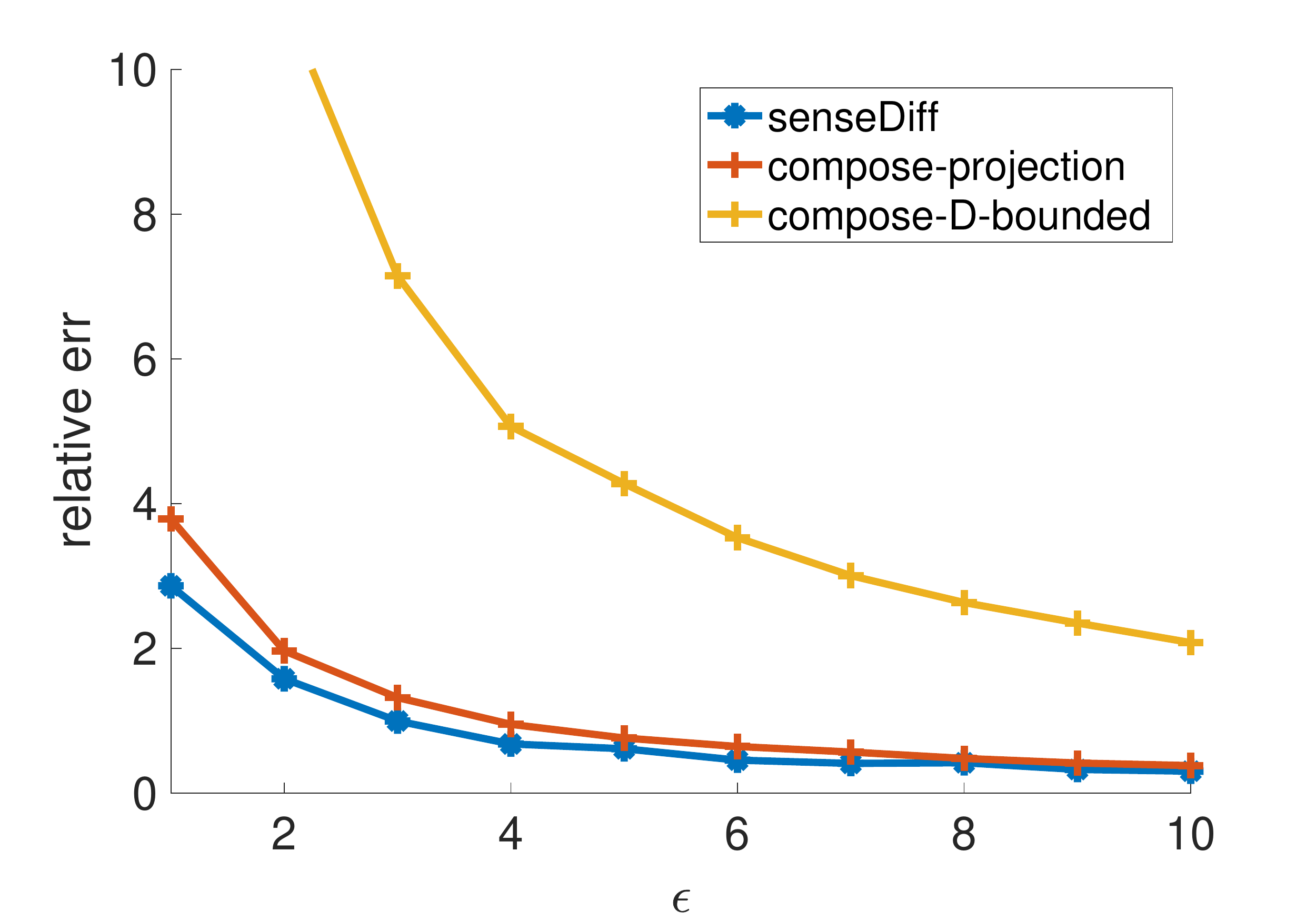}
& \includegraphics[width=0.245\textwidth]{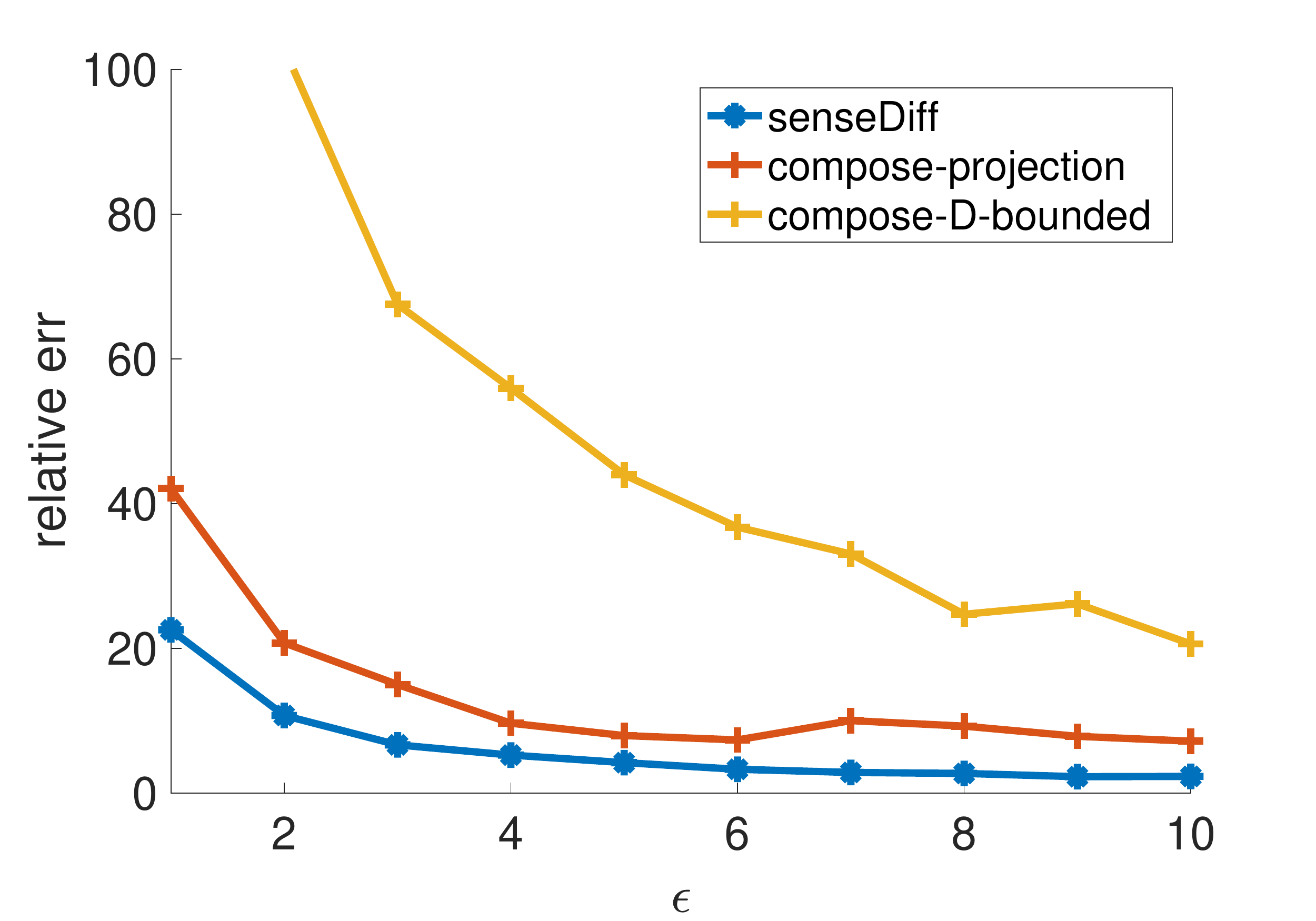}
& \includegraphics[width=0.245\textwidth]{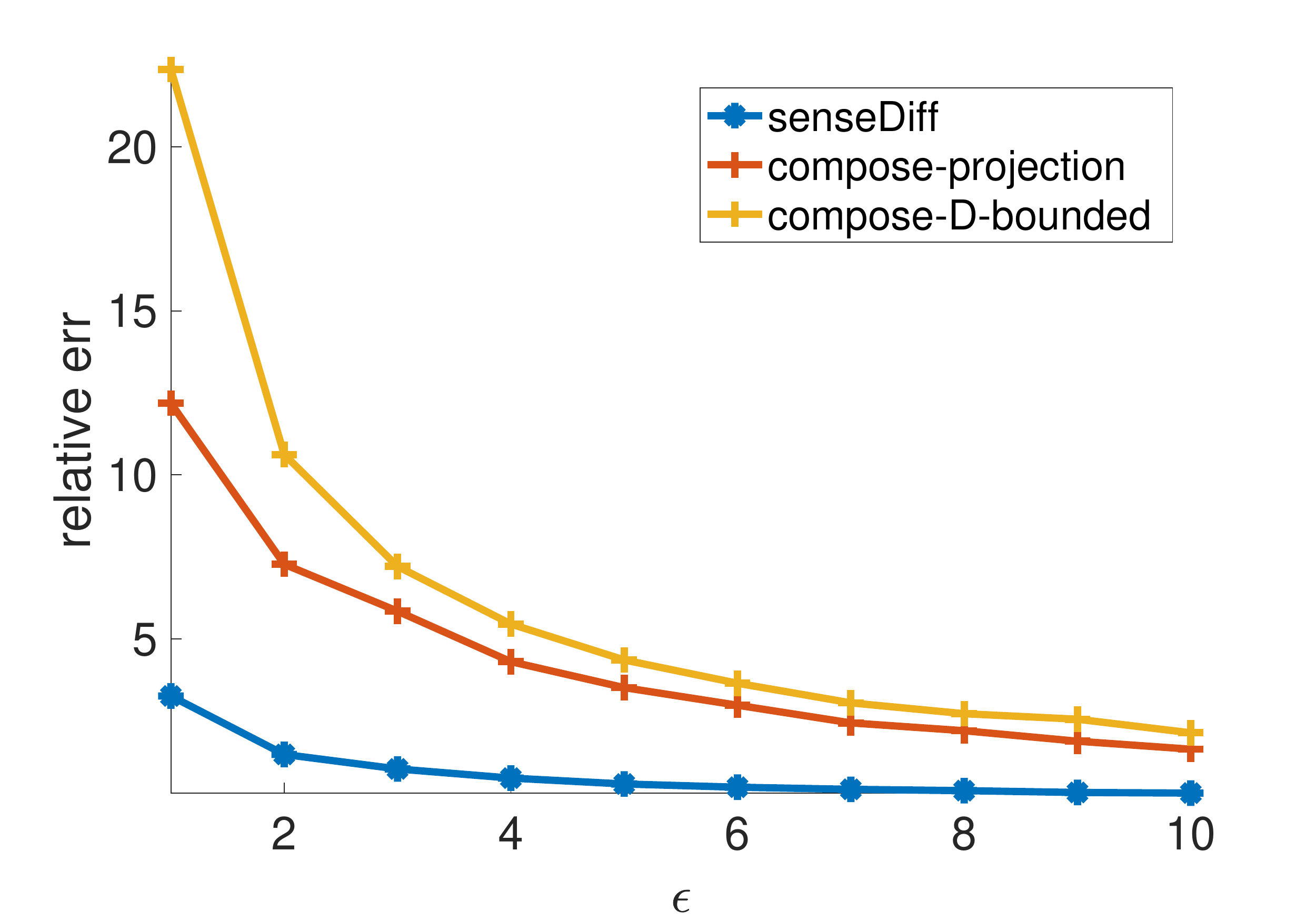}
& \includegraphics[width=0.245\textwidth]{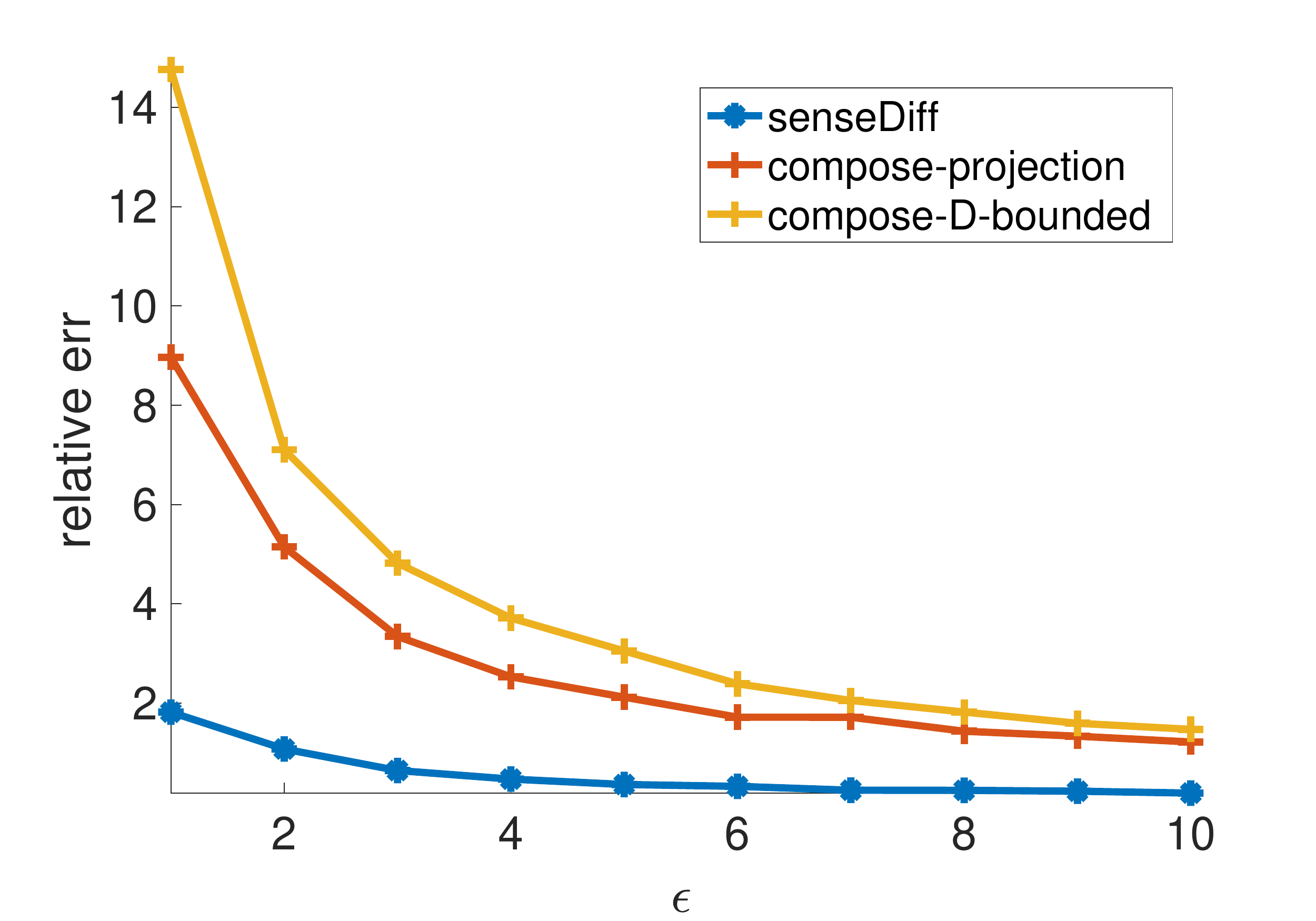}
\\ 
\rotatebox[origin=l]{90}{\qquad \quad Synthetic \rom{2}}
&\includegraphics[width=0.245\textwidth]{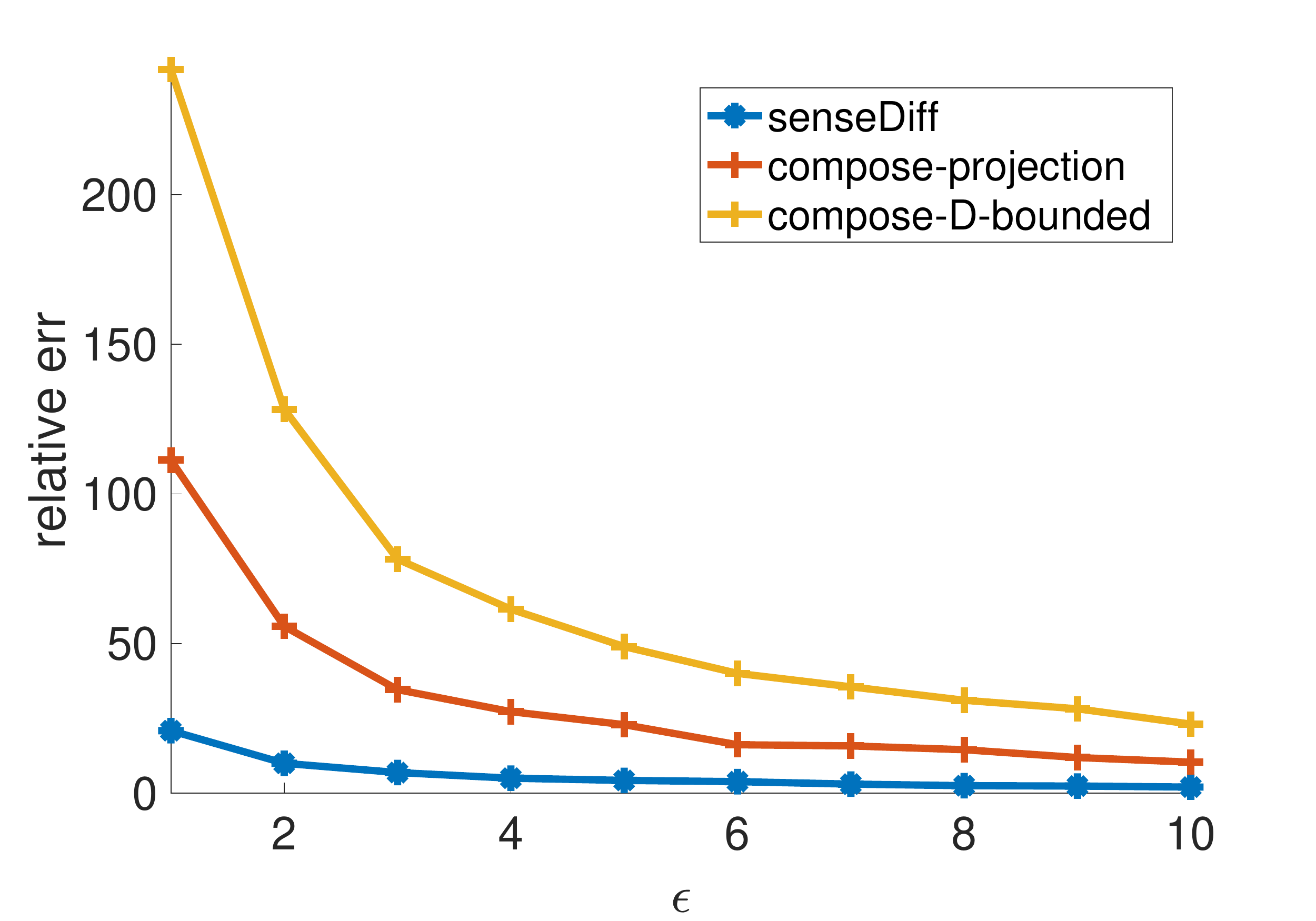}
& \includegraphics[width=0.245\textwidth]{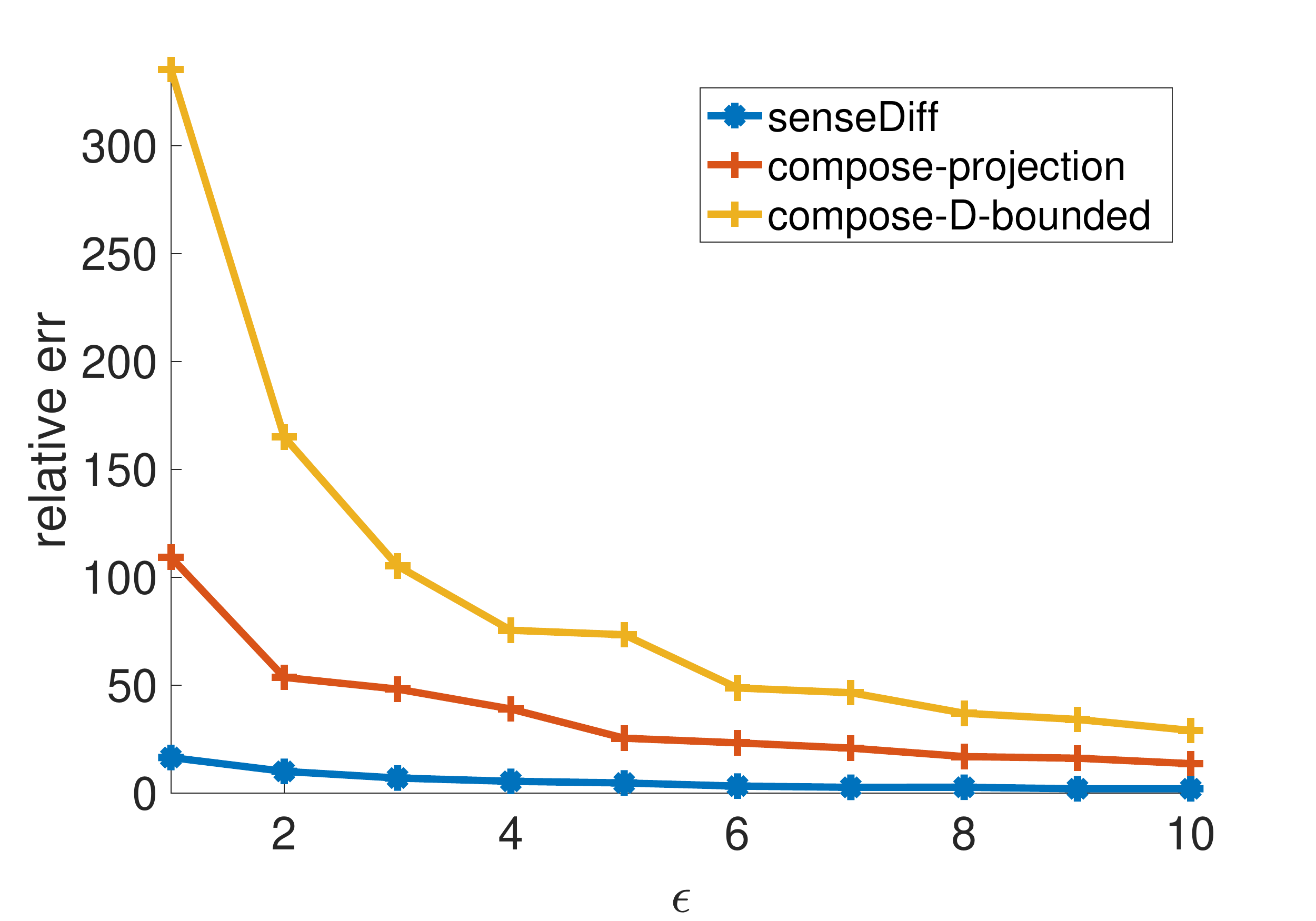}
& \includegraphics[width=0.245\textwidth]{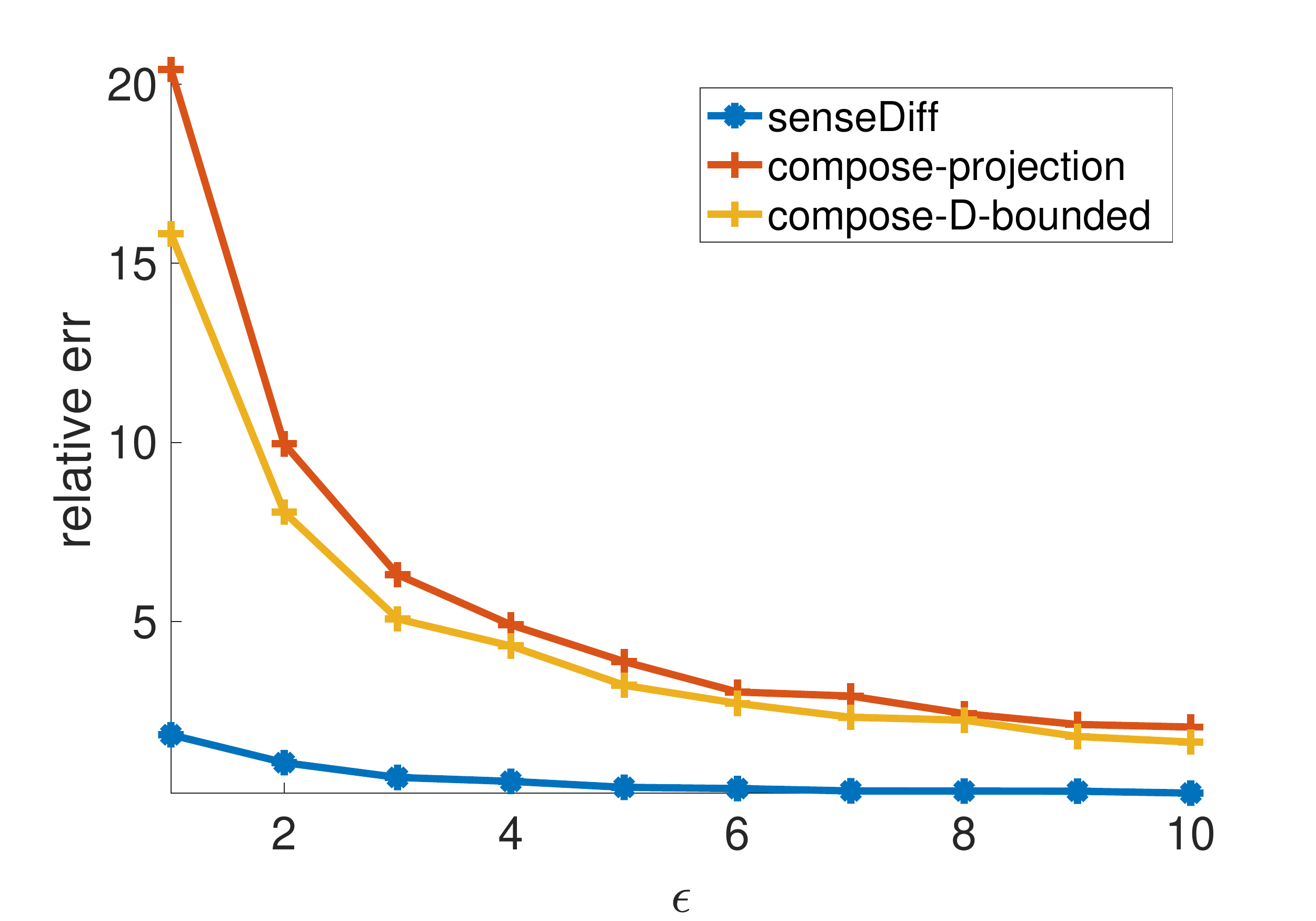}
& \includegraphics[width=0.245\textwidth]{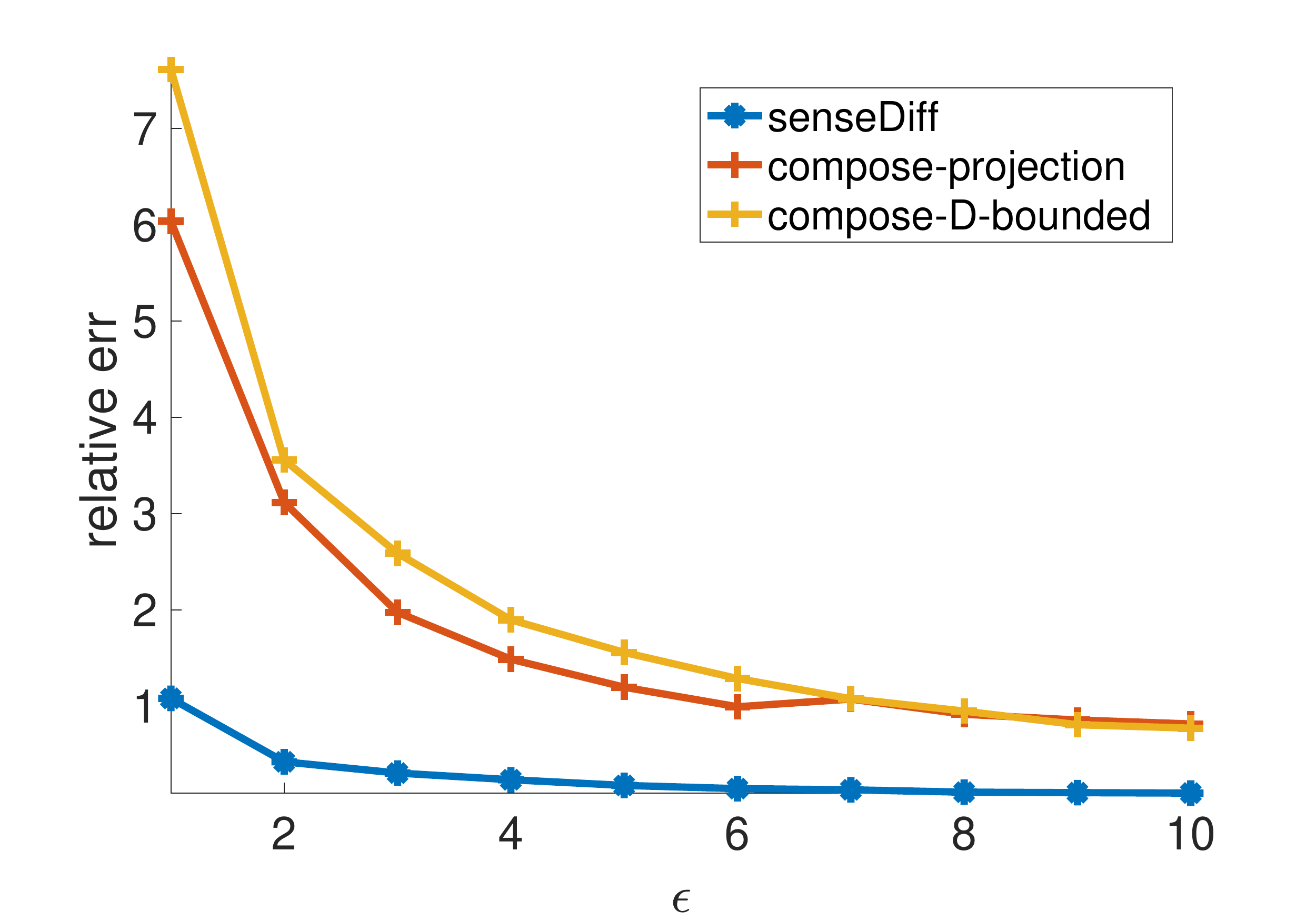}\\
\end{tabularx}}
\captionof{figure}{$L_1$ error vs. privacy parameter $\epsilon$. Publish every $1$ year. Averaged over $100$ runs.}
\label{fig:epsilon}
\end{table*}

\begin{table*}[!t]
\setlength\tabcolsep{0pt} %
\resizebox{\textwidth}{!}{
\begin{tabularx}{\textwidth}{c c c c c}
& Directed, high-degree & Undirected, high-degree & Directed, edge & Undirected, edge \\ 
\rotatebox{90}{\quad HIV transmission}
&\includegraphics[width=0.245\textwidth]{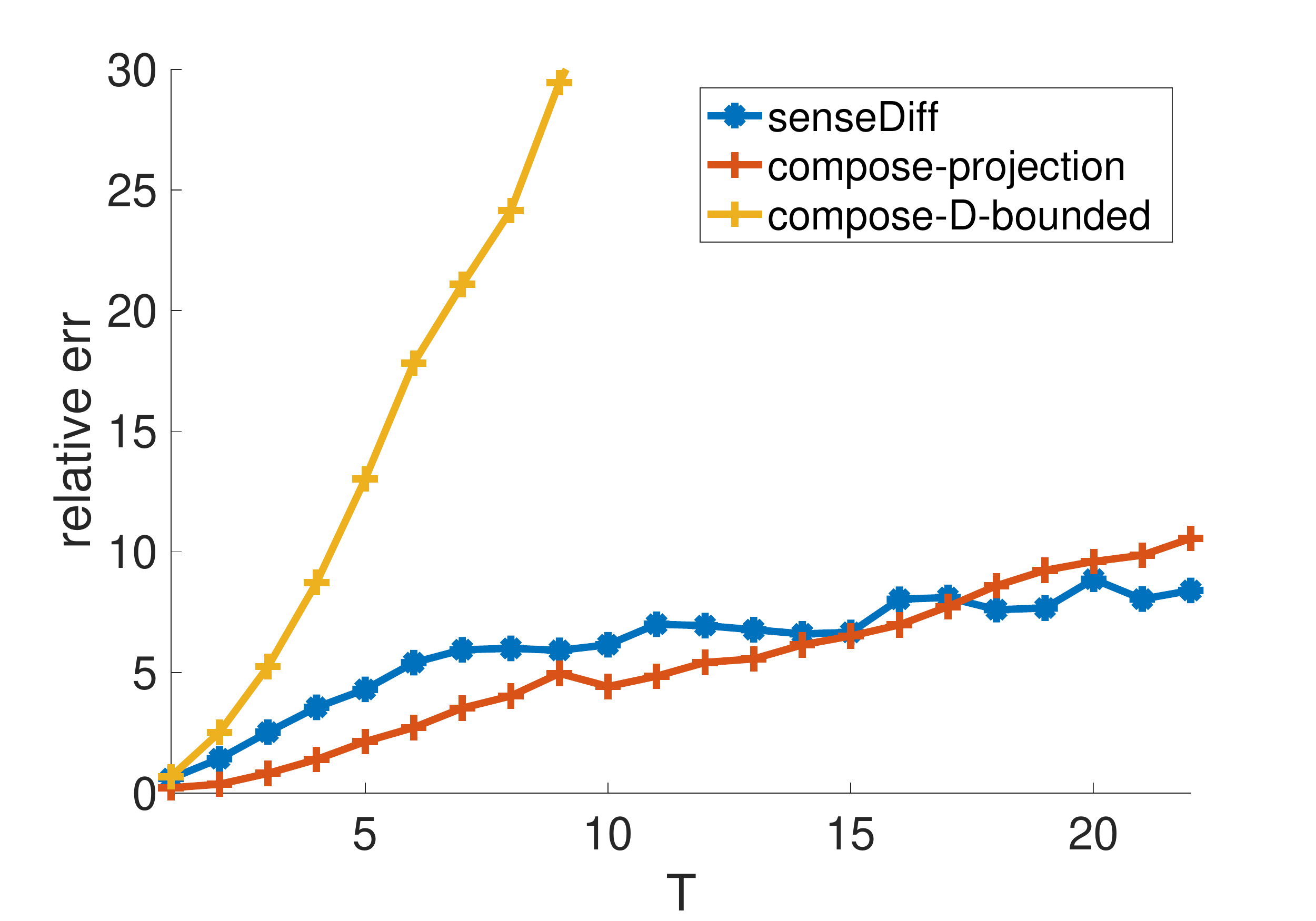}
& \includegraphics[width=0.245\textwidth]{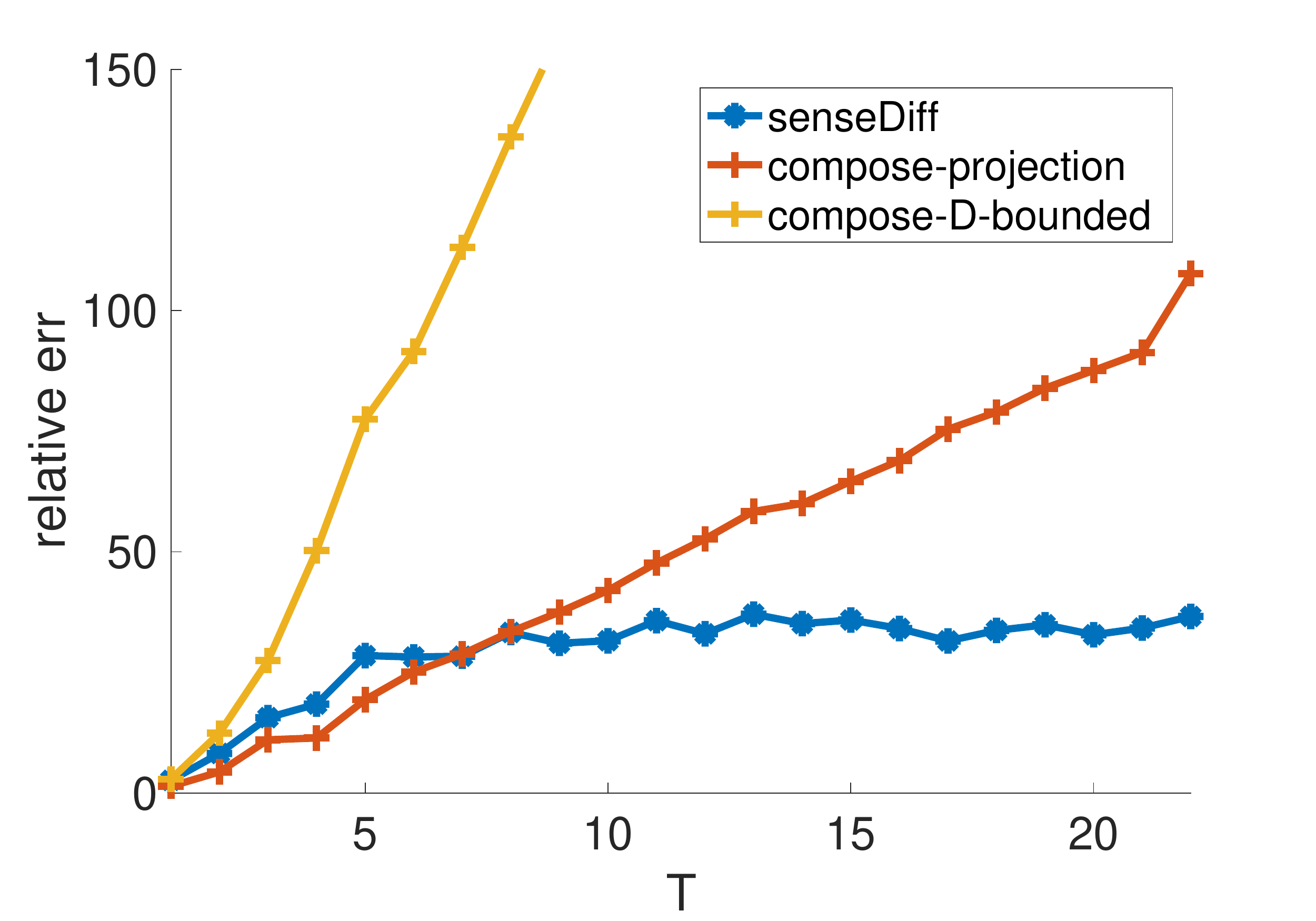}
& \includegraphics[width=0.245\textwidth]{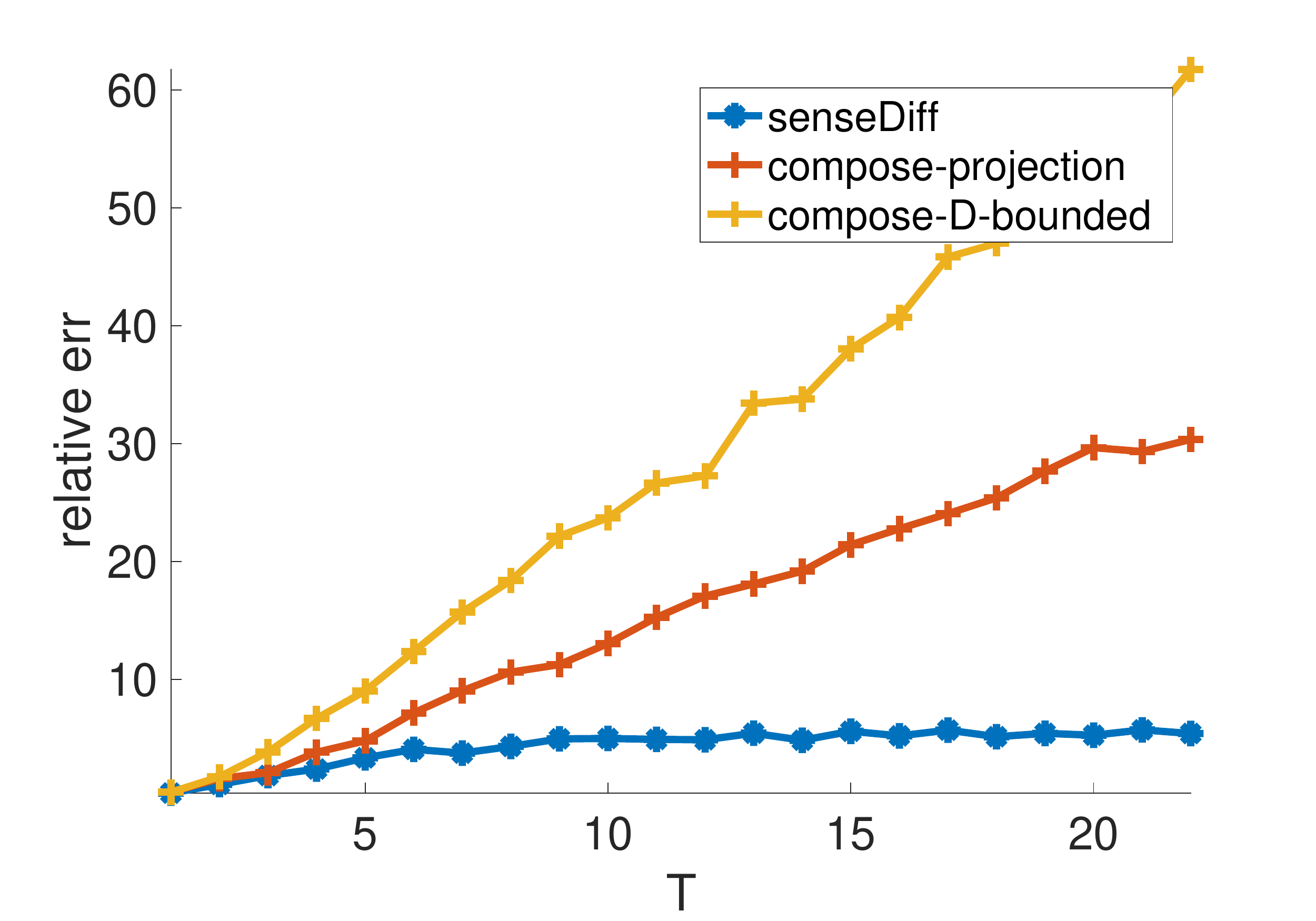}
& \includegraphics[width=0.245\textwidth]{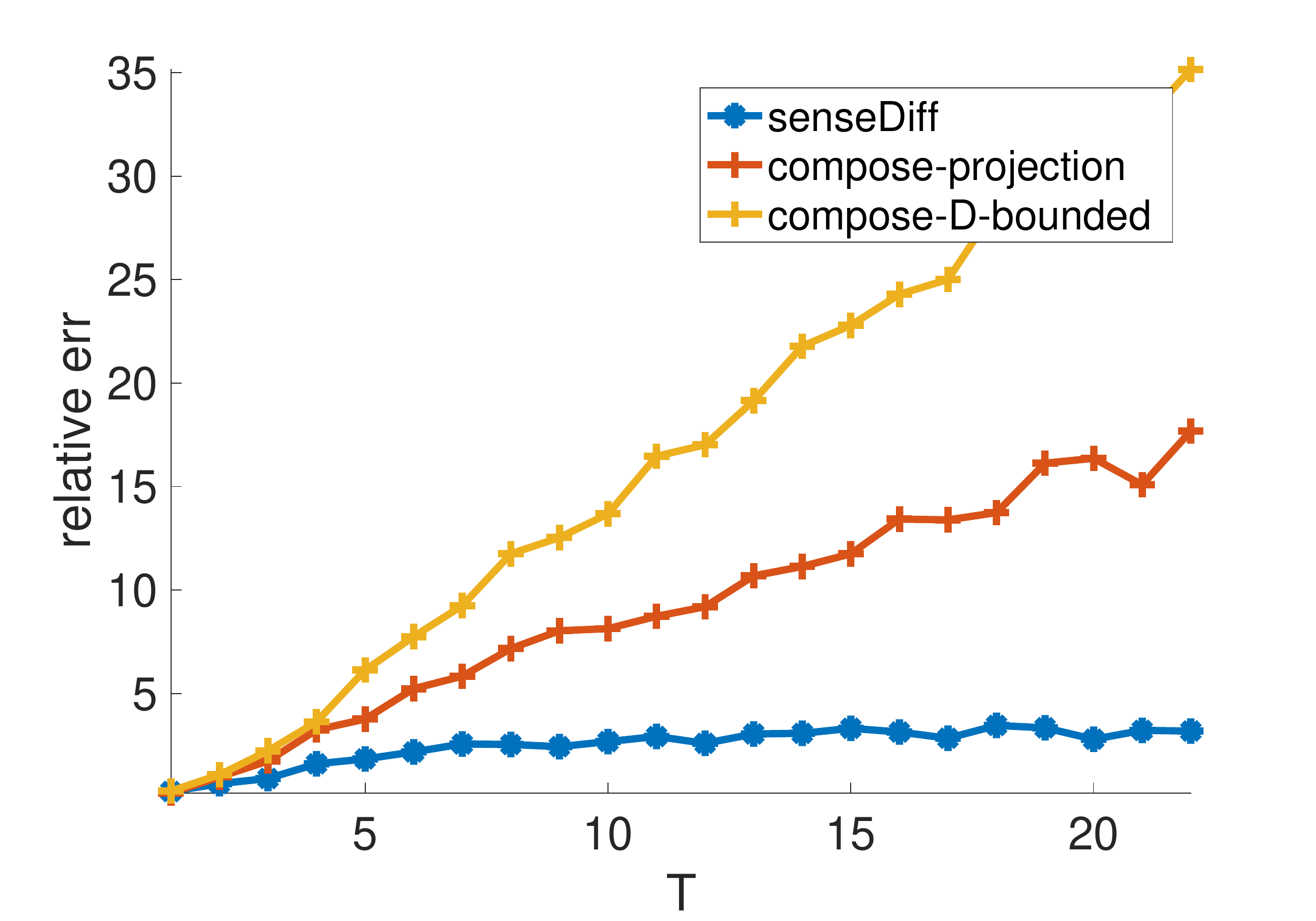}
\\ 
\rotatebox{90}{\qquad Patent citation}
&\includegraphics[width=0.245\textwidth]{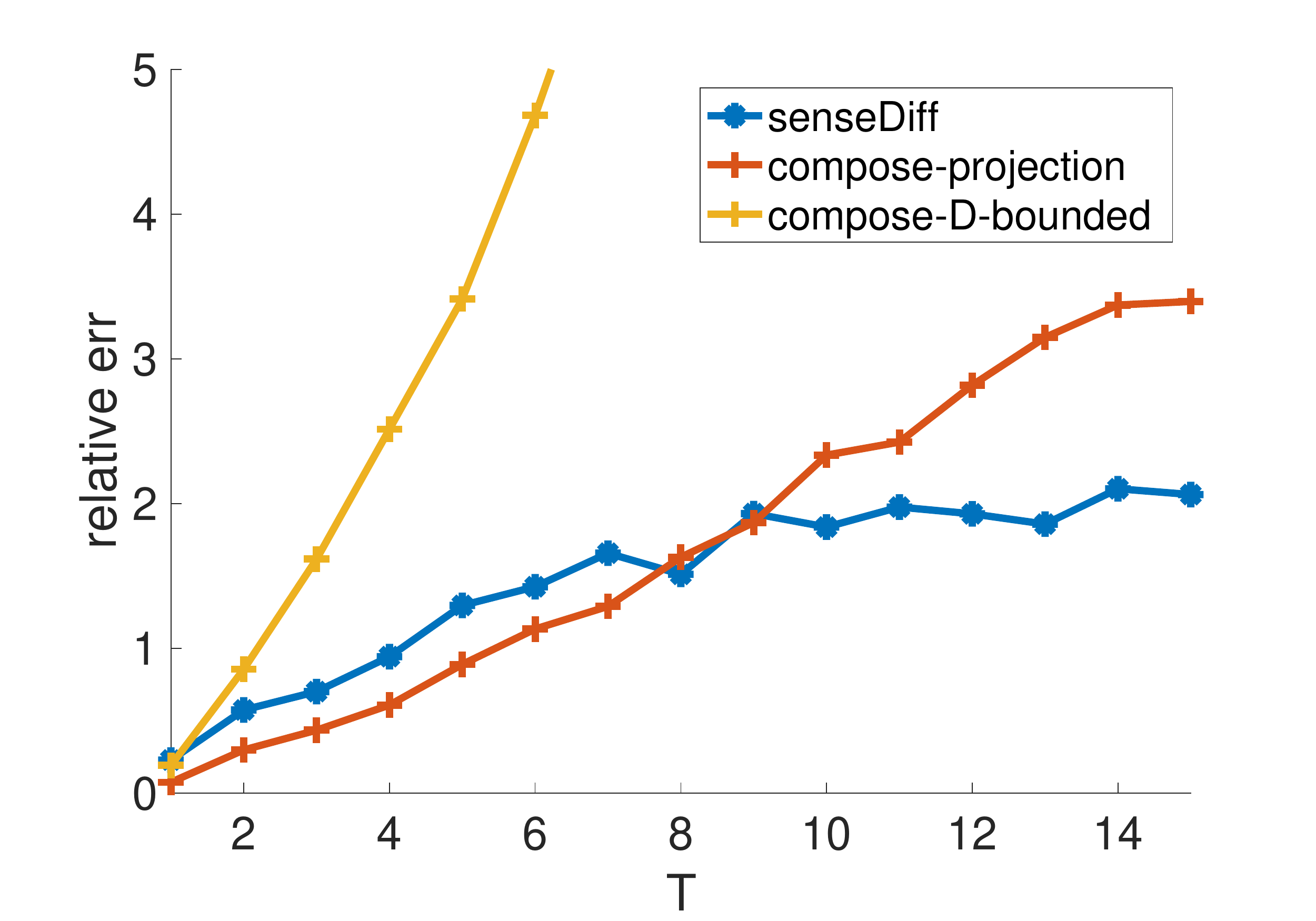}
& \includegraphics[width=0.245\textwidth]{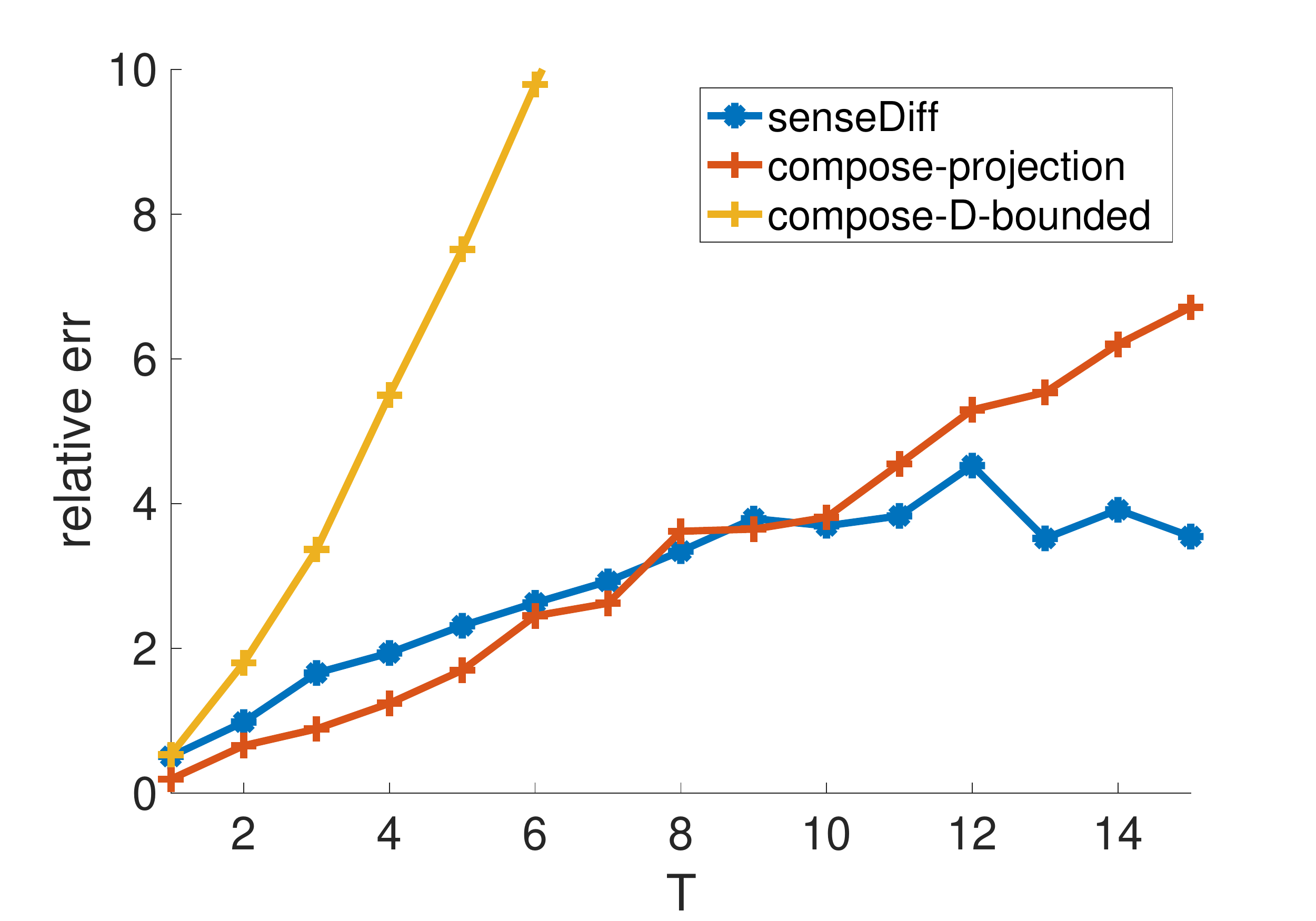}
& \includegraphics[width=0.245\textwidth]{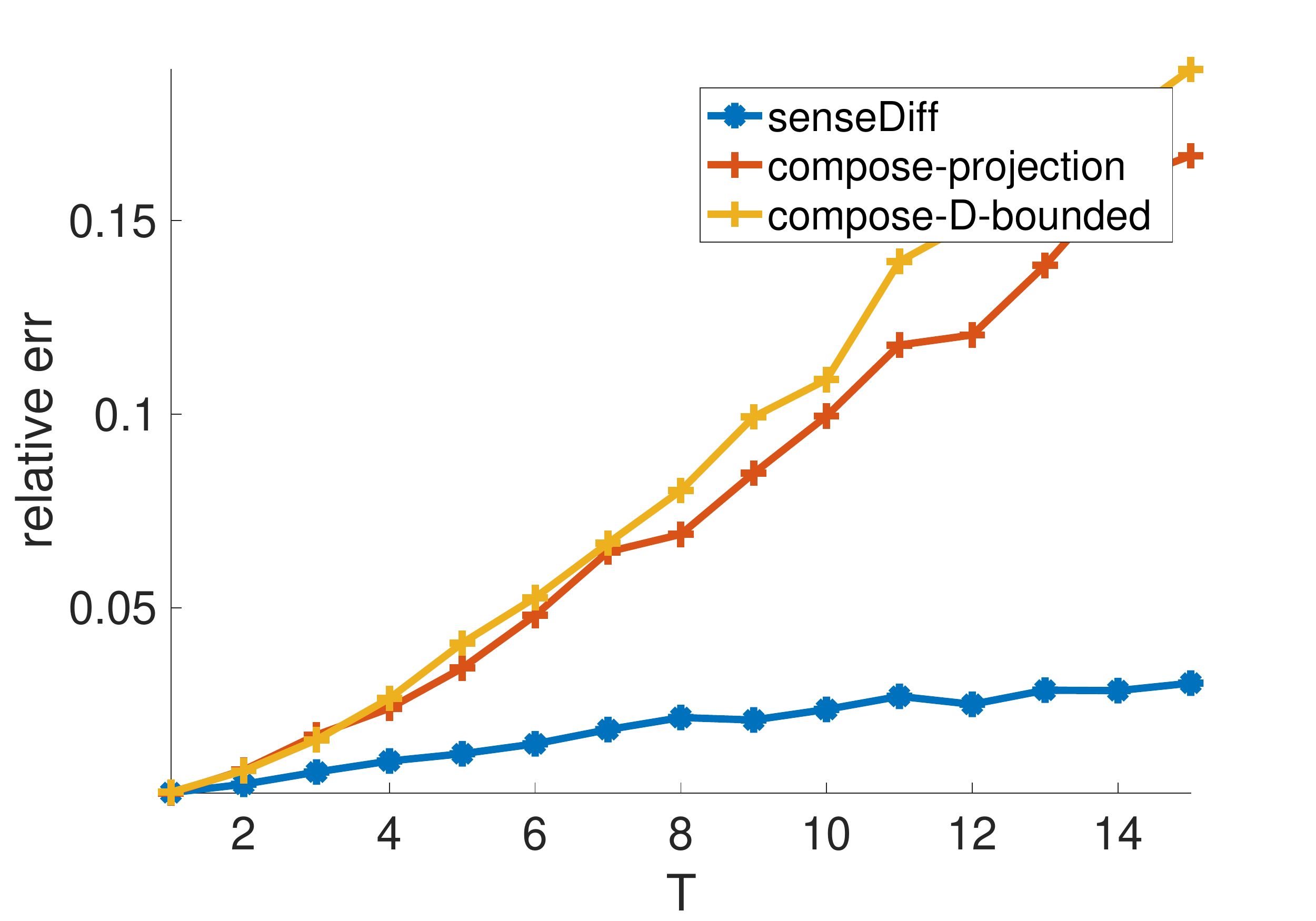}
& \includegraphics[width=0.245\textwidth]{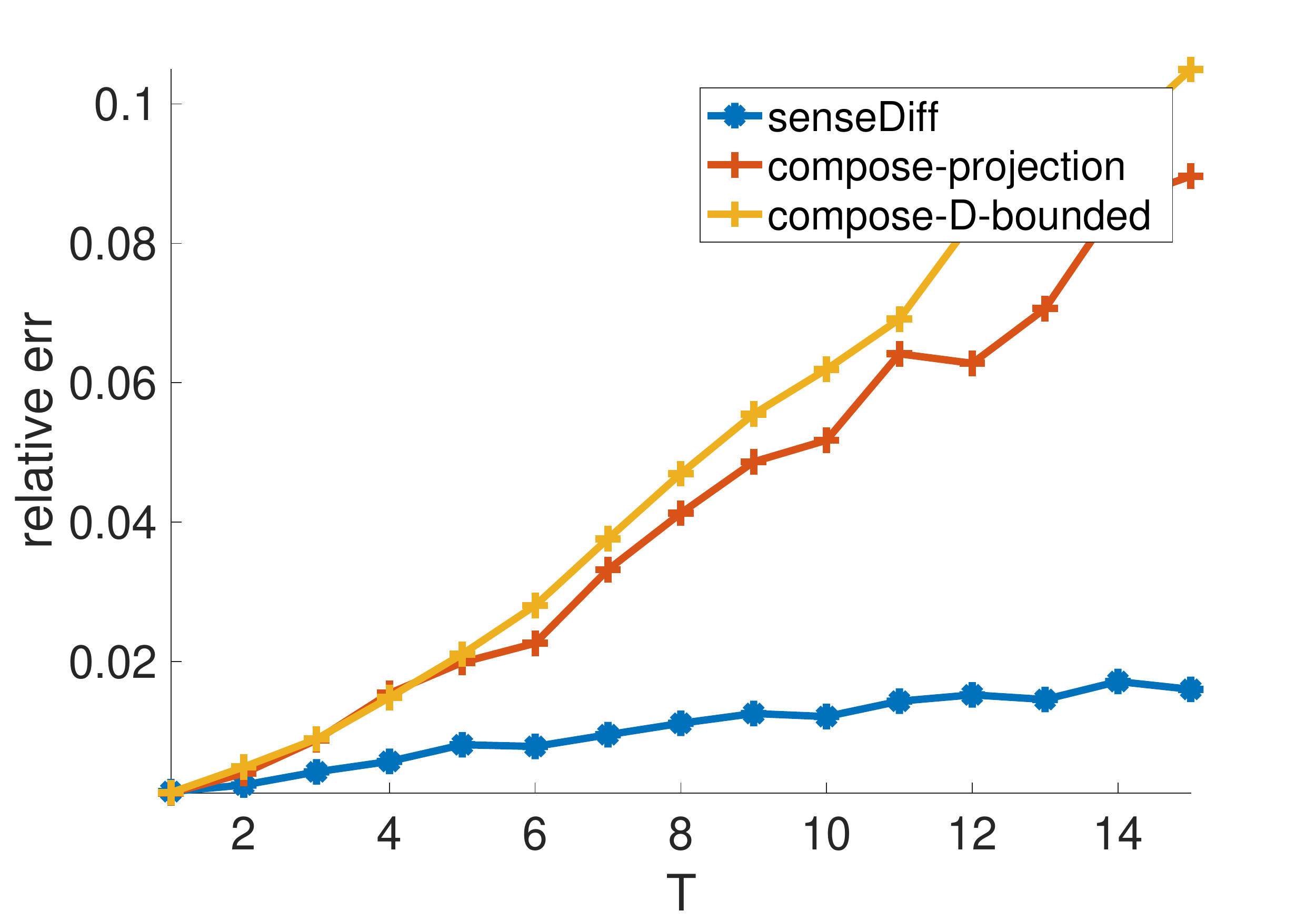}
\\ 
\rotatebox{90}{\qquad Paper citation}
&\includegraphics[width=0.245\textwidth]{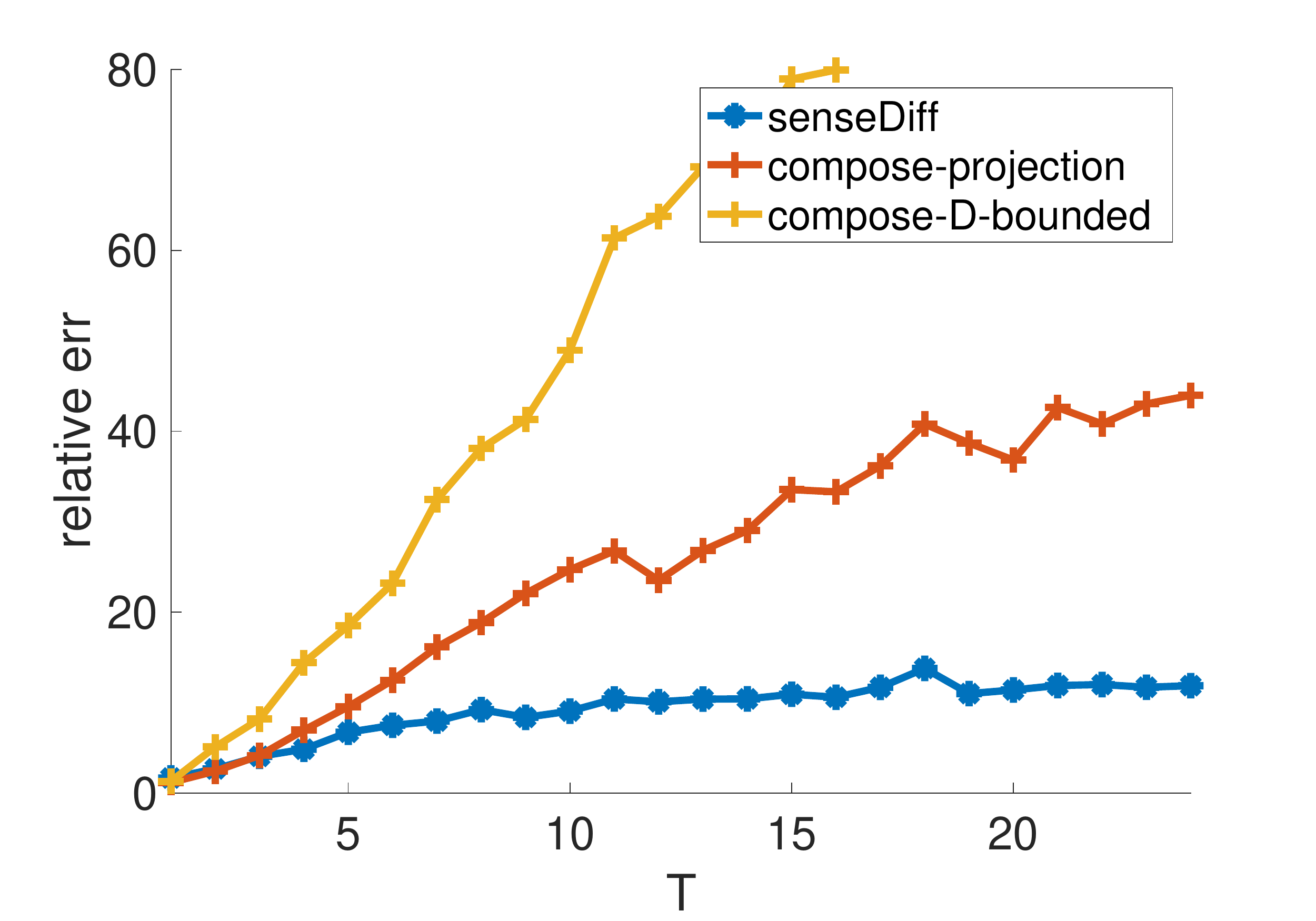}
& \includegraphics[width=0.245\textwidth]{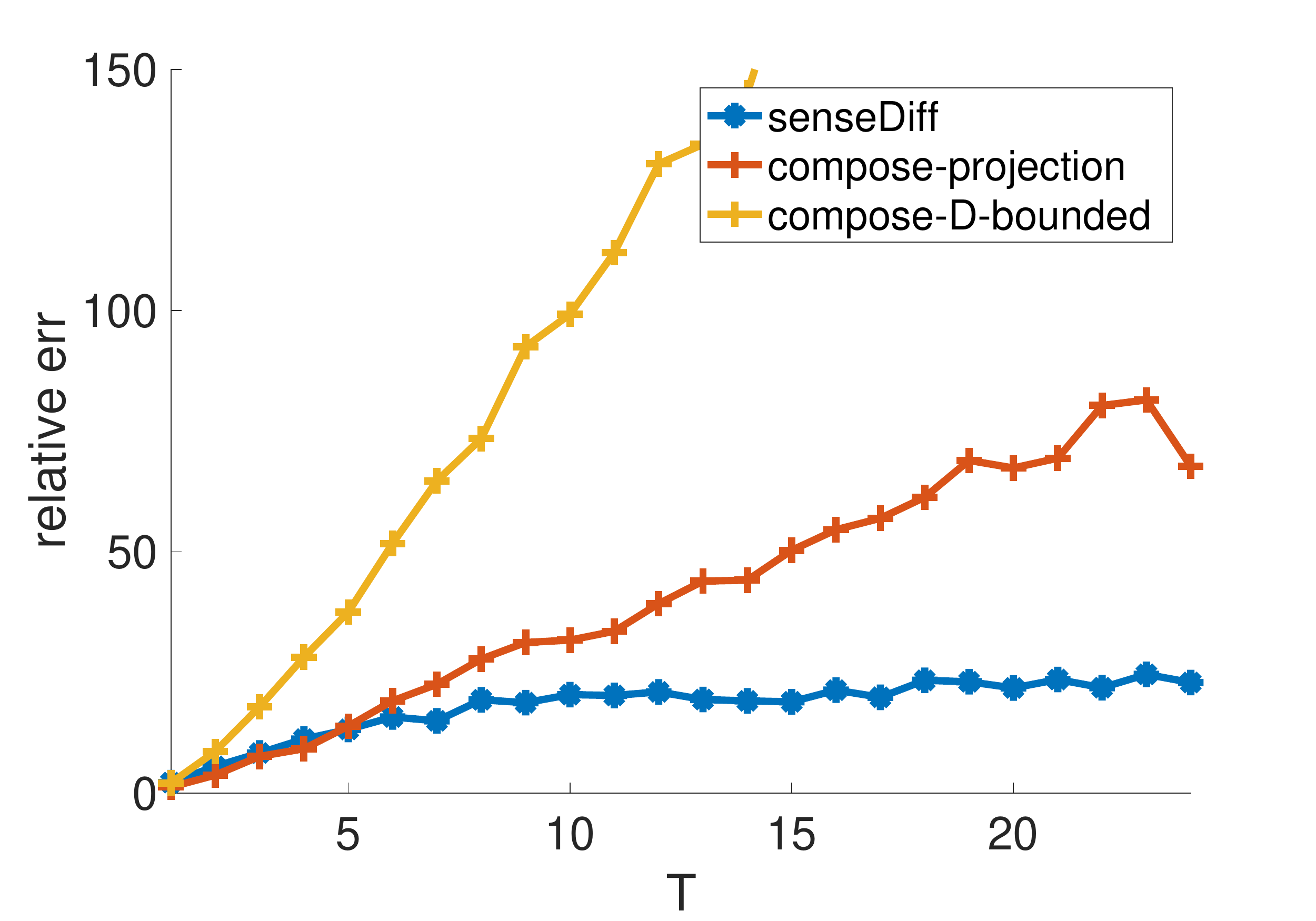}
& \includegraphics[width=0.245\textwidth]{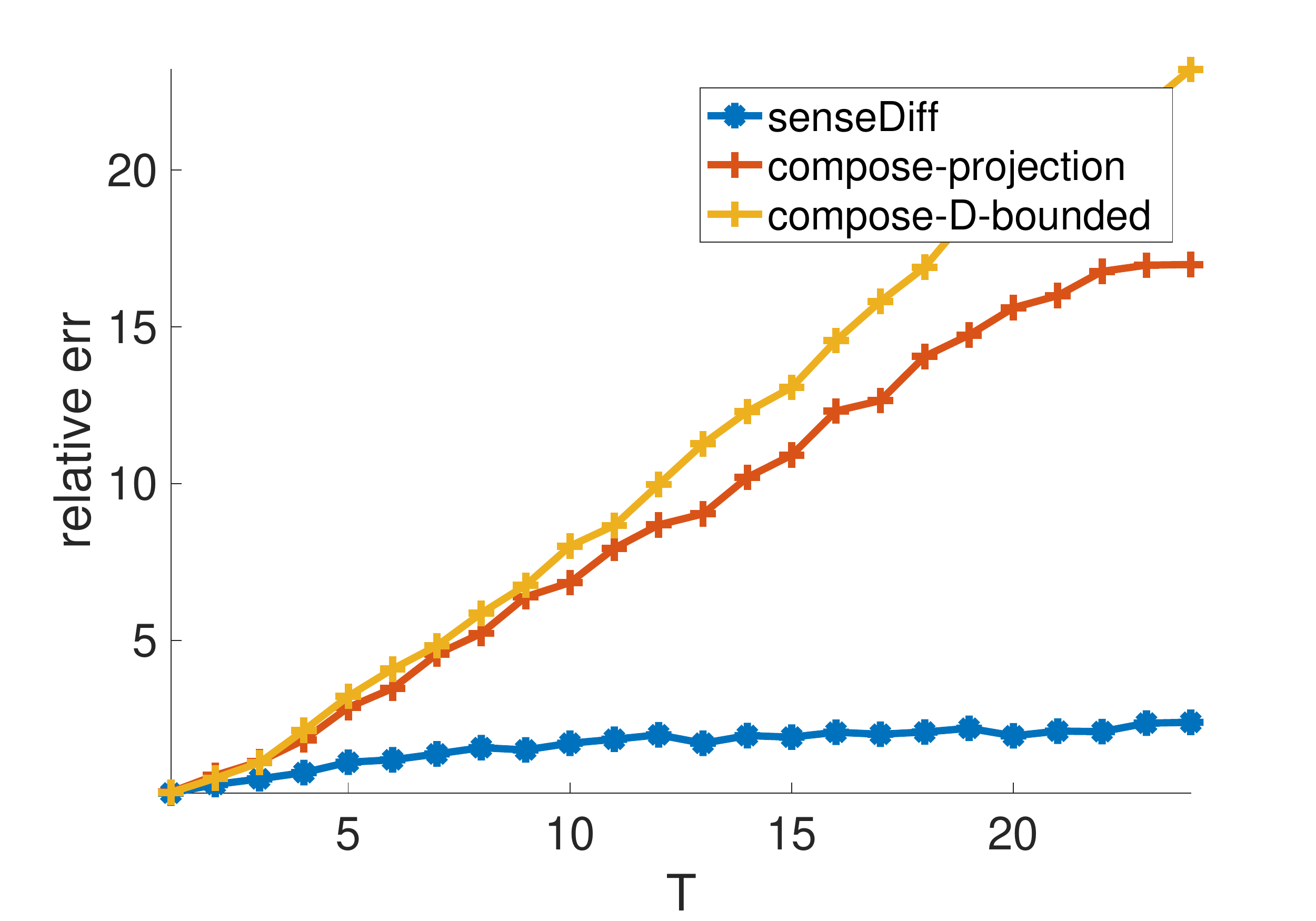}
& \includegraphics[width=0.245\textwidth]{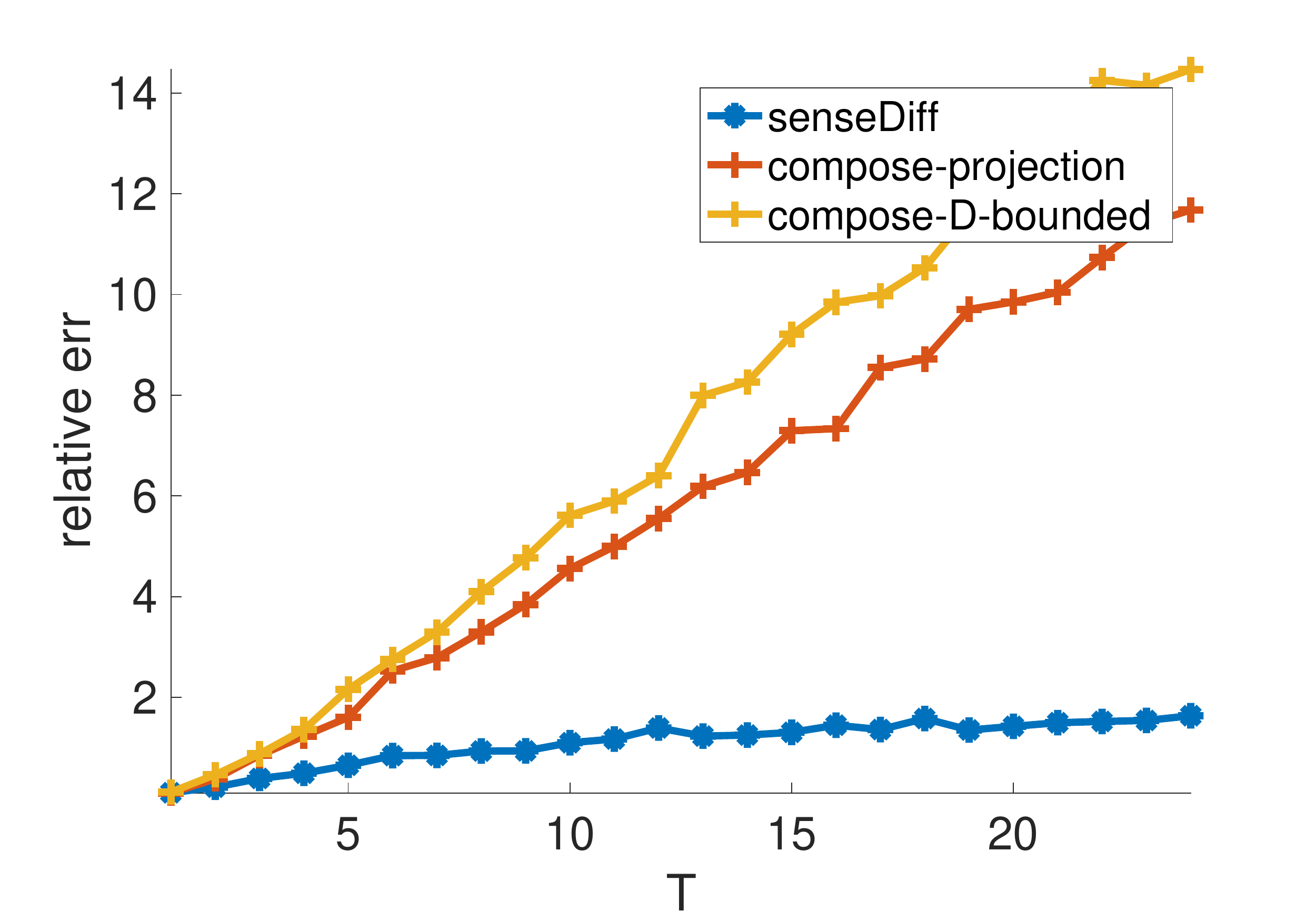}
\\ 
\rotatebox{90}{\qquad \quad Synthetic \rom{1}}
&\includegraphics[width=0.245\textwidth]{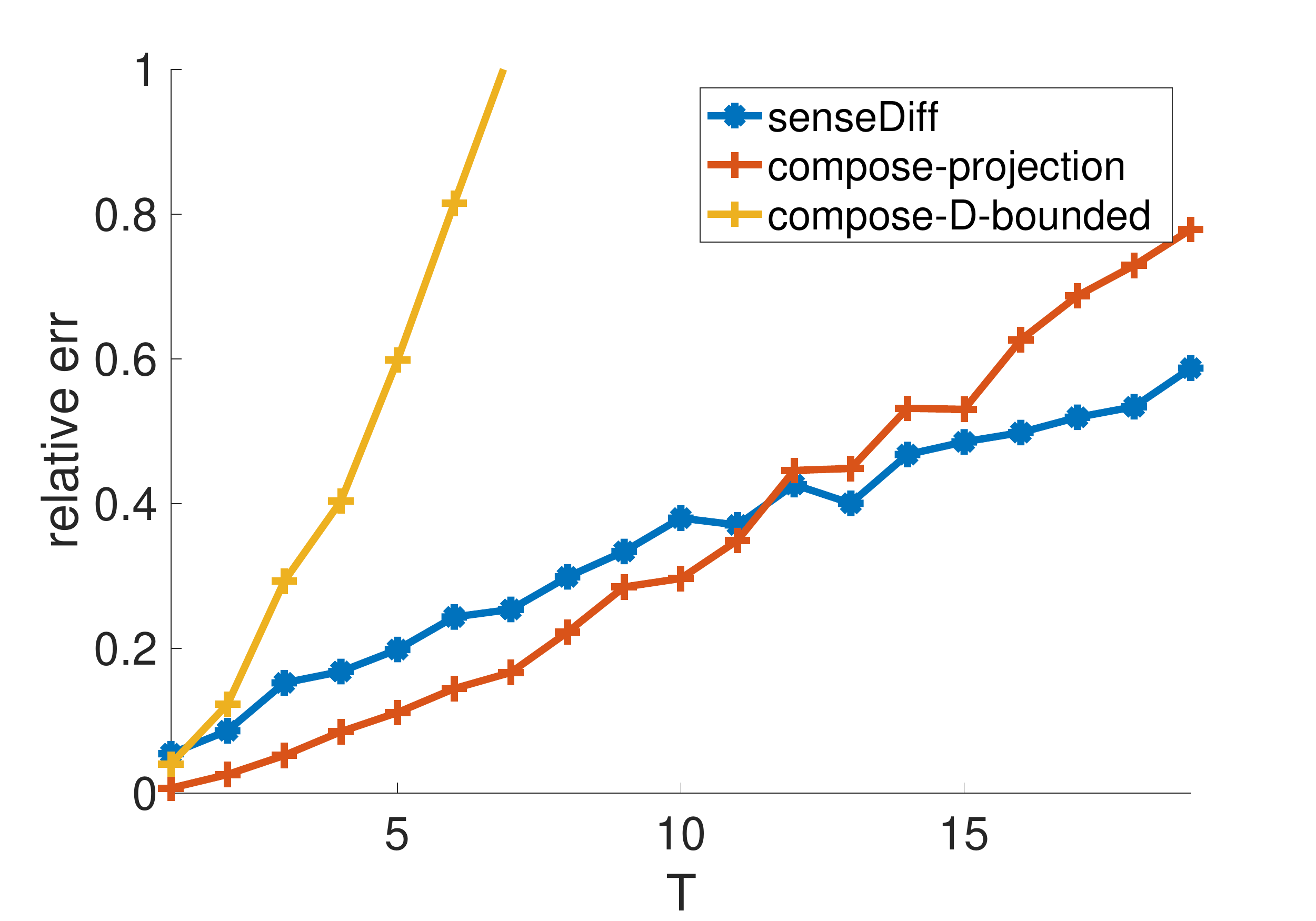}
& \includegraphics[width=0.245\textwidth]{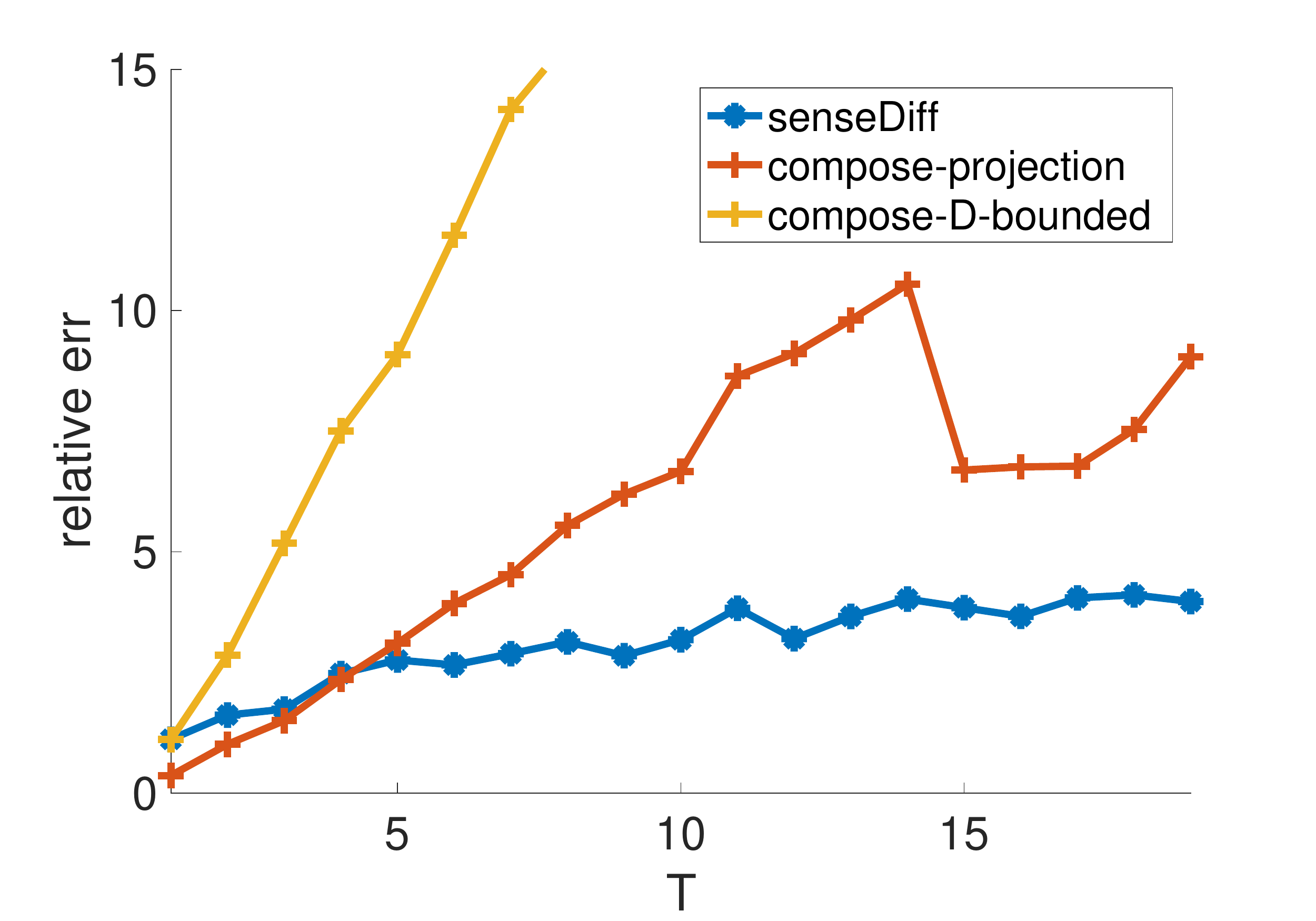}
& \includegraphics[width=0.245\textwidth]{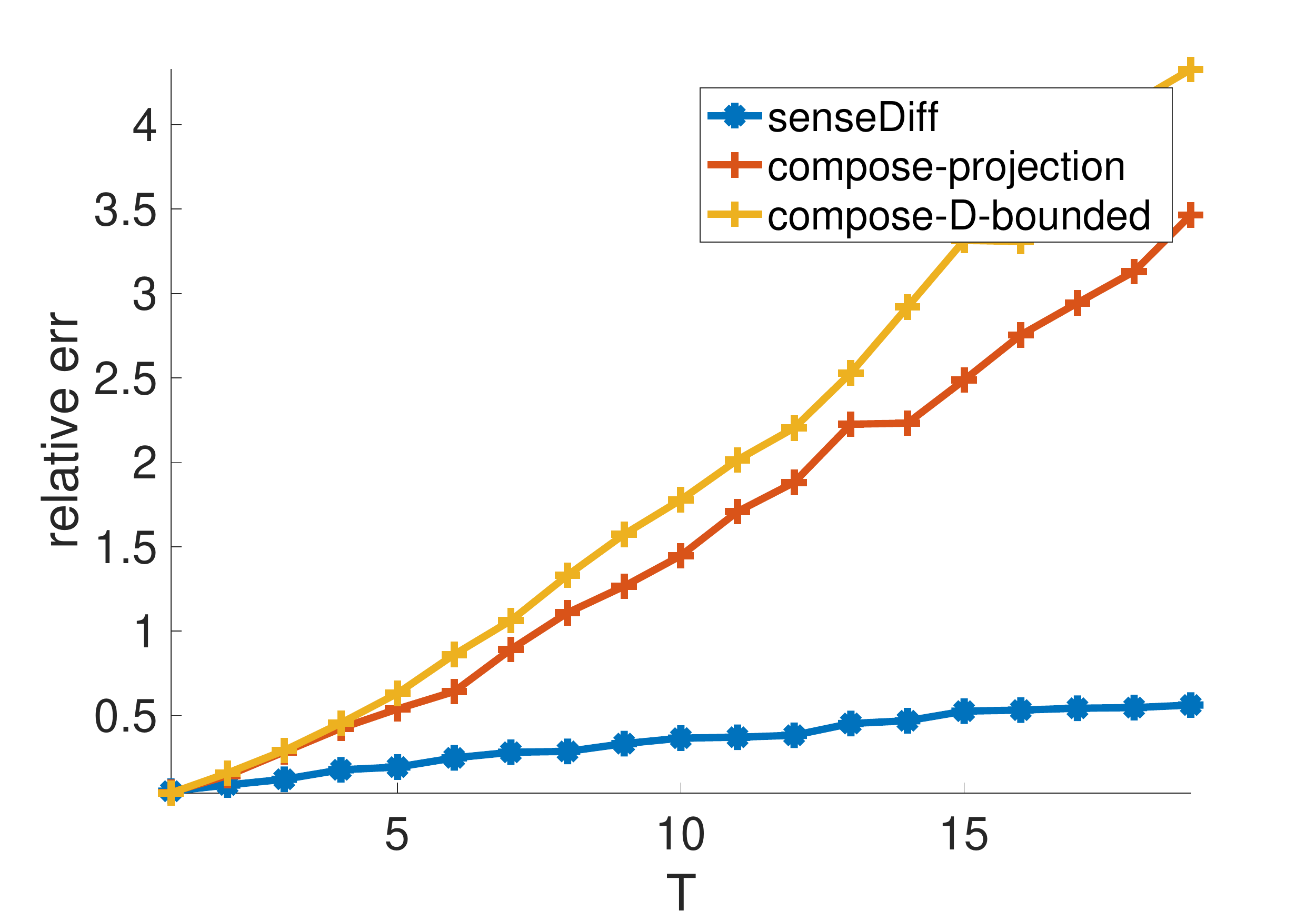}
& \includegraphics[width=0.245\textwidth]{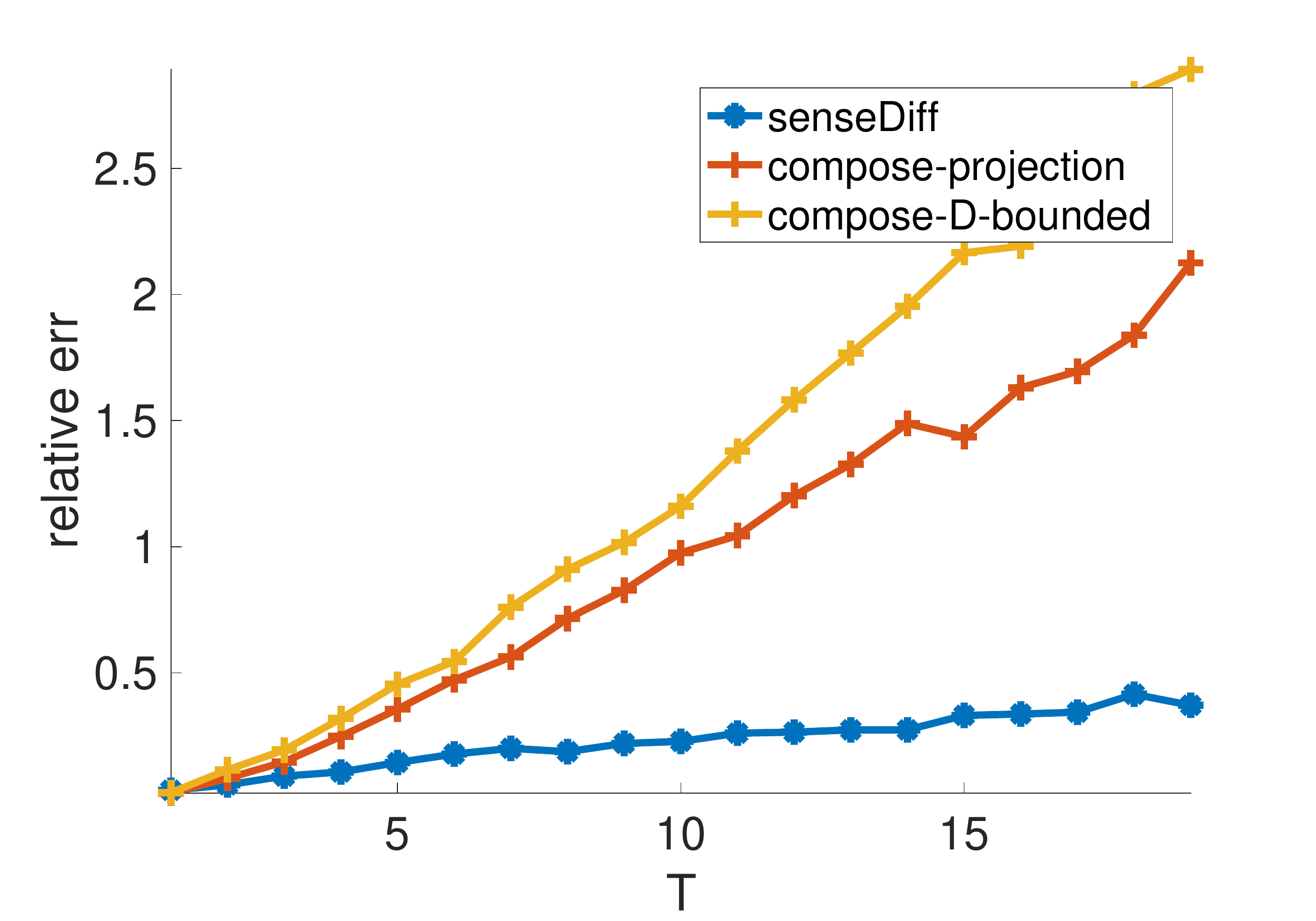}
\\ 
\rotatebox{90}{\qquad \quad Synthetic \rom{2}}
&\includegraphics[width=0.245\textwidth]{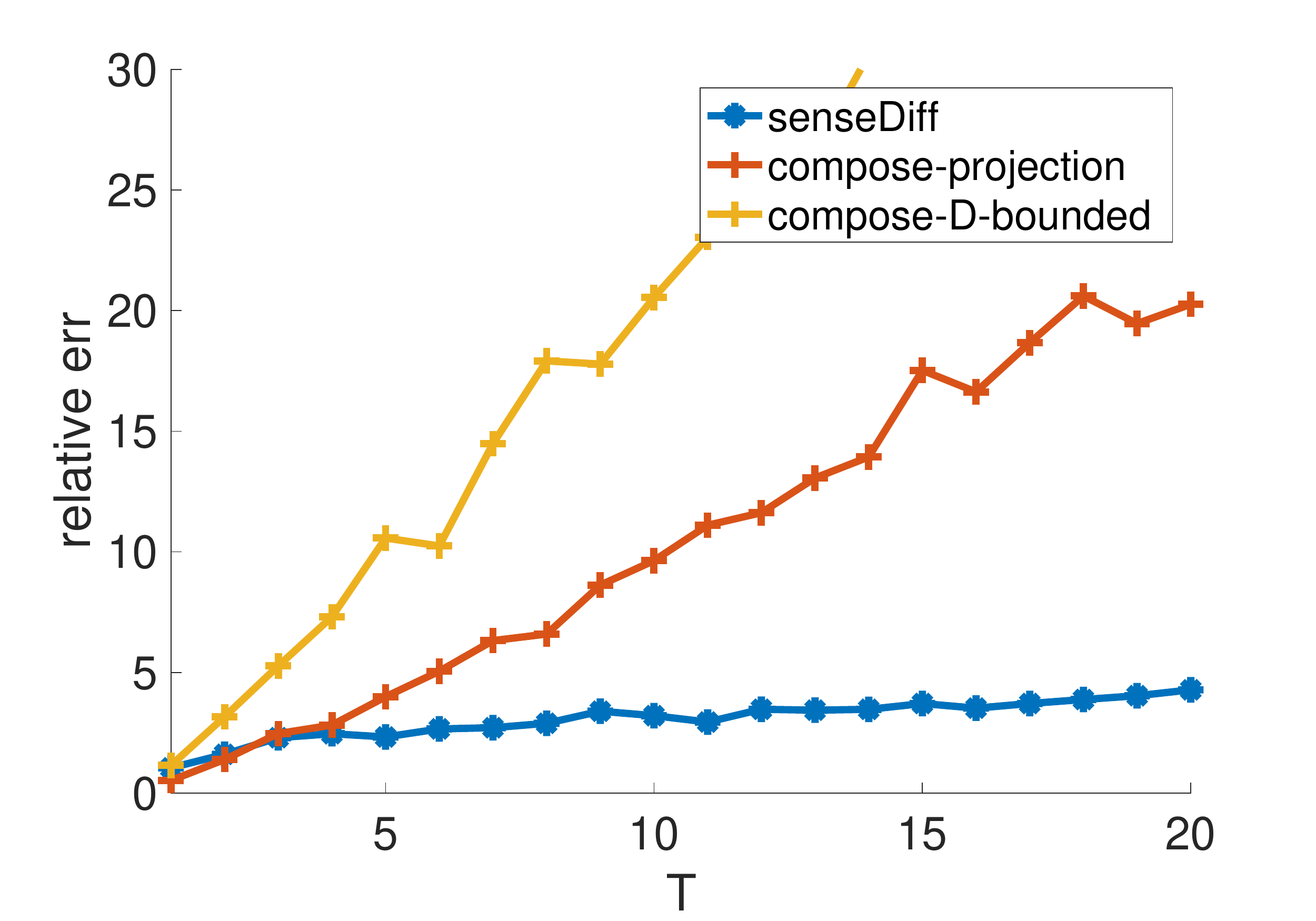}
& \includegraphics[width=0.245\textwidth]{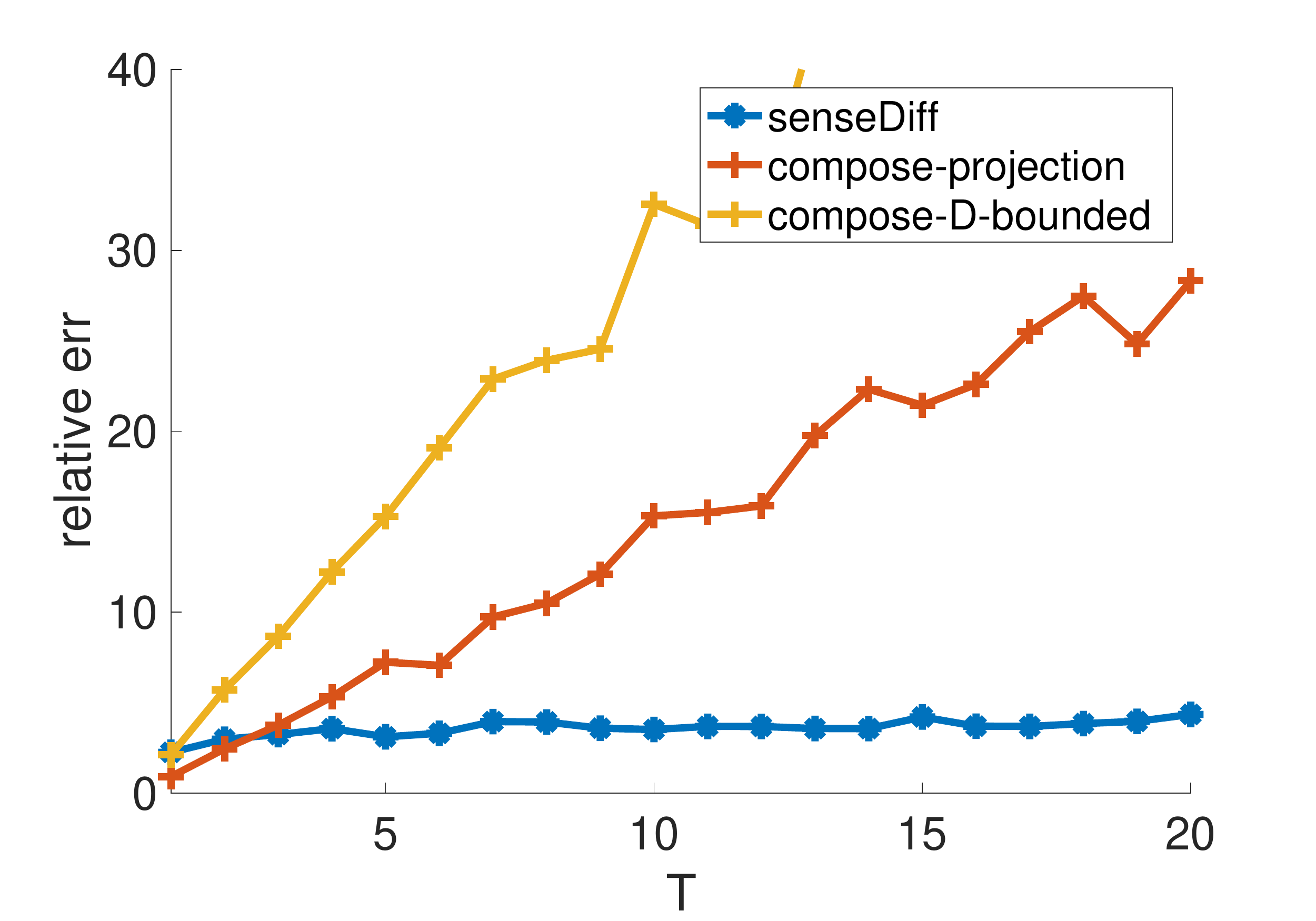}
& \includegraphics[width=0.245\textwidth]{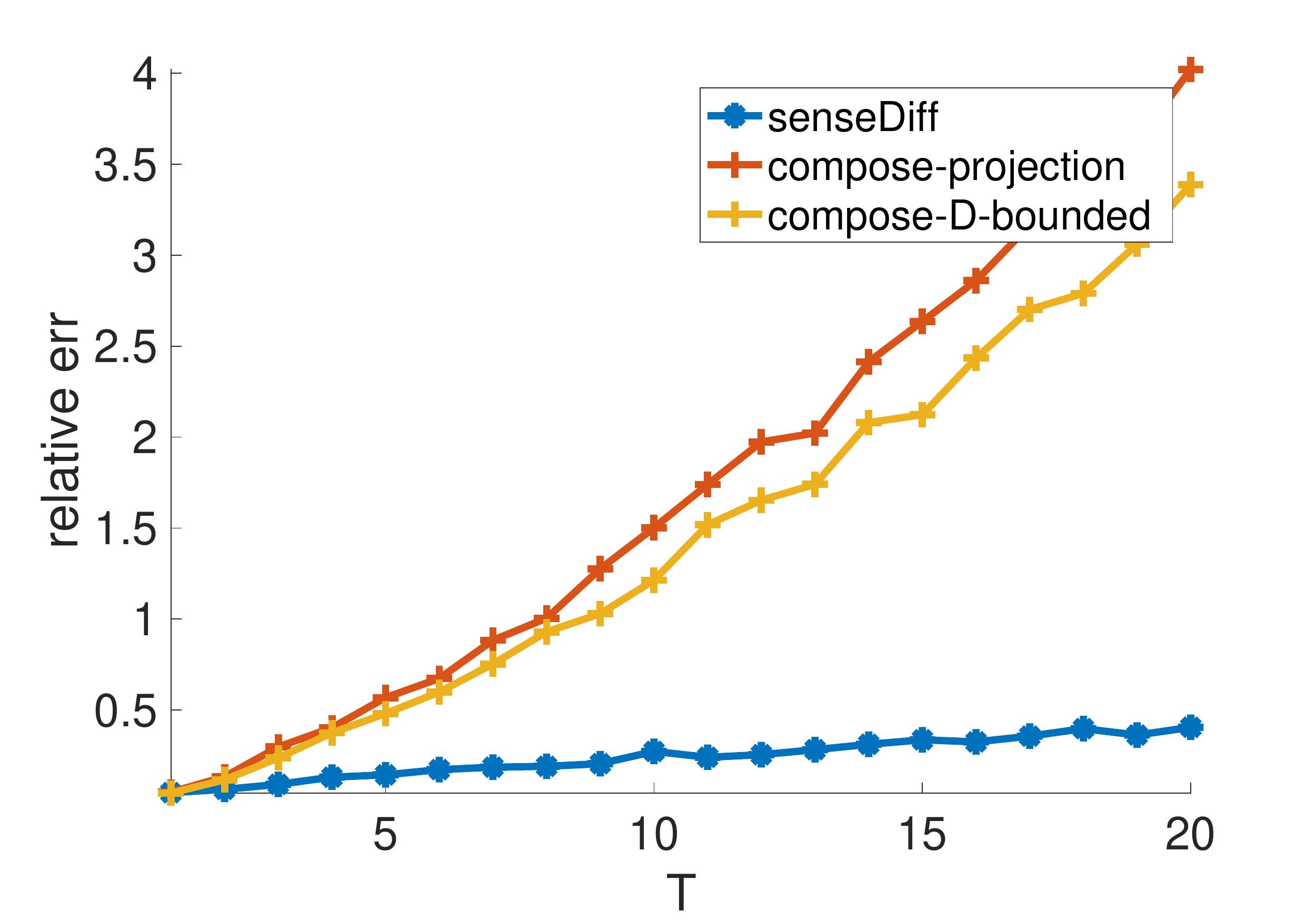}
& \includegraphics[width=0.245\textwidth]{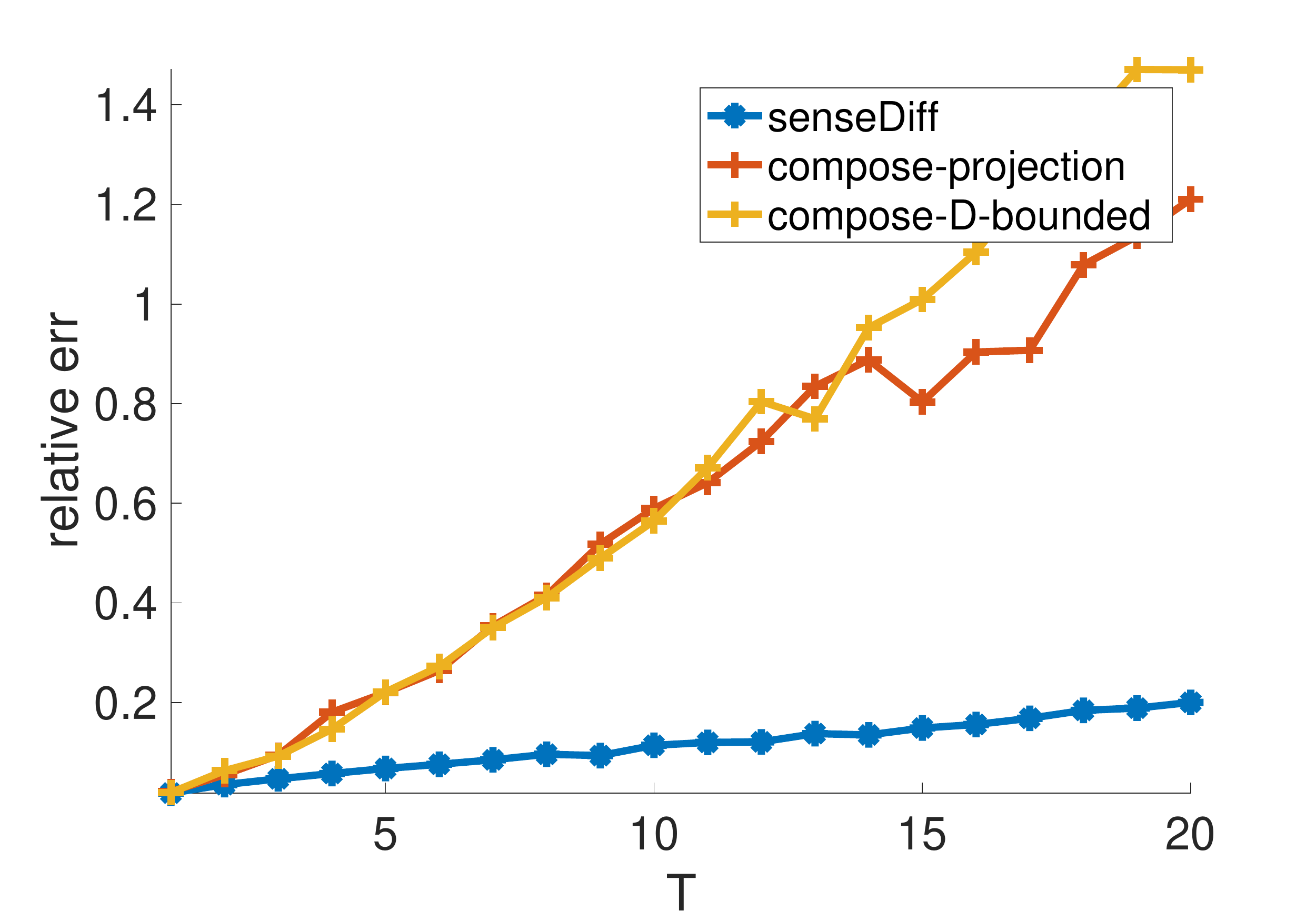}\\
\end{tabularx}}
\captionof{figure}{$L_1$ error vs. number of publications $T$. $\epsilon = 5$. Averaged over $100$ runs.}
\label{fig:time}
\end{table*}

We compare \sensdiff\ with the baseline algorithms \composeD\ and \composeproj\
under varying $\epsilon$ -- the privacy parameter and varying $T$ -- the number
of releases.  We measure utility by the relative $L_1$ error $\sum_{t=1}^T
\abs{\calA(G_t) - f(G_t)}/f(G_t)$, where $f(G_t)$ is the non-private statistic
and $\calA(G_t)$ is the value estimated by an algorithm.  

Figure~\ref{fig:epsilon} shows the privacy-utility tradeoffs of the three
algorithms for the number of high-degree nodes (1st and 2nd columns) and number
of edges (3rd and 4th columns) across all datasets (different rows).  There is
a clear trend of decreasing error as $\epsilon$ increases for all algorithms;
however, \sensdiff\ yields smaller error compared to the other two baselines in
all the cases. We point out that these results are overly optimistic for
\composeproj\ -- as we do not spend any privacy budget tuning the projection
thresholds.

Figure~\ref{fig:time} presents the utility of all algorithms across different datasets (different rows) and statistics (different columns) as a function of the number of releases $T$. We fix $\epsilon$ to be $5$.  We see that the relative error increases as $T$ increases for all algorithms, across all datasets and statistics. 
Similar to the previous experiment, \sensdiff\ achieves lower errors in most cases for the number of high-degree nodes and significantly lower errors in all cases for the number of edges. 
In addition, we observe that the error of \sensdiff\ increases at a smaller rate with increasing number of releases compared to the two baselines. This implies that \sensdiff\ achieves better utility for large $T$, and confirms the theoretical arguments in Section~\ref{sec:main_alg}.

Reconsidering the two questions proposed in the beginning of the section, we conclude that \sensdiff\ offers better utility under a wide range of $\epsilon$ and $T$ compared to both baselines. \composeD\ yields the worst utility across all datasets for both statistics, which is to be expected as it does not take advantage of either projections or additional properties of the graph sequence. \sensdiff\ outperforms \composeproj\ in most cases, and has a more significant advantage for large $T$.

\section{Conclusion}
In summary, we present a general algorithm for continually releasing statistics
of a graph sequence.  Our proposed algorithm exploits the difference sequence
of the statistics, which has lower sensitivity compared to the original
sequence, to achieve improved utility.  We derive the global sensitivity of
the difference sequences for common statistics including degree statistics and
subgraph counts for bounded-degree graphs.  Evaluations on real and synthetic
graphs demonstrate the practical applicability of the proposed algorithm by
showing that it outperforms two natural baselines over a wide range of
parameters.  In particular, the proposed algorithm is much less sensitive with
respect to the number of releases compared with the baselines.  

Our work is thus a step towards privacy-preserving analysis of
online graphs, and has the potential to lead to insights in epidemiology of
stigma-inducing diseases such as HIV, while still preserving the patient
privacy.

\section*{Acknowledgment}
The work was sponsored by NIH under MH100974, AI106039, ONR under N00014-16-1-261, and NSF under IIS 1253942.

\bibliographystyle{ieeetran}
\bibliography{ref}

\appendices
\section{Notations}
We first introduce some notations which are going to be used throughout the Appendix.

Given an undirected graph $G$, let $\edge{G}{v}$ denote the set of edges that are adjacent to node $v$. 
Given a directed graph $G$, let $\inedge{G}{v}$ and $\outedge{G}{v}$ denote the set of in-edges and the set of out-edges adjacent to $v$ respectively.
We may omit $\tau$ and write $\highdegree{}{G}$ or $\highoutdegree{}{G}$ if it is clear from the context.

\section{Proofs of Section~\ref{sec:degree}}

\subsection{Undirected Graph}

\begin{proof}(of Lemma~\ref{lem:high-degree-undirected})
Consider $G = (V, E)$ with $V = \cup_{t=1}^{\infty}\partial{V_{t}}$ and $G = (V', E')$ with $V' = \cup_{t=1}^{\infty}\partial{V'_{t}}$ such that $\partial{V'_{i}} = \partial{V_{i}} \cup \{v^*\}$ and $\partial{V_{j}} = \partial{V'_{j}}$ for $j\neq i$. Let $\seq{\Delta_t}{t=1}{\infty}$ and $\seq{\Delta'_t}{t=1}{\infty}$ be the difference sequences corresponding to $G$ and $G'$. Observe that $\Delta_t = \highdegree{\tau}{G_t} - \highdegree{\tau}{G_{t-1}}$ is the number of nodes in $G$ whose degrees {\em{cross}} the threshold $\tau$ at time $t$, and our goal is to bound $\sum_{t} |\Delta_t - \Delta'_t|$.

First, observe that as new edges become visible, for any node $v$, and for $j > i$, we have $\degree{G_{i}}{v} \leq \degree{G_{j}}{v}$. Thus the degree of $v$ can cross the threshold $\tau$ in at most one time step. If this happens at time step $t$, then $v$'s crossing over increases only the entry $\Delta_t$ by $1$, and changes no other $\Delta_{t'}$ for $t \neq t'$.

Now consider the effect of adding $v^*$ to the graph. Two things can happen as a result of this addition. First, $v^*$ itself may have degree higher than $\tau$, but since its degree can cross $\tau$ at only one time step $t$, this will increase at most one $\Delta_t$ by at most one. 

Second, some of the nodes that are connected to $v^*$ can cross the degree threshold $\tau$ early as a result on this extra connection. Let $v_{\ell}$ be a node connected to $v^*$. If the degree of $v_{\ell}$ crossed the threshold $\tau$ at time $t$ in $G$ and now crosses it at time $t' < t$ in $G'$, then this can change only the two entries $\Delta_t$ and $\Delta_{t'}$, each by at most $1$. Since there are at most $D$ such nodes, and any node that does not connect to $v^*$ crosses $\tau$ at the same time in $G$ and $G'$, the total change $\sum_{t} | \Delta_t - \Delta'_{t}|$ in the difference sequences is at most $2D+1$. 

Next, we show two neighboring graphs $G$ and $G'$ where the $L_1$-distance between the two difference sequences is exactly $2D + 1$ for $\tau < D$.
For any $D$ and $\tau$, we can construct the following two graphs.
Let $\partial{V_1} = \{u_1,\dots,u_{\tau - 1}\} \cup \{v_1,\dots, v_D\}$, $\partial{V_2} = \{u_{\tau}\}$ and $\partial{V_i} = \emptyset$ for any $i \geq 3$. Let $\partial{V'_1} = \partial{V_1} \cup \{v^*\}$ and $\partial{V'_i} = \partial{V_i}$ for any $i \geq 2$.
Let $E = \{(v_i, u_j):i\in[D], j\in [\tau]\}$ and $E' = E\cup \{(v_i, v^*):i\in[D]\}$.
It is obvious that $\degree{}{u_i} = D$, $\degree{}{v^*} = D$, $\degree{G'}{v_j} = \tau + 1 \leq D$ for all $i\in[D], j\in [\tau]$, which means both $G$ and $G'$ are $D$-bounded.
The difference sequence of $G$ is $(\tau-1, D+1, 0, \dots)$ and that of $G'$ is $(D+\tau, 1, 0, \dots)$, with $L_1$ distance $2D+1$.

We can conclude that the $D$-bounded global sensitivity of the difference sequence is always upper bounded by $2D+1$, and is exactly $2D+1$ for any $\tau < D$.
\end{proof}

\begin{proof}(of Lemma~\ref{lem:high-degree projection})
Consider the following graph $G$.
Let there be $T$ copies of a $(\tD-1)$-star, denoted by $\{S_1,\dots,S_T\}$. For each $t\in[T]$, let the center node of star $S_t$ be $C_t$. 
For every $t\in[T]$, let there be an edge $(C_i, C_{i+1})$. 
And for each $t\in[T]$, let the time stamps of all nodes in $S_t$ be $t$. 

Let $G'$ be a graph which is the same as $G$ except for an additional node $v^*$ and an additional edge $(v^*, C_1)$. Let $v^*$ have time stamp $0$.

Now we consider the projected graphs with projection threshold $\tD$ for $G$ and $G'$.

The edges in $G^{\tD}$ includes all the edges in $S_i$ for all $i$, and $(C_1,C_2)$, $(C_3,C_4)$, $(C_5,C_6)$, \dots. 
We thus have $\seq{\Delta_t}{t=1}{T}  = (0, 2, 0, 2, \dots)$.

The edges in $G'^{\tD}$ includes all the edges in $S_i$ for all $i$, and $(v^*, C_1)$, $(C_2,C_3)$, $(C_4,C_5)$, \dots. 
We thus have $\seq{\Delta_t}{t=1}{T}  = (1, 0, 2, 0, 2 \dots)$.

Therefore, $\sum_{t=1}^{T} \abs{\Delta_t - \Delta'_t} = 2(T-1) + 1 = 2T - 1$ for any $T$. 
Notice that if there are $\tD$ copies of the star sequences, and $v^*$ is connected to the center of the first stars of all the sequences, then $\sum_{t=1}^{T} \abs{\Delta_t - \Delta'_t}$ would become $(2T - 1) \tD$.
\end{proof}

\begin{proof}(of Lemma~\ref{lem:degree-hist undirected})
Consider $G = (V, E)$ with $V = \cup_{t=1}^{T}\partial{V_{t}}$ and $G' = (V', E')$ with $V' = \cup_{t=1}^{T}\partial{V'_{t}}$ with $\partial{V'_{i}} = \partial{V_{i}} \cup \{v^*\}$ and $\partial{V_{j}} = \partial{V'_{j}}$ for $j\neq i$. 

Moreover, observe that for any $t$, the $t$-th entry in the difference sequence, $\Delta_t = (n_{t, 0}, n_{t, 1}, \ldots, n_{t, D})$, where $n_{t, k}$ is the number of nodes which have degree $k$ for the {\em{first time}} in $G_t$, minus the number of nodes whose degree changes from $k$ to $k+1$ at time $t$. 

If $(\Delta_t)$ and $(\Delta'_t)$ are the difference sequences corresponding to $G$ and $G'$ respectively, then our goal is to bound $\sum_{t} \|\Delta_t - \Delta'_t\|_1$. 

First, observe that as new edges become visible, for any node $v$, and for $j > i$, we have $\degree{G_{i}}{v} \leq \degree{G_{j}}{v}$. Consider any node $v\in V$ with $\degree{G}{v} = d$. Since $\degree{G_{i}}{v}$ is increasing with respect to $i$, it can change at most $d$ times after $v$ enters the graph. Changing from $k$ to $k'$ in time $t$ would decrease $n_{t, k}$ by $1$ and increase $n_{t, k'}$ by $1$, for a total change of $2$, and the entrance of $v$ to the graph with degree $k_0$ at time $t$ would increase $n_{t, k_0}$ by $1$.

Now consider the effect of adding $v^*$ on the difference sequence. Since the $\degree{G'}{v^*} \leq D$, with the analysis above, it at most can increase $\sum_{t} \|\Delta_t - \Delta'_t\|_1$ by $2D+1$ -- at most one for $v^*$'s entrance to the graph, and at most $2D$ for degree increases.

Additionally, since  $\degree{G'}{v^*} \leq D$, there are at most $D$ nodes connected to $v^*$; let us call them $u_1, \ldots, u_D$. Consider a particular $u_k$, which now has an extra edge to $v^*$. Observe that adding $v^*$ cannot cause $u_k$ to enter the graph at a different time; however, it can cause $u_k$ to change its degree from $j$ to $j+1$ for any $j \in \{0, \ldots, D-1\}$ at an earlier time step, which may cause $\sum_{t} |\Delta_t - \Delta'_t|_1$ to increase by at most $4D$. Since there are $D$ such nodes, the total change due to the addition of $v^*$ is at most $4D^2 + 2D + 1$. The lemma follows.
\end{proof}

\subsection{Directed Graph}
\begin{proof}(of Lemma~\ref{lem:high-degree-directed})
Consider $G = (V, E)$ with $V = \cup_{t=1}^{T}\partial{V_{t}}$ and $G = (V', E')$ with $V' = \cup_{t=1}^{T}\partial{V'_{t}}$ with $\partial{V'_{i}} = \partial{V_{i}} \cup \{v^*\}$ and $\partial{V_{j}} = \partial{V'_{j}}$ for $j\neq i$. Let $(\Delta_t)$ and $(\Delta'_t)$ be the difference sequences corresponding to the graphs $G$ and $G'$ and $f = \highoutdegree{\tau}{\cdot}$. Observe that $\Delta_t = \highoutdegree{\tau}{G_t} - \highoutdegree{\tau}{G_{t-1}}$ is the number of nodes in $G$ whose out-degrees cross the threshold $\tau$ at time step $t$, and our goal is to bound $\sum_{t} |\Delta_t - \Delta'_t|$.

First, observe that as new edges become visible, for any node $v$, and for $j > i$, we have $\outdegree{G_{i}}{v} \leq \outdegree{G_{j}}{v}$. Thus the out-degree of $v$ can cross the threshold $\tau$ in at most one time step. If this happens at time step $t$, then $v$'s crossing over increases only the entry $\Delta_t$ by $1$, and changes no other entries.

Now consider the effect of adding $v^*$ to the graph. Two things can happen as a result of this addition. First, $v^*$ itself may have out-degree higher than $\tau$, but since its degree can cross $\tau$ at only one time step $t$, this will increase at most one $|\Delta_t-\Delta'_t|$ by at most one. 

Second, some of the nodes that are connected to $v^*$ can cross the degree threshold $\tau$ early as a result on this extra connection. Let $v_{\ell}$ be a node that points to $v^*$. If the out-degree of $v_{\ell}$ crossed the threshold $\tau$ at time $t$ in $G$ and now crosses it at time $t' < t$ in $G'$, then this can change only the two entries $\Delta_t$ and $\Delta_{t'}$, each by at most $1$. Since there are at most $\Din$ such nodes, and any node that does not point to $v^*$ crosses $\tau$ at the same time in $G$ and $G'$, the total change $\sum_{t} | \Delta_t - \Delta'_{t}|$ in the difference sequences is at most $2\Din+1$. 

Now we show the lower bound by constructing $G$ and $G'$ with $L_1$ distance $2\Din+1$ between their difference sequences.
For any $\Din$ and $\tau$, we can construct the following two graphs.
Let $\partial{V_1} = \{u_1,\dots,u_{\tau - 1}\} \cup \{v_1,\dots, v_{\Din}\}$, $\partial{V_2} = \{u_{\tau}\}$, $\partial{V_3} = \{w_1, \dots, w_{\tau}\}$ and $\partial{V_i} = \emptyset$ for any $i \geq 4$. Let $\partial{V'_1} = \partial{V_1} \cup \{v^*\}$ and $\partial{V'_i} = \partial{V_i}$ for any $i \geq 2$.
Let $E = \{(v_i, u_j):i\in[\Din], j\in [\tau]\}$ and $E' = E\cup \{(v_i, v^*):i\in[\Din]\} \cup \{(v^*, w_k):k\in[\tau]\}$.
It is obvious that $\indegree{}{u_i} = \Din$, $\indegree{}{v^*} = \Din$, $\indegree{}{v_j} = 0$, $\indegree{G'}{w_k} = 1$ for all $i\in[\Din], j, k\in [\tau]$, which means both $G$ and $G'$ are $\Din$-in-bounded.
The difference sequence of $G$ is $(0, \Din, 0, \dots)$ and that of $G'$ is $(\Din, 0, 1, 0, \dots)$, with $L_1$ distance $2\Din+1$.

We can conclude that the $\Din$-in-bounded global sensitivity of the difference sequence is $2\Din+1$.
\end{proof}

\begin{proof}(of Lemma~\ref{lem:out-degree-hist directed})
Consider $G = (V, E)$ with $V = \cup_{t=1}^{\infty}\partial{V_{t}}$ and $G' = (V', E')$ with $V' = \cup_{t=1}^{\infty}\partial{V'_{t}}$ with $\partial{V'_{i}} = \partial{V_{i}} \cup \{v^*\}$ and $\partial{V_{j}} = \partial{V'_{j}}$ for $j\neq i$. 

Moreover, observe that for any $t$, the $t$-th entry in the difference sequence, $\Delta_t = (n_{t, 0}, n_{t, 1}, \ldots, n_{t, D})$, where $n_{t, k}$ is
the number of nodes which have out-degree $k$ for the {\em{first time}}
subtract 
the number of nodes whose out-degree changes from $k$ to $k+1$ in $G_t$. If $(\Delta_t)$ and $(\Delta'_t)$ are the difference sequences corresponding to $G$ and $G'$ respectively, then our goal is to bound $\sum_{t} \|\Delta_t - \Delta'_t\|_1$. 

First, observe that as new edges become visible, for any node $v$, and for $j > i$, we have $\outdegree{G_{i}}{v} \leq \outdegree{G_{j}}{v}$. Consider any node $v\in V$ with $\outdegree{G}{v} = d$. Since $\outdegree{G_{i}}{v}$ is increasing with respect to $i$, it can change at most $d$ times after $v$ enters the graph. Changing from $k$ to $k'$ in time $t$ would decrease $n_{t, k}$ by $1$ and increase $n_{t, k'}$ by $1$, and the entering of $v$ to the graph with out-degree $k_0$ at time $t$ would increase $n_{t, k_0}$ by $1$.

Now consider the effect of adding $v^*$ on the difference sequence. Since the $\outdegree{G'}{v^*} \leq \Dout$, with the analysis above, it at most can increase $\sum_{t} \|\Delta_t - \Delta'_t\|_1$ by $2\Dout+1$ -- at most one for $v^*$'s entrance to the graph, and at most $2\Dout$ for out-degree increases.

Additionally, since  $\indegree{G'}{v^*} \leq \Din$, there are at most $\Din$ nodes that point to $v^*$; let us call them $u_1, \ldots, u_D$. Consider a particular $u_k$, which now has an extra edge to $v^*$. Observe that adding $v^*$ cannot cause $u_k$ to enter the graph at a different time; however, it can cause $u_k$ to change its out-degree from $j$ to $j+1$ for any $j \in \{0, \ldots, D-1\}$ at an earlier time step, which may cause $\sum_{t} |\Delta_t - \Delta'_t|_1$ to increase by $4\Dout$. Since there are at most $\Din$ such nodes, the total change due to the addition of $v^*$ is at most $4\Dout\Din + 2\Dout + 1$. The lemma follows.
\end{proof}

\section{Proofs of Section~\ref{sec:subgraph}}

First, let $\S{G}{v}$ denote the number of copies of $S$ in $G$ that include $v$ as one of its nodes.

\subsection{Undirected Graph}
\begin{proof}(of Lemma~\ref{lem:subgraph undirected})
The proof of Lemma~\ref{lem:subgraph directed} applies.
\end{proof}

\begin{proof}(of Lemma~\ref{lem:subgraph undirected patterns})
Consider $G = (V, E)$ and $G' = (V', E')$ where $V' = V \cup \{v^*\}$ and $E' = E \cup {\Delta E}$ with ${\Delta E}$ consists of all edges adjacent to $v^*$. By the requirement, $|{\Delta E}| \leq D$. 

For any subgraph $S$, let $\calS$ denote the copies of $S$ in $G$ and $\calS'$ denote that in $G'$. Notice that $\calS \subseteq \calS'$, since $G$ is a subgraph of $G'$ and thus any subgraph that appears in $G$ would appear in $G'$. Moreover, any subgraph $\bar{S} \in \calS' \backslash \calS$ contain $v^*$, since otherwise it would appear in $\calS$. 

We now consider each of the subgraphs.

\begin{enumerate}
\item
Every edge $(v^*, v)$ can form one copy of $S$. There are in total at most $D$ such edges.

\item
Every pair of edges $(v^*, v_1) \in {\Delta E}$ and $(v^*, v_2) \in {\Delta E}$ can form one copy of $S$ if there $(v_1, v_2) \in E$.
There are in total at most ${D \choose 2}$ such pairs, and thus $S_{+} = {D \choose 2}$.

\item
If $k > D$, then there cannot be any copy of $S$ in the graph, and thus $S_{+} = 0$. Now we consider $k \leq D$.
There are two types of stars in $ \calS' \backslash \calS$, one with $v^*$ as the center of the star, the other with $v^*$ as a non-central node.
Every $k$ edges adjacent to $v^*$ can form one star with $v^*$ as its center; there are in total at most ${D \choose k}$ such stars.
Every edge $(v, v^*)$ of $v^*$ can participate in at most ${D-1 \choose k - 1}$ stars. This is because in $G$, there are in total ${\degree{}{v} \choose k}$ stars that are centered at $v$ and in $G'$, the value becomes ${\degree{}{v}+1 \choose k}$. Because of the degree bound, we have ${\degree{}{v}+1 \choose k} - {\degree{}{v} \choose k} = {\degree{}{v} \choose k - 1} \leq {D-1 \choose k - 1}$, where the last step follows because $\degree{}{v}+1 \leq D$. There are at most $D$ in-edges of $v^*$, which in total can form at most $D {D-1 \choose k - 1}$ stars.
Therefore, $S_{+} = {D \choose k} + D {D-1 \choose k - 1}$ for $k \leq D$.
\end{enumerate}
\end{proof}

\subsection{Directed Graph}
\begin{proof}(of Lemma~\ref{lem:subgraph directed})
Consider $G = (V, E)$ with $V = \cup_{t=1}^{\infty}\partial{V_{t}}$ and $G = (V', E')$ with $V' = \cup_{t=1}^{\infty}\partial{V'_{t}}$ with $\partial{V'_{i}} = \partial{V_{i}} \cup \{v^*\}$ and $\partial{V_{j}} = \partial{V'_{j}}$ for $j\neq i$. 

Let $(\Delta_t)$ denote the difference sequence for subgraph $S$ in $G$ and $(\Delta'_t)$ for that in $G'$.
Recall that as defined in Section~\ref{sec:notation}, for a subgraph pattern $S$, $\S{}{G}$ is the total number of $S$ in graph $G$.

First, observe that for any graph $G$, as more nodes (and corresponding edges) become visible, the total number of $S$ in $G$ can only increase, i.e., $\S{}{G_i} \leq \S{}{G_j}$ for any $i < j$. 
So all elements $\Delta_t$ (and $\Delta_t$) in the difference sequence are non-negative. Moreover, the sum of this sequence equals to $\S{}{G}$ (and $\S{}{G'}$ for $G$).

Consider the additional node $v^*$. By requirement, if there are at most $\Din$ in-edges and at most $\Dout$ out-edges to $v^*$, $\S{}{G'} - \S{}{G}$ is at most $S_{+}$.

Second, let us compare $\Delta_t$ and $\Delta'_t$. 
For any $t$, let $\calS_t$ denote the copies of $S$ that appears at time $t$ for the first time in $G$, and $\calS'_t$ denote that for $G'$. 
Any $S_t \in \calS_t$ consists of some nodes in $\partial{V_t}$ and some nodes in $V_{t-1}$. Consider the appearance of $S_t$ in $G'$. $\partial{V_t} \subseteq \partial{V'_t}$ and $V_{t-1} \subseteq V'_{t-1}$ implies that $S_t$ must appear in $G'$ at some $t' \leq t$; while $\partial{V_t} \cap V'_{t-1} = \emptyset$ implies that $S_t$ appears at exactly $t$. Therefore, $\calS_t \subseteq \calS'_t$ and $\Delta'_t \geq \Delta_t$ for any $t$.

Therefore,
$\sum_t |\Delta_t - \Delta'_t| = \sum_t (\Delta'_t - \Delta_t) = \sum_t \Delta'_t - \sum_t \Delta_t = \S{}{G'} - \S{}{G} \leq S_{+}$, and the lemma follows.
\end{proof}

\begin{proof}(of Lemma~\ref{lem:subgraph directed patterns})
Consider $G = (V, E)$ and $G' = (V', E')$ where $V' = V \cup \{v^*\}$ and $E' = E \cup E_1 \cup E_2$ with $E_1$ consists of all in-edges of $v^*$, $E_2$ consists of all out-edges of $v^*$. By the requirement, $|E_1| \leq \Din$ and $|E_2| \leq \Dout$. 

For any subgraph $S$, let $\calS$ denote the copies of $S$ in $G$ and $\calS'$ denote that in $G'$. Notice that $\calS \subseteq \calS'$, since $G$ is a subgraph of $G'$ and thus any subgraph that appears in $G$ would appear in $G'$. Moreover, any subgraph $\bar{S} \in \calS' \backslash \calS$ contain $v^*$, since otherwise it would appear in $\calS$. 

We now consider each of the subgraphs.

\begin{enumerate}
\item
Every edge $(v^*, v)$ or $(v^*, v)$ can form one copy of $S$. There are in total at most $\Din+\Dout$ such edges.

\item 
Consider any $\bar{S} \in \calS' \backslash \calS$. According to the definition of $S$, $\bar{S}$ must contain one in-edge and one out-edge of $v^*$. Notice that every pair of edges $(v_1, v^*) \in E_1$ and $(v^*, v_2) \in E_2$ can form one copy of $S$ if there $(v_2, v_1) \in E$. There are in total at most $\Din \Dout$ such pairs, and thus $S_{+} = \Din \Dout$.

\item
Every pair of edges $(v_1, v^*) \in E_1$ and $(v^*, v_2) \in E_2$ can form one copy of $S$ if there $(v_1, v_2) \in E$;
every pair of edges $(v_1, v^*) \in E_1$ and $(v'_1, v^*) \in E_1$ can form one copy of $S$ if there $(v_1, v'_1) \in E$ or $(v'_1, v_1) \in E$,
every pair of edges $(v^*, v_2) \in E_2$ and $(v^*, v'_2) \in E_2$ can form one copy of $S$ if there $(v_2, v'_2) \in E$ or $(v'_2, v_2) \in E$.
There are in total at most ${\Din + \Dout \choose 2}$ such pairs, and thus $S_{+} = {\Din + \Dout \choose 2}$.

\item
If $k > \Dout$, then there cannot be any copy of $S$ in the graph, and thus $S_{+} = 0$. Now we consider $k \leq \Dout$.
There are two types in $ \calS' \backslash \calS$, one with $v^*$ as the center of the star, the other with $v^*$ as a non-central node.
Every $k$ out-edges of $v^*$ can form one star with $v^*$ as its center; there are in total at most ${\Dout \choose k}$ such stars.
Every in-edge $(v, v^*)$ of $v^*$ can form at most ${\Dout-1 \choose k - 1}$ stars. This is because in $G$, there are in total ${\outdegree{}{v} \choose k}$ stars that are centered at $v$ and in $G'$, the value becomes ${\outdegree{}{v}+1 \choose k}$. Because of the degree bound, we have ${\outdegree{}{v}+1 \choose k} - {\outdegree{}{v} \choose k} = {\outdegree{}{v} \choose k - 1} \leq {\Dout-1 \choose k - 1}$, where the last step follows because $\outdegree{}{v}+1 \leq \Dout$. There are at most $\Din$ in-edges of $v^*$, which in total can form at most $\Din {\Dout-1 \choose k - 1}$ stars.
Therefore, $S_{+} = {\Dout \choose k} + \Din {\Dout-1 \choose k - 1}$ for $k \leq \Dout$. 

\item 
If we reverse the graph by converting every edge $(u,v)$ to $(v,u)$, then the problem reduces to the above problem. The same result applies with $\Din$ and $\Dout$ flipped, since the reversed graph is $(\Dout, \Din)$-bounded.
\end{enumerate}
\end{proof}

\section{Baseline Algorithms and Analyses}\label{sec:appendix_baseline}
In this section, we describes and analyze some baseline algorithms that are used in our experiments to publish the number of high-degree nodes. Since some of them use an extension of the projection algorithm in \cite{day2016publishing}, we would formally describe the extensions of the projection algorithm and then describe our baseline algorithms with their sensitivity analysis.

\subsection{Extension of the Projection Algorithm in \cite{day2016publishing} to Directed and Online Graphs}
\cite{day2016publishing} provides a projection algorithm for undirected graph. In this section, we show that it can be adapted to directed graph and graph sequence. We describe the adapted algorithm in this section, and provide analysis in the next section for using it to publish the number of high-degree nodes.

The idea of the algorithm in \cite{day2016publishing} is to consider each edge one by one, and add it to the projected graph only if both ends of the edge have degree lower than the projection threshold $\tD$.
The algorithm therefore needs an ordering of all edges, and the privacy proof in \cite{day2016publishing} requires the ordering to be consistent under $G$, $G'$, i.e., two edges, if appears in both $G$ and $G'$, would have same relative order.

The algorithm can be naturally adapted to the directed graph setting, by having two projection thresholds, $\tDin$ for in-degree and $\tDout$ for out-degree.
It can also be naturally adapted to the online graph setting, where we have a sequence of graphs $(G_1,G_2,\dots,)$, and would like to publish a sequence of projected graphs $(G_1^{\tD}, G_2^{\tD}, \dots)$. Here we can consider all edges in $\partial{E_1}$ one by one and publish $G_1^{\tD}$, and then consider edges in $\partial{E_2}$ and publish $G_2^{\tD}$ etc. More formally, we require the ordering used in the algorithm to be consistent with their order of appearance in the graph sequence.

For a graph $G$, let $\Lambda$ be the ordering, represented by a sequence of all edges, such that $e_1$'s appearing before $e_2$ means that $e_1$ is considered before $e_2$ by the algorithm.
For a graph sequence, let there be an ordering sequence $(\Lambda_1, \dots, \Lambda_T)$, where each $\Lambda_i$ is an ordered sequence of all edges in $\partial{E_i}$. All edges in $\Lambda_i$ are considered before edges in $\Lambda_j$ for $i < j$, and edge $e_1$ is considered before $e_2$ if $e_1$ appears before $e_2$ in some $\Lambda_i$.
Corresponding to $\Lambda$ or $(\Lambda_1,\dots,\Lambda_T)$, we can define a function $\lambda$ that maps an edge to a unique value, such that $\lambda(e_1) < \lambda(e_2)$ means $e_1$ appears before $e_2$ in $\Lambda$.

Now, given the time stamps of all nodes, we can define the time stamp of an edge $e = (v_1, v_2)$ as $e.\time = \max\{v_1.\time, v_2.\time\}$.
For a graph sequence, we want $(\Lambda_1, \dots, \Lambda_T)$ to be consistent with the time stamps of the edges in $G$, i.e., for two edges $e$, $e'$ with $e.\time < e'.\time$, either $e \in \Lambda_i$, $e'\in \Lambda_j$ with $i < j$, or $e$ and $e'$ both appear in some $\Lambda_i$ and $e$ appears before $e'$ in $\Lambda_i$.
Obviously, such ordering (sequence) is consistent under $G$ and $G'$.

Now we formally state the algorithms. Algorithm~\ref{alg:DLL_oneshot} is for projection of one undirected or directed graph, and Algorithm~\ref{alg:DLL_seq} is for an undirected or directed graph sequence. We note that Algorithm~\ref{alg:DLL_oneshot} with undirected graph is the original projection algorithm proposed in \cite{day2016publishing}.

\begin{algorithm}[h!]
\caption{Projection algorithm for one undirected / directed graph (Graph $G=(V,E)$, edge ordering $\Lambda$, projection parameter $\tD$ (for undirected graph) or $\tDin$, $\tDout$ (for directed graph))}
\label{alg:DLL_oneshot}
\begin{algorithmic}
\State{$\tilde{V} = V$, $\tilde{E} = \emptyset$.}
\For{$(u, v) \in \Lambda$}
	\If{$\degree{(\tilde{V},\tilde{E})}{u} < \tD$ {\bf and} $\degree{(\tilde{V},\tilde{E})}{v} < \tD$ (or for directed graph, $\outdegree{(\tilde{V},\tilde{E})}{u} < \tDout$ {\bf and} $\indegree{(\tilde{V},\tilde{E})}{v} < \tDin$)}
		\State{$\tilde{E} = \tilde{E} \cup \{(u, v)\}$.}
	\EndIf
\EndFor
\State{\textbf{return}$\tilde{G} = (\tilde{V},\tilde{E})$}
\end{algorithmic}
\end{algorithm}

\begin{algorithm}[h!]
\caption{Projection algorithm for undirected / directed graph sequence (Graph $G=(V,E)$ and sequence $\calG=(G_1,\dots,G_m)$, edge ordering sequence $(\Lambda_1, \dots, \Lambda_m)$, projection parameter $\tD$ (for undirected graph) or $\tDin$, $\tDout$ (for directed graph))}
\label{alg:DLL_seq}
\begin{algorithmic}
\State{$\tilde{E} = \emptyset$.}
\For{$i = 1, \dots, m$}
	\State{$\tilde{V} = V_i$}
	\For{$(u, v) \in \Lambda_i$}
		\If{$\degree{(\tilde{V},\tilde{E})}{u} < \tD$ {\bf and} $\degree{(\tilde{V},\tilde{E})}{v} < \tD$ (or for directed graph, $\outdegree{(\tilde{V},\tilde{E})}{u} < \tDout$ {\bf and} $\indegree{(\tilde{V},\tilde{E})}{v} < \tDin$)}
			\State{$\tilde{E} = \tilde{E} \cup \{(u, v)\}$}
		\EndIf
	\EndFor
	\State{$\tilde{G}_i = (\tilde{V},\tilde{E})$}
\EndFor
\State{\textbf{return}$(\tilde{G}_1,\dots,\tilde{G}_m)$}
\end{algorithmic}
\end{algorithm}

\subsection{Baselines for Publishing the Number of High-Degree Nodes}
In this section, we state the baseline algorithms for publishing the number of high-degree nodes of a graph sequence and show their sensitivity. All algorithms are based on the global sensitivity mechanism introduced in Section~\ref{sec:dp}, namely, to publish some $f(G)$, we compute the global sensitivity or the degree-bounded global sensitivity $\GS{}{f}$ of $f(G)$, and add noise proportional to $\GS{}{f}/\epsilon$.

As has been mentioned in Section~\ref{sec:main_alg}, there are two baseline algorithms for high-degree nodes count. Combining with the algorithm proposed in \cite{day2016publishing}, there are four baseline algorithms for both undirected and directed graph sequences.

For an undirected graph sequence $(G_1,\dots,G_T)$,
\begin{enumerate}
\item \composeD: For each $i \in [T]$, run the global sensitivity algorithm to publish $\highdegree{\tau}{G_i}$ with privacy parameter $\epsilon / T$ and $D$-bounded global sensitivity $D+1$ (Lemma~\ref{lem:baseline_undirected_one_D}).
\item \composeproj: For each $i \in [T]$, run Algorithm~\ref{alg:DLL_oneshot} on $G_i$ with projection parameter $\tD$ to get $G_i^{\tD}$. Run the global sensitivity algorithm to publish $\highdegree{\tau}{G_i^{\tD}}$ with privacy parameter $\epsilon / T$ and global sensitivity $\tD+1$ (Lemma~\ref{lem:baseline_undirected_one_DLL}).
\item \sensseqD: Run the global sensitivity algorithm to publish $\seq{\highdegree{\tau}{G_i}}{i=1}{T}$ with privacy parameter $\epsilon$ and $D$-bounded global sensitivity $(D + 1)T$ (Lemma~\ref{lem:baseline_undirected_T_D}).
\item \sensseqproj: Run Algorithm~\ref{alg:DLL_seq} on $\seq{G_i}{i=1}{T}$ with projection parameter $\tD$ to get $\seq{G_i^{\tD}}{i=1}{T}$. Run the global sensitivity algorithm to publish $\seq{\highdegree{\tau}{G_i^{\tD}}}{i=1}{T}$ with privacy parameter $\epsilon$ and computed global sensitivity (Lemma~\ref{lem:baseline_T_DLL}).
\end{enumerate}

For a directed graph sequence $(G_1,\dots,G_T)$,
\begin{enumerate}
\item \composeD: For each $i \in [T]$, run the global sensitivity algorithm to publish $\highoutdegree{\tau}{G_i}$ with privacy parameter $\epsilon / T$ and $\Din$-in-bounded global sensitivity $\Din+1$ (Lemma~\ref{lem:baseline_directed_one_D}).
\item \composeproj: For each $i \in [T]$, run Algorithm~\ref{alg:DLL_oneshot} on $G_i$ with projection parameter $\tDin$, $\tDout$ to get $G_i^{\tDin,\tDout}$. Run the global sensitivity algorithm to publish $\highoutdegree{\tau}{G_i^{\tDin,\tDout}}$ with privacy parameter $\epsilon / T$ and global sensitivity $\max\{\tDin+1,\tDout-1\}$ (Lemma~\ref{lem:baseline_directed_one_DLL}).
\item \sensseqD: Run the global sensitivity algorithm to publish $\seq{\highoutdegree{\tau}{G_i}}{i=1}{T}$ with privacy parameter $\epsilon$ and $\Din$-in-bounded global sensitivity $(\Din + 1)T$ (Lemma~\ref{lem:baseline_directed_T_D}).
\item \sensseqproj: Run Algorithm~\ref{alg:DLL_seq} on $\seq{G_i}{i=1}{T}$ with projection parameter $\tDin$, $\tDout$ to get $\seq{G_i^{\tDin,\tDout}}{i=1}{T}$. Run the global sensitivity algorithm to publish $\seq{\highoutdegree{\tau}{G_i^{\tDin,\tDout}}}{i=1}{T}$ with privacy parameter $\epsilon$ and computed global sensitivity (Lemma~\ref{lem:baseline_T_DLL}).
\end{enumerate}

The lemmas below show that for both undirected and directed graph sequences, (3) is worse than (1); and (4) is worse than (2). Therefore, we only need to run (1) and (2), i.e., \composeD\ and \composeproj.

\subsubsection{Analyses for \composeD}
\begin{lemma}\label{lem:baseline_undirected_one_D}
Given an undirected graph $G = (V, E)$, the $D$-bounded global sensitivity of publishing $\highdegree{\tau}{G}$ is $D+1$ for any $\tau \leq D$.
\end{lemma}
\begin{proof}
Consider neighboring graphs $G = (V, E)$ and $G' = (V', E')$ with $V' = V \cup \{v^*\}$ and $E' = E \cup E^*$ where $E^*$ consists of all edges adjacent to $v^*$.

Observe that for any node $v\in V$, $\degree{G}{v}$ and $\degree{G'}{v}$ are equal if $v$ is not connected to $v^*$ in $G'$.  There are in total $|E^*| \leq D$ nodes that are connected to $v^*$; this extra connection can potentially push their degrees over the threshold. Moreover, $v^*$ itself can have degree over $\tau$. Therefore, $\highdegree{\tau}{G'} - \highdegree{\tau}{G} \leq D + 1$.

Now we show neighboring $D$-bounded graphs $G$ and $G'$ such that $\abs{\highdegree{\tau}{G'} - \highdegree{\tau}{G}} = D+1$. Let $G$ consist of $D$ $(\tau - 1)$-stars, and $G'$ be equal to $G$ except for an added node $v^*$ that is connected to the center of all the stars. Since $D \geq \tau$, the degrees of all nodes still at most $D$. Additionally, we have $\highdegree{\tau}{G} = 0$ and $\highdegree{\tau}{G'} = D+1$. The lemma follows.
\end{proof}

\begin{lemma}\label{lem:baseline_directed_one_D}
Given a directed graph $G = (V, E)$, the $\Din$-in-bounded global sensitivity of publishing $\highoutdegree{\tau}{G}$ is $\Din+1$ for any $\tau \leq \Dout$.
\end{lemma}
\begin{proof}
Consider neighboring graphs $G = (V, E)$ and $G' = (V', E')$ with $V' = V \cup \{v^*\}$ and $E' = E \cup E^*_i \cup E^*_o$ where $E^*_i$ consists of all in-edges adjacent to $v^*$ and $E^*_o$ consists of all out-edges adjacent to $v^*$. For any node $v\in V$, $\outdegree{G}{v}$ and $\outdegree{G'}{v}$ are equal if $(v, v^*) \notin E^*_i$; and there are in total $|E^*_i|$ nodes that have inbound edges to $v^*$; this extra edge can push their outdegrees over the threshold. Additionally, $v^*$ itself may have outdegree over the threshold $\tau$. Therefore, $\highoutdegree{\tau}{G'} - \highoutdegree{\tau}{G} \leq |E^*_i| + 1 \leq \Din + 1$.

Now we show neighboring $D$-bounded graphs $G$ and $G'$ such that $\abs{\highoutdegree{\tau}{G'} - \highoutdegree{\tau}{G}} = \Din+1$.
Let $G$ consists of $\Din$ $(\tau - 1)$-out-stars and $\tau$ isolated nodes. Let $G'$ be equal to $G$ except for a node $v^*$ that has inbound edges from the center of all the stars and outbound edges to all the isolated nodes. Observe that as $\tau \leq \Din, \Dout$, the in-degrees of all nodes are no larger than $\Din$ and out-degrees of all nodes are no larger than $\Dout$, and we have $\highoutdegree{\tau}{G} = 0$ and $\highoutdegree{\tau}{G'} = \Din+1$.
\end{proof}

\subsubsection{Analyses for \composeproj}
\begin{lemma}\label{lem:baseline_undirected_one_DLL}
Using Algorithm~\ref{alg:DLL_oneshot} on an undirected graph $G$ to get $G^{\tD}$ and publishing $\highdegree{\tau}{G^{\tD}}$ has global sensitivity $\tD + 1$ for any $0 < \tau \leq \tD$.
\end{lemma}
\begin{proof}
For two graphs $H_1$ and $H_2$ with the same set of nodes, we say an edge has the same status in graph $H_1$ and $H_2$ if it is present in both graphs or absent in both.

Consider neighboring graphs $G = (V, E)$ and $G' = (V', E')$ with $V' = V \cup \{v^*\}$ and $E' = E \cup E^*$, where $E^*$ consists of all edges adjacent to $v^*$.
Let $\Lambda$ and $\lambda$ be the edge order and the corresponding ordering function.
Let the projected graphs be $G^{\tD}$ and $G'^{\tD}$.

In $G'$, $v^*$ may have a large number of adjacent edges $E^*$, but it is adjacent to at most $\tD$ edges in $G'^{\tD}$. Let such edges be $\{e^*_1,\dots,e^*_i\}$ with $\lambda(e^*_j) < \lambda(e^*_{j+1})$ for any $j$. Obviously, $i \leq \tD$. 
Notice that the projected graph of $(V', E\cup E^*)$ and that of $(V', E\cup \{e^*_1,\dots,e^*_i\})$ are the same, therefore we can assume $E^* = \{e^*_1,\dots,e^*_i\}$, which would not change $G'^{\tD}$.

Let $G_0 = (V', E)$, $G_1 = (V', E\cup \{e^*_1\})$, $G_2 = (V', E\cup \{e^*_1, e^*_2\})$, \dots, $G_i = (V', E\cup \{e^*_1,\dots,e^*_i\})$, i.e., $G_0$ is different from $G$ by only an isolated node $v^*$, and $G_i$ is exactly the same as $G'$. Moreover, considering the projected graphs, $G_0^{\tD}$ is the same as $G^{\tD}$ except for an isolated node $v^*$, since the isolated $v^*$ in $G_0$ does not influence the projection algorithm; obviously, $G_i^{\tD}$ is the same as $G'^{\tD}$.

Given any graph $H$, let $\f{H} = \highdegree{\tau}{H} - \ind{\degree{H}{v^*} \geq \tau}$, i.e., the number of high degree node in $H$ not counting $v^*$. We let $\ind{\degree{H}{v^*} \geq \tau}$ be $0$ if $v^*$ is not in $H$. 
Obviously, $\f{G_0^{\tD}} = \f{G^{\tD}}$ (since an isolated node $v^*$ does not have high degree) and $\f{G_i^{\tD}} = \f{G'^{\tD}}$.
So we have
\begin{align*}
&	\abs{\f{G'^{\tD}} - \f{G^{\tD}}}
= \abs{\f{G_i^{\tD}} - \f{G_0^{\tD}}}\\
= 	&|\f{G_i^{\tD}} - \f{G_{i-1}^{\tD}} + \f{G_{i-1}^{\tD}} - \dots \\& -  \f{G_1^{\tD}} + \f{G_1^{\tD}} - \f{G_0^{\tD}}|\\
\leq &\sum_{j=1}^i \abs{\f{G_j^{\tD}} - \f{G_{j-1}^{\tD}}}.
\end{align*}

Now we aim at bounding $\abs{\f{G_j^{\tD}} - \f{G_{j-1}^{\tD}}}$ for any $j \in [i]$.
Notice that $G_j$ and $G_{j-1}$ only differ by one edge $e^*_j$. Let us consider the influence of $e^*_j$ in the projected graph. 
Suppose $e^* = (v^*,v_0)$. 
Then there exists a finite ``alternating'' edge sequences $(v^*, v_1)$, $(v_1, v_2)$, $(v_2, v_3)$, \dots, $(v_{\ell-1}, v_{\ell})$ for some $\ell \geq 1$ which satisfies the following conditions:
\begin{itemize}
\item All edges in the sequence are in $E'$ and $\lambda((v^*,v_1)) < \lambda((v_1,v_2)) < \dots < \lambda((v_{\ell-1}, v_{\ell}))$
\item For all $k \in [\ell-1] \cap 2\mathbb{Z}+1$ (edges at odd positions), $\degree{G_{j-1}^{\tD}}{v_k} = \tD$ and $(v_k, v_{k+1}) = \argmax_{e \in F} \lambda(e)$ for $F = \edge{G_{j-1}^{\tD}}{v_k}$ ($(v_k, v_{k+1})$ has the lowest priority among all adjacent edges of $v_k$ in $G_{j-1}^{\tD}$)
\item For all $k \in [\ell-2] \cap 2\mathbb{Z}$ (edges at even positions), $\degree{G_{j-1}^{\tD}}{v_k} < \degree{G_{j-1}}{v_k}$ and $(v_k, v_{k+1}) = \argmin_{e \in F} \lambda(e)$ for $F = \edge{G_{j-1}}{v_k} \backslash \edge{\tilde{G_{j-1}}}{v_k}$ ($(v_k, v_{k+1})$ has the highest priority among all adjacent edges of $v_k$ in $G_{j-1}$ but not in $G_{j-1}^{\tD}$)
\end{itemize}

Running the projection algorithm, we would have the following process.
$(v^*,v_1)$ is added to $G_j^{\tD}$. Since $\degree{G_{j-1}^{\tD}}{v_1} = \tD$, $(v^*,v_1)$ will prevent another edge,
$(v_1,v_2)$, from being added to $G_j^{\tD}$. This saves one quota in the budget of $v_2$. Since some edge in $\edge{G_{j-1}}{v_2}$ is not present in $G_{j-1}^{\tD}$, we will have the first of them,
$(v_2,v_3)$, added to $G_j^{\tD}$
\dots
This process keeps going and stops at $(v_{\ell-1},v_{\ell})$. 
In a word, all odd positioned edges in the sequence are added to $G_j^{\tD}$ and all even positioned edges are not. Any edge that is not in the sequence has the same status in $G_{j-1}^{\tD}$ and $G_j^{\tD}$. 

Considering the degrees of all nodes in $\{v_1,\dots,v_{\ell}\}$, only $v_{\ell}$ has different degrees in $G_{j-1}^{\tD}$ and $G_j^{\tD}$, and the difference is either $1$ or $-1$; all other nodes in $\{v_1,\dots,v_{\ell}\}$ have the same degree in both graphs. Also, all nodes in $V\backslash \{v_1,\dots,v_{\ell}\}$ has the same degree in both graphs as well.

Therefore, $\abs{\f{G_j^{\tD}} - \f{G_{j-1}^{\tD}}} \leq 1$, and thus $\abs{\f{G'^{\tD}} - \f{G^{\tD}}} \leq i \leq \tD$. Taking into consideration that $v^*$ can have high degree, we have $\abs{\highdegree{\tau}{G'^{\tD}} - \highdegree{\tau}{G^{\tD}}} \leq \tD + 1$.

We now show that for any $\tD$ and $\tau$ such that $0 < \tau \leq \tD$, there exists neighboring graphs $G$ and $G'$ such that $\abs{\highdegree{\tau}{G'^{\tD}} - \highdegree{\tau}{G^{\tD}}} = \tD + 1$.
Let $G$ be a graph with $\tD$ $(\tau-1)$-stars, and $G'$ be the same as $G$ except for an additional node $v^*$ that is connected to all the centers of the stars. Since $\tau \leq \tD$, all nodes in $G$ and $G'$ have degree no more than $\tD$, and thus $G^{\tD} = G$ and $G'^{\tD} = G'$. We have $\highdegree{\tau}{G} = 0$ and $\highdegree{\tau}{G'} = \tD+1$.

Therefore, the global sensitivity is $\tD+1$.
\end{proof}

\begin{lemma}\label{lem:baseline_directed_one_DLL}
Using Algorithm~\ref{alg:DLL_oneshot} on a directed graph $G$ to get $G^{\tDin, \tDout}$ and publishing $\highdegree{\tau}{G^{\tDin, \tDout}}$ has global sensitivity $\max\{\tDin + 1, \tDout-1\}$ for any $0 < \tau \leq \tDout$.
\end{lemma}

\begin{proof}
For two graphs $H_1$ and $H_2$ with the same set of nodes, we say an edge has the same status in graph $H_1$ and $H_2$ if it is present in both graphs or absent in both.

Consider neighboring graphs $G = (V, E)$ and $G' = (V', E')$, with $V' = V \cup \{v^*\}$ and $E' = E \cup E^*_i \cup E^*_o$, where $E^*_i$ consists of all in-edges adjacent to $v^*$ and $^*E_o$ consists of all out-edges adjacent to $v^*$.
Let $\Lambda$ and $\lambda$ be the edge order and the corresponding ordering function.
To simplify the notation, we use $\tilde{G}$ to denote $G^{\tDin, \tDout}$ for any $G$.
Let the projected graphs be $\tilde{G}$ and $\tilde{G'}$.

In $G'$, $v^*$ may have a large number of adjacent edges $E^*$, but it is adjacent to at most $\tDin$ in-edges and at most $\tDout$ out-edges in $\tilde{G'}$. Let such edges be $\{e^*_1,\dots,e^*_i\}$ with $\lambda(e^*_j) < \lambda(e^*_{j+1})$ for any $j$. Obviously, this set contains at most $\tDin$ in-edges and at most $\tDout$ out-edges of $v^*$.
Notice that the projected graph of $(V', E\cup E^*)$ and that of $(V', E\cup \{e^*_1,\dots,e^*_i\})$ are the same, therefore we can assume that $E^*_i$ contains only the in-edges of $v^*$ that are present in $\tilde{G'}$ and $E^*_o$ contains only the out-edges of $v^*$ that are present in $\tilde{G'}$ and have $E^*_i \cup E^*_o = \{e^*_1,\dots,e^*_i\}$. This assumption does not change $\tilde{G'}$.

Let $G_0 = (V', E)$, $G_1 = (V', E\cup \{e^*_1\})$, $G_2 = (V', E\cup \{e^*_1, e^*_2\})$, \dots, $G_t = (V', E\cup E^*)$, i.e., $G_0$ is different from $G$ by only an isolated node $v^*$, and $G_i$ is exactly the same as $G'$. 
Moreover, considering the projected graphs, $\tilde{G_0}$ is the same as $\tilde{G}$ except for an isolated node $v^*$, since the isolated $v^*$ in $G_0$ does not influence the projection algorithm; obviously, $\tilde{G_i}$ is the same as $\tilde{G'}$.

Given any graph $H$, let $\f{H} = \highoutdegree{\tau}{H} - \ind{\outdegree{H}{v^*} \geq \tau}$, i.e., the number of high out-degree nodes in $H$ not counting $v^*$. We let $\ind{\outdegree{H}{v^*} \geq \tau}$ be $0$ if $v^*$ is not in $H$. 
Obviously, $\f{\tilde{G_0}} = \f{\tilde{G}}$ (since an isolated node $v^*$ does not have high out-degree) and $\f{\tilde{G_t}} = \f{\tilde{G'}}$.
So we have
\begin{align*}
&	{\f{\tilde{G'}} - \f{\tilde{G}}}
=   {\f{\tilde{G_{i}}} - \f{\tilde{G_{0}}}}\\
= & \f{\tilde{G_{i}}} - \f{\tilde{G_{i-1}}} + \f{\tilde{G_{i-1}}} - \dots \\&- \f{\tilde{G_{1}}} + \f{\tilde{G_{1}}} - \f{\tilde{G_{0}}}\\
= & \paren{\f{\tilde{G_{i}}} - \f{\tilde{G_{i-1}}}} + \dots + \paren{\f{\tilde{G_{1}}} - \f{\tilde{G_{0}}}}.
\end{align*}

Now we aim at calculating $\f{\tilde{G_{j}}} - \f{\tilde{G_{j-1}}}$.
Notice that $G_j$ and $G_{j-1}$ only differ by one edge $e^*_j$.
Let us consider the influence of $e^*_j$ in the projected graph. We need to consider two cases -- $e^*_j$ is an out-edge of $v^*$ and $e^*_j$ is an in-edge.

First, suppose $e^*_j = (v^*,v_1)$ is an out-edge of $v^*$.

There exists a finite ``alternating'' edges sequences $e^*_j = (v^*, v_1)$, $e_2 = (v_2, v_1)$, $e_3 = (v_2, v_3)$, $e_4 = (v_4,v_3)$ \dots, $e_{\ell} = (v_{\ell-1}, v_{\ell})$ (or the last one might be $e_{\ell} = (v_{\ell},v_{\ell-1})$) for some $\ell \geq 1$ which satisfies the following conditions:
\begin{itemize}
\item All edges in the sequence are in $V'$ and $\lambda(e^*_j) < \lambda(e_2) < \dots < \lambda(e_{\ell})$
\item For all $k \in [\ell-1] \cap 2\mathbb{Z}+1$, $\indegree{\tilde{G_{j-1}}}{v_k} = \tDin$ and $e_{k+1} = (v_{k+1},v_k) = \argmax_{e \in E'} \lambda(e)$ for $E' = \inedge{\tilde{G_{j-1}}}{v_k}$ ($e_{k+1}$ has the lowest priority among all adjacent in-edges of $v_k$ that are in $\tilde{G_{j-1}}$)
\item For all $k \in [\ell-2] \cap 2\mathbb{Z}$,   $\outdegree{\tilde{G_{j-1}}}{v_k} < \outdegree{G_{j-1}}{v_k}$ and $e_{k+1} = (v_k, v_{k+1}) = \argmin_{e \in E'} \lambda(e)$ for $E' = \outedge{G_{j-1}}{v_k} \backslash \outedge{\tilde{G_{j-1}}}{v_k}$ ($e_{k+1}$ has the highest priority among all adjacent out-edges of $v_k$ that are not in $\tilde{G_{j-1}}$)
\end{itemize}

Running the projection algorithm, we would have the following process.
$(v^*,v_1)$ is added to $\tilde{G_j}$. Since $\indegree{\tilde{G_{j-1}}}{v_1} = \tDin$, $(v^*,v_1)$ will prevent another edge,
$(v_2,v_1)$, from being added to $\tilde{G_j}$. This saves one quota in the out-degree budget of $v_2$. Since some edge in $\outedge{G_{j-1}}{v_2}$ is not present in $\tilde{G_{j-1}}$, we will have the first of them,
$(v_3,v_2)$, added to $\tilde{G_j}$
\dots
This process keeps going and stops at $e_{\ell}$. 
In a word, all odd positioned edges in the sequence are added to $\tilde{G_j}$ and all even positioned edges are not. Any edge that is not in the sequence has the same status in $\tilde{G_{j-1}}$ and $\tilde{G_j}$. 

Considering the out-degrees of all nodes in $\{v_1,\dots,v_{\ell}\}$, if $\ell$ is even, only $v_{\ell}$ has different out-degrees in $\tilde{G_{j-1}}$ and $\tilde{G_j}$ and $\outdegree{\tilde{G_{j-1}}}{v_{\ell}} - \outdegree{\tilde{G_j}}{v_{\ell}} = 1$ and all other nodes have the same degree in the two graphs; and if $\ell$ is odd, then all nodes have the same out-degrees in $\tilde{G_{j-1}}$ and $\tilde{G_j}$.
All nodes in $V\backslash \{v_1,\dots,v_{\ell}\}$ have the same out-degree in $\tilde{G_{j-1}}$ and $\tilde{G_j}$. 

Second, suppose $e^*_j = (v_1, v^*)$ is an in-edge of $v^*$.

There exists a finite ``alternating'' edges sequences $e^*_j = (v_1,v^*)$, $e_2 = (v_1,v_2)$, $e_3 = (v_3,v_2)$, $e_4 = (v_3,v_4)$ \dots, $e_{\ell} = (v_{\ell-1}, v_{\ell})$ (or the last one might be $e_{\ell} = (v_{\ell},v_{\ell-1})$) for some $\ell \geq 1$ which satisfies the following conditions:
\begin{itemize}
\item All edges in the sequence are in $V'$ and $\lambda(e^*_j) < \lambda(e_2) < \dots < \lambda(e_{\ell})$
\item For all $k \in [\ell-1] \cap 2\mathbb{Z}+1$, $\outdegree{\tilde{G_{j-1}}}{v_k} = \tDout$ and $e_{k+1} = (v_{k+1},v_k) = \argmax_{e \in E'} \lambda(e)$ for $E' = \outedge{\tilde{G_{j-1}}}{v_k}$ ($e_{k+1}$ has the lowest priority among all adjacent out-edges of $v_k$ that are in $\tilde{G_{j-1}}$)
\item For all $k \in [\ell-2] \cap 2\mathbb{Z}$,   $\indegree{\tilde{G_{j-1}}}{v_k} < \indegree{G_{j-1}}{v_k}$ and $e_{k+1} = (v_k, v_{k+1}) = \argmin_{e \in E'} \lambda(e)$ for $E' = \inedge{G_{j-1}}{v_k} \backslash \inedge{\tilde{G_{j-1}}}{v_k}$ ($e_{k+1}$ has the highest priority among all adjacent in-edges of $v_k$ that are not in $\tilde{G_{j-1}}$)
\end{itemize}

In the projection algorithm, all odd positioned edges in the sequence are added to $\tilde{G_j}$ and all even positioned edges are not. Any edge that is not in the sequence has the same status in $\tilde{G_{j-1}}$ and $\tilde{G_j}$. 

Considering the out-degrees of all nodes in $\{v_1,\dots,v_{\ell}\}$, if $\ell$ is odd, only $v_{\ell}$ has different out-degrees in $\tilde{G_{j-1}}$ and $\tilde{G_j}$ and $\outdegree{\tilde{G_j}}{v_{\ell}} - \outdegree{\tilde{G_{j-1}}}{v_{\ell}} = 1$, and all other nodes have the same degree in the two graphs; and if $\ell$ is even, then all nodes have the same out-degrees in $\tilde{G_{j-1}}$ and $\tilde{G_j}$.
All nodes in $V\backslash \{v_1,\dots,v_{\ell}\}$ has the same out-degree in $\tilde{G_{j-1}}$ and $\tilde{G_j}$. 

Therefore, 
if $e^*_j$ is an out-edge of $v^*$, then $\f{\tilde{G_j}} - \f{\tilde{G_{j-1}}}$ is either $-1$ or $0$;
if $e^*_j$ an in-edge of $v^*$, then $\f{\tilde{G_j}} - \f{\tilde{G_{j-1}}}$ is either $1$ or $0$.

We have
\begin{align*}
&\paren{\f{\tilde{G_{i}}} - \f{\tilde{G_{i-1}}}} + \dots + \paren{\f{\tilde{G_{1}}} - \f{\tilde{G_{0}}}}
\\&\in [- \outdegree{\tilde{G'}}{v^*}, \indegree{\tilde{G'}}{v^*}],
\end{align*}
and since
\begin{align*}
&	\highoutdegree{\tau}{\tilde{G'}} - \highoutdegree{\tau}{\tilde{G}}\\
= &	\f{\tilde{G'}} - \f{\tilde{G}} + \ind{\outdegree{\tilde{G'}}{v^*} \geq \tau},
\end{align*}
we have
\begin{align*}
&	\highoutdegree{\tau}{\tilde{G'}} - \highoutdegree{\tau}{\tilde{G}} \\
\leq&
\indegree{\tilde{G'}}{v^*} + \ind{\outdegree{\tilde{G'}}{v^*} \geq \tau}
\leq
\tDin + 1,\\
&\highoutdegree{\tau}{\tilde{G}} - \highoutdegree{\tau}{\tilde{G'}}\\
\leq&
\outdegree{\tilde{G'}}{v^*} - \ind{\outdegree{\tilde{G'}}{v^*} \geq \tau}
\leq \tDout - 1.
\end{align*}

Therefore the global sensitivity is upper bounded by $\max\{\tDin + 1, \tDout-1\}$.

We now show that for any $\tDin,\tDout > 0$, there exists neighboring graphs $G$ and $G'$ such that the difference between $\highoutdegree{\tau}{\tilde{G'}}$ and $\highoutdegree{\tau}{\tilde{G}}$ is $\tDin + 1$ and another neighboring graphs $G$ and $G'$ such that the difference is $\tDout - 1$.

Let $G$ be a graph with $\tDin$ $(\tau-1)$-out-stars, and $\tau$ isolated nodes. Let $G'$ be the same as $G$ except for an additional node $v^*$ that has in-edges from all the centers of the stars and out-edges to all the isolated nodes. All nodes in $G$ and $G'$ have out-degree no more than $\tDout$ and in-degree no more than $\tDin$, and thus $\tilde{G} = G$ and $\tilde{G'} = G'$. We have $\highoutdegree{\tau}{G} = 0$ and $\highoutdegree{\tau}{G'} = \tDin + 1$.

Let $G$ be a graph with $\tDout$ $\tau$-out-stars, with the nodes of the $i$-th star labeled as $\{c^i, b^i_1, \dots, b^i_{\tau}\}$, where $c^i$ is the center node, and $b^i_j$ denote the $j$-th outer node of the $i$-th star. Additionally, for every $i\in[\tau]$, let there be a set of $\Din - 1$ nodes $\{a^i_1,\dots,a^i_{\tDin-1}\}$ that point to $b^i_1$. Suppose all edges in the stars have the lowest priority in $\Lambda$.
Let $G'$ be the same as $G$ except for an additional node $v^*$ with additional edges $\{(v^*, b^i_1):i\in[\tau]\}$, i.e., $v^*$ points to the first non-central node of all stars. Let this set of new edges take highest priority in the ordering $\Lambda$.
In $G$, we have $\outdegree{}{c^i} = \tau$, $\outdegree{}{a^i_j} = 1$, $\outdegree{}{b^i_j} = 0$ for all $i, j$; $\indegree{}{c^i} = 0$, $\indegree{}{a^i_j} = 0$ (for all $j$), $\indegree{}{b^i_1} = \tDin$, $\indegree{}{b^i_j} = 1$ (for $j > 1$) for all $i$. So $\tilde{G} = G$, and $\highoutdegree{\tau}{\tilde{G}} = \tDout$.
On the other hand, running Algorithm~\ref{alg:DLL_oneshot} in $G'$, we would first consider the edges adjacent to $v^*$ and edges adjacent to $a^i_j$ and add all of them to $\tilde{G'}$. Now, since for all $i$, $b^i_1$ has in-degree equal to $\tDin$, we can no longer add edge $(c^i, b^i_1)$; we can still add $(c^i, b^i_j)$ for all $j > 1$ and have our final $\tilde{G'}$.
In $\tilde{G'}$, we have $\outdegree{}{c^i} = \tau - 1$, $\outdegree{}{a^i_j} = 1$, $\outdegree{}{b^i_j} = 0$, $\outdegree{}{v^*} = \tDout$ for all $i, j$, and thus $\highoutdegree{\tau}{\tilde{G'}} = 1$.
So $\highoutdegree{\tau}{\tilde{G}} - \highoutdegree{\tau}{\tilde{G'}} = \tDout - 1$.

Therefore, the global sensitivity is $\max\{\tDin + 1, \tDout - 1\}$.

We note that it is a must to put bound on both in-degree and out-degree. Apparently, we must have a bound on the in-degree (otherwise $v^*$ can have many in-edges and each can make the out-degree of one node cross $\tau$). Yet if we only bound in-degree, then $v^*$ can have many out-edges. Each out-edge $(v^*, v_0)$ can saturate the in-degree of $v_0$ earlier, forcing it to discard another of its in-edge $(v_1, v_0)$, causing the out-degree of $v_1$ to decrease by $1$. So the sensitivity can be at least the out-degree of $v^*$.
\end{proof}

\subsubsection{Analyses for \sensseqD}
\begin{lemma}\label{lem:baseline_undirected_T_D}
Given an undirected graph $G$, the $D$-bounded global sensitivity for publishing $\seq{\highdegree{\tau}{G_i})}{i=1}{T}$ is $(D+1)T$.
\end{lemma}
\begin{proof}
Consider $G = (V, E)$ and $G = (V', E')$ such that $V = \cup_{t=1}^{\infty}\partial{V_{t}}$, $V' = \cup_{t=1}^{\infty}\partial{V'_{t}}$, $\partial{V'_{i}} = \partial{V_{i}} \cup \{v^*\}$ and $\partial{V_{j}} = \partial{V'_{j}}$ for $j\neq i$.
For any $v \in V$ that is not connected with $v^*$, $\degree{G_i}{v}$ and $\degree{G'_i}{v}$ are the same for any $i$. %
There are at most $D$ nodes that are connected with $v^*$; let $V^* = \{v\in V:(v, v^*) \in E'\}$ be the set of all such nodes.
We have
\begin{align*}
&	\sum_{i=1}^T \abs{\highdegree{\tau}{G_i} - \highdegree{\tau}{G'_i}}\\
=&	\sum_{i=1}^T \abs{\sum_{v\in V} \ind{\degree{G_i}{v} \geq \tau} - \sum_{v\in V'} \ind{\degree{G'_i}{v} \geq \tau}}\\
=&	\sum_{i=1}^T \abs*{\sum_{v\in V} \left(\ind{\degree{G_i}{v} \geq \tau} - \ind{\degree{G'_i}{v} \geq \tau}\right) \\&- \ind{\degree{G'_i}{v^*} \geq \tau}}\\
=&	\sum_{i=1}^T \abs*{\sum_{v\in V^*} \left(\ind{\degree{G_i}{v} \geq \tau} - \ind{\degree{G'_i}{v} \geq \tau}\right) \\&- \ind{\degree{G'_i}{v^*} \geq \tau}}\\
\leq&\sum_{i=1}^T (|V^*| + 1)
\leq (D+1)T,
\end{align*}
where the third equality follows from the fact that $\degree{G'_i}{v} = \degree{G_i}{v}$ for $v \notin V^*$, and the last inequality follows from the fact that the size of $V^*$, which consists of all nodes adjacent to $v^*$, is at most $D$.

Now we show that there exists neighboring $G$ and $G'$ such that the $D$-bounded global sensitivity is $(D+1)T$.
Lemma~\ref{lem:baseline_undirected_one_D} shows that there exists neighboring $D$-bounded graphs $G$ and $G'$, such that $\abs{\highdegree{}{G'} - \highdegree{}{G}} = D + 1$. Let $G_1 = G$, $G'_1 = G'$. We then have $\highdegree{}{G_i} = \highdegree{}{G_1} = \highdegree{}{G}$ and the same for $G'$. Therefore, $\sum_{i=1}^T \abs{\highdegree{}{G_i} - \highdegree{}{G'_i}} = (D+1)T$.
\end{proof}

\begin{lemma}\label{lem:baseline_directed_T_D}
Given a directed graph $G$, the $\Din$-bounded global sensitivity for publishing $\seq{\highoutdegree{\tau}{G_i}}{i=1}{T}$ is $(\Din+1)T$.
\end{lemma}
\begin{proof}
Consider $G = (V, E)$ and $G = (V', E')$ such that $V = \cup_{t=1}^{\infty}\partial{V_{t}}$, $V' = \cup_{t=1}^{\infty}\partial{V'_{t}}$, $\partial{V'_{i}} = \partial{V_{i}} \cup \{v^*\}$ and $\partial{V_{j}} = \partial{V'_{j}}$ for $j\neq i$.
For any $v \in V$ that does not point to $v^*$, $\outdegree{G_i}{v}$ and $\outdegree{G'_i}{v}$ are the same for any $i$.
There are at most $\Din$ nodes that point to $v^*$; let $V^* = \{v\in V:(v, v^*) \in E'\}$ denote all such nodes.
We have
\begin{align*}
&	\sum_{i=1}^T \abs{\highoutdegree{\tau}{G_i} - \highoutdegree{\tau}{G'_i}}\\
=&	\sum_{i=1}^T \abs*{\sum_{v\in V} \ind{\outdegree{G_i}{v} \geq \tau} \\&- \sum_{v\in V'} \ind{\outdegree{G'_i}{v} \geq \tau}}\\
=&	\sum_{i=1}^T \abs*{\sum_{v\in V} \left(\ind{\outdegree{G_i}{v} \geq \tau} \right.\\&\left.- \ind{\outdegree{G'_i}{v} \geq \tau}\right) - \ind{\outdegree{G'_i}{v^*} \geq \tau}}\\
=&	\sum_{i=1}^T \abs*{\sum_{v\in V^*} \left(\ind{\outdegree{G_i}{v} \geq \tau} \right.\\&\left.- \ind{\outdegree{G'_i}{v} \geq \tau}\right) - \ind{\outdegree{G'_i}{v^*} \geq \tau}}\\
\leq&\sum_{i=1}^T (|V^*| + 1)
\leq (\Din+1)T,
\end{align*}
where the third equality follows from the fact that $\outdegree{G'_i}{v} = \outdegree{G_i}{v}$ for $v \notin V^*$, and the last inequality follows from the fact that the size of $V^*$, which consists of all nodes that point to $v^*$, is at most $\Din$.

Now we show that there exists neighboring $G$ and $G'$ such that the $\Din$-in-bounded global sensitivity is $(\Din+1)T$.
Lemma~\ref{lem:baseline_directed_one_D} shows that there exists neighboring $\Din$-bounded graphs $G$ and $G'$, such that $\abs{\highoutdegree{}{G'} - \highoutdegree{}{G}} = \Din + 1$. Let $G_1 = G$, $G'_1 = G'$. We then have $\highoutdegree{}{G_i} = \highoutdegree{}{G_1} = \highoutdegree{}{G}$ and the same for $G'$. Therefore, $\sum_{i=1}^T \abs{\highoutdegree{}{G_i} - \highoutdegree{}{G'_i}} = (\Din+1)T$.
\end{proof}

\subsubsection{Analyses for \sensseqproj}
\begin{lemma}\label{lem:baseline_T_DLL}
Using Algorithm~\ref{alg:DLL_seq} on an undirected (or directed) graph sequence $\seq{G_i}{i=1}{T}$ to get $\seq{G_i^{\tD}}{i=1}{T}$ (or $\seq{G_i^{\tDin,\tDout}}{i=1}{T}$) and publishing $\seq{\highdegree{\tau}{G_i^{\tD}}}{i=1}{T}$ (or $\seq{\highoutdegree{\tau}{G_i^{\tDin, \tDout}}}{i=1}{T}$) has global sensitivity at least $(D+1)T$ (or at least $\max\{\tDin+1, \tDout-1\} \cdot T$).
\end{lemma}
\begin{proof}
Let $\tilde{G}$ denote $G^{\tD}$ or $G^{\tDin,\tDout}$ for any $G$.

From Lemma~\ref{lem:baseline_undirected_one_DLL} (or Lemma~\ref{lem:baseline_directed_one_DLL}), there exists graphs $G=(V,E)$ and $G'=(V',E')$, such that $\abs{\highdegree{\tau}{\tilde{G'}} - \highdegree{\tau}{\tilde{G}}}$ (or the formula with $\highoutdegree{}{}$) equal to $D+1$ (or $\max\{\tDin+1, \tDout-1\}$). Let $G$ and $G'$ be such graphs.
Let $G_1 = G$, $G'_1 = G'$ and $\partial{V_{j}} = \partial{V'_{j}} = \emptyset$ for $j > 1$.
We thus have $\tilde{G_i} = \tilde{G_1} = \tilde{G}$ and $\tilde{G'_i} = \tilde{G'_1} = \tilde{G'}$ for any $i$. Therefore, all elements in the published sequence are the same, with value $\highdegree{\tau}{\tilde{G}}$ and $\highdegree{\tau}{\tilde{G'}}$ (or $\highoutdegree{\tau}{\tilde{G}}$ and $\highoutdegree{\tau}{\tilde{G'}}$).

Therefore,
\begin{align*}
\sum_{i=1}^T \abs{\highdegree{\tau}{\tilde{G_i}} - \highdegree{\tau}{\tilde{G_i}}} = (\tD+1)T
\end{align*}
or
\begin{align*}
&\sum_{i=1}^T \abs{\highoutdegree{\tau}{\tilde{G_i}} - \highoutdegree{\tau}{\tilde{G_i}}}\\ =& \max\{\tDin+1, \tDout-1\} \cdot T,
\end{align*}
and we can conclude that the global sensitivity of publishing the whole sequence is at least $(\tD+1)T$ (or $\max\{\tDin+1, \tDout-1\} \cdot T$).
\end{proof}

\subsection{Baselines for Publishing the Number of Edges}
Similar to the algorithms for publishing the number of high-degree node, we have four baselines for publishing the number of edges. 

As has been mentioned in Section~\ref{sec:main_alg}, there are two baseline algorithms for high-degree nodes count. Combining with the algorithm proposed in \cite{day2016publishing}, there are four baseline algorithms for both undirected and directed graph sequences.

For an undirected graph sequence $(G_1,\dots,G_T)$,
\begin{enumerate}
\item \composeD: For each $i \in [T]$, run the global sensitivity algorithm to publish $\numedge{}{G_i}$ with privacy parameter $\epsilon / T$ and $D$-bounded global sensitivity $D$ (Lemma~\ref{lem:baseline_one_D_edge}).
\item \composeproj: For each $i \in [T]$, run Algorithm~\ref{alg:DLL_oneshot} on $G_i$ with projection parameter $\tD$ to get $G_i^{\tD}$. Run the global sensitivity algorithm to publish $\numedge{}{G_i^{\tD}}$ with privacy parameter $\epsilon / T$ and global sensitivity $\tD$ (Lemma~\ref{lem:baseline_one_DLL_edge}).
\item \sensseqD: Run the global sensitivity algorithm to publish $\seq{\numedge{}{G_i}}{i=1}{T}$ with privacy parameter $\epsilon$ and $D$-bounded global sensitivity.
\item \sensseqproj: Run Algorithm~\ref{alg:DLL_seq} on $\seq{G_i}{i=1}{T}$ with projection parameter $\tD$ to get $\seq{G_i^{\tD}}{i=1}{T}$. Run the global sensitivity algorithm to publish $\seq{\numedge{}{G_i^{\tD}}}{i=1}{T}$ with privacy parameter $\epsilon$ and computed global sensitivity.
\end{enumerate}

For a directed graph sequence $(G_1,\dots,G_T)$,
\begin{enumerate}
\item \composeD: For each $i \in [T]$, run the global sensitivity algorithm to publish $\numedge{}{G_i}$ with privacy parameter $\epsilon / T$ and $(\Din,\Dout)$-bounded global sensitivity $\Din+\Dout$ (Lemma~\ref{lem:baseline_one_D_edge}).
\item \composeproj: For each $i \in [T]$, run Algorithm~\ref{alg:DLL_oneshot} on $G_i$ with projection parameter $\tDin$, $\tDout$ to get $G_i^{\tDin,\tDout}$. Run the global sensitivity algorithm to publish $\numedge{}{G_i^{\tDin,\tDout}}$ with privacy parameter $\epsilon / T$ and global sensitivity $\tDin + \tDout$ (Lemma~\ref{lem:baseline_one_DLL_edge}).
\item \sensseqD: Run the global sensitivity algorithm to publish $\seq{\numedge{}{G_i}}{i=1}{T}$ with privacy parameter $\epsilon$ and $\Din$-in-bounded global sensitivity.
\item \sensseqproj: Run Algorithm~\ref{alg:DLL_seq} on $\seq{G_i}{i=1}{T}$ with projection parameter $\tDin$, $\tDout$ to get $\seq{G_i^{\tDin,\tDout}}{i=1}{T}$. Run the global sensitivity algorithm to publish $\seq{\numedge{}{G_i^{\tDin,\tDout}}}{i=1}{T}$ with privacy parameter $\epsilon$ and computed global sensitivity.
\end{enumerate}

Similar to the results for number of high-degree nodes, it is not hard to see that for both undirected and directed graphs, (3) is worse than (1); and (4) is worse than (2). Therefore, we only need to run (1) and (2), i.e., \composeD\ and \composeproj.

\subsubsection{Analyses for \composeD}
\begin{lemma}\label{lem:baseline_one_D_edge}
Given an undirected graph $G = (V, E)$, the $D$-bounded global sensitivity of publishing $\numedge{}{G}$ is $D$.\\
Given an directed graph $G = (V, E)$, the $(\Din,\Dout)$-bounded global sensitivity of publishing $\numedge{}{G}$ is $\Din + \Dout$.
\end{lemma}
\begin{proof}
For undirected graph, suppose we add an additional node $v^*$ to $G = (V, E)$ with $d$ adjacent edges. This does not affect any edge in $E$, but only adds $d \leq D$ edges to the graph. And thus $\numedge{}{G}$ changes by at most $D$. When $d = D$, the change is exactly $D$.

Similar holds for directed graph. A node $v^*$ with $\Din$-in-edges and $\Dout$-out-edges can increase the total number of edges by $\Din + \Dout$.
\end{proof}

\subsubsection{Analyses for \composeproj}
\begin{lemma}\label{lem:baseline_one_DLL_edge}
Using Algorithm~\ref{alg:DLL_oneshot} on an undirected graph $G$ to get $G^{\tD}$ and publishing $\numedge{}{G^{\tD}}$ has global sensitivity $\tD$.\\
Using Algorithm~\ref{alg:DLL_oneshot} on an directed graph $G$ to get $G^{\tDin,\tDout}$ and publishing $\numedge{}{G^{\tDin,\tDout}}$ has global sensitivity $\tDin + \tDout$.
\end{lemma}
\begin{proof}
We follow the same analysis as in Lemma~\ref{lem:baseline_undirected_one_DLL} and \ref{lem:baseline_directed_one_DLL}. 

For undirected graph $G$, for every additional edge $(v^*, v_1)$ adjacent to $v^*$, we know from the proof of Lemma~\ref{lem:baseline_undirected_one_DLL} that it yields an ``alternating sequence'' of edges such that all odd positioned edges are added while all even positioned edges are deleted from the projected graph without $(v^*, v_1)$. So the number of edges can only increase by $1$ or decrease by $1$. Since there are at most $\tD$ edges adjacent to $v^*$ that are added in the projected graph, the number of edges changes by at most $\tD$. 

For directed graph, the same analysis holds. For every in-edge $(v_1,v^*)$ or every out-edge $(v^*, v_1)$, there is an ``alternating sequence'' of edges such that that all odd positioned edges are added while all even positioned edges are deleted from the projected graph. So the number of edges changes by at most $\tDin + \tDout$. 

For both undirected and directed graph, it is easy to see that the number of edges change by exactly $\tD$ and $\tDin + \tDout$ if $G$ does not contain any edge.
\end{proof}

\section{Other}
\begin{lemma}
Suppose we use \sensdiff\ to publish $\{ \sum_{s=1}^{t} \tilde{\Delta}_s \}_{t=1}^T$.
For any $t \leq T$, the standard deviation of $\sum_{s=1}^{t} \tilde{\Delta}_s - f(G_t)$ is of order $\bigO{\sqrt{t}}$ .
\end{lemma}
\begin{proof}
$\sum_{s=1}^{t} \tilde{\Delta}_s - f(G_t) = \sum_{s=1}^t \Lap{\frac{\GS{D}{\Delta}}{\epsilon}}$. The sum of these $t$ i.i.d. Laplace random variables follows the same distribution as the difference between two Gamma random variables with scale parameters $\frac{\GS{D}{\Delta}}{\epsilon}$ and shape parameters $t$. 
Therefore, to analyze the variance of the difference, we only need to analyze the variance of $X - Y$, where $X, Y \sim \Gamma(t, a)$ with $a = \frac{\GS{D}{\Delta}}{\epsilon}$. 
Since $X$ and $Y$ are independent, we have $\Var{X - Y} = \Var{X} + \Var{Y} = 2t a^2$, and thus the standard deviation of $\sum_{s=1}^{t} \tilde{\Delta}_s - f(G_t)$ is $\sqrt{2}\frac{\GS{D}{\Delta}}{\epsilon}\sqrt{t}$.
\end{proof}

On the other hand, it is easy to see that \composeD\ adds noise with standard deviation $\sqrt{2} \frac{\GS{D}{f}}{\epsilon/T} = \sqrt{2} \frac{\GS{D}{f}}{\epsilon} T$.

\end{document}